\documentclass[11pt]{article}
\usepackage{fullpage,amsfonts,amsmath,amsthm,amssymb,mathtools,mathrsfs,graphicx,upgreek}

\usepackage[left=2.54cm,top=2.54cm,right=2.54cm,bottom=2.54cm]{geometry}
\usepackage{epsfig}
\usepackage{bm}

\usepackage{hyperref}

\numberwithin{equation}{section}

\newtheorem{theorem}{Theorem}[section]
\newtheorem{proposition}[theorem]{Proposition}
\newtheorem{corollary}[theorem]{Corollary}
\newtheorem{lemma}[theorem]{Lemma}
\newtheorem{example}[theorem]{Example}

\newtheorem*{remark}{Remark}

\usepackage{bbm}
\usepackage{algorithm}
\usepackage{algpseudocode}
\usepackage{datetime}

\newcommand{\st}{\text{st}}
\newcommand{\scr}{\mathcal}
\newcommand{\mb}{\mathbb}
\newcommand{\til}{\widetilde}

\newcommand{\val}{\text{val}}
\newcommand{\OPT}{\text{OPT}}
\newcommand{\LPOPT}{\text{LPOPT}}

\newcommand{\Ber}{\textup{Ber}}

\usepackage{titlesec}
\usepackage{sectsty}

\sectionfont{\Large}
\subsectionfont{\large}
\subsubsectionfont{\large}

\usepackage[usenames,dvipsnames]{xcolor}

\allowdisplaybreaks
\begin{document}

\title{Bipartite Stochastic Matching: Online, Random Order, and I.I.D.  Models}

\author{Allan Borodin
\thanks{Department of Computer Science, University of Toronto, Toronto, ON, Canada
\texttt{bor@cs.toronto.edu}}
\and
Calum MacRury
\thanks{Department of Computer Science, University of Toronto, Toronto, ON, Canada
\texttt{cmacrury@cs.toronto.edu}}
\and
Akash Rakheja
\thanks{Department of Computer Science, University of Toronto, Toronto, ON, Canada
\texttt{rakhejaakash@gmail.com}}
}

\date{}
\maketitle

\begin{abstract}
 
Within the context of stochastic probing with commitment, we consider
the online stochastic matching problem; that is, the one sided online bipartite
matching problem where edges adjacent to an online node must be probed to
determine if they exist, based on known edge probabilities. If a probed edge
exists, it must be used in the matching (if possible). We study this problem in the generality of a patience (or timeout) constraint which limits the number of probes that can be made to edges adjacent to an online node. Arbitrary  patience constraints result in modelling and computational efficiency issues that are not encountered in the special cases of 
unit patience and full (i.e., unlimited) patience. The stochastic matching problem leads to a variety of settings. Our main contribution is to provide a new LP relaxation and a unified approach for  establishing  new and improved competitive bounds in three  different input model settings (namely, adversarial, random order, and known i.i.d.).     
In
all these settings, the algorithm does not have any control on the ordering of
the online nodes. 
We establish competitive bounds
in these settings, all of which generalize the standard non-stochastic setting when edges do not need to be probed (i.e., exist with certainty). 
All of our results hold for arbitrary edge probabilities and patience constraints. Specifically, we establish the following competitive ratio results:
\\
\begin{enumerate}
\item A $1-1/e$ ratio when the stochastic graph is known, offline vertices are weighted and online arrivals are adversarial.
\item A $1-1/e$ ratio when the stochastic graph is known, edges are weighted, and online arrivals are given in random order (i.e., in ROM, the random order model).  
\item A $1-1/e$ ratio when online arrivals are drawn i.i.d. from a known stochastic type graph and edges are weighted.
\item A (tight) $1/e$ ratio when the stochastic graph is unknown, edges are weighted and online arrivals are given in random order. \\
\end{enumerate}

We note that while results for stochastic graphs in the ROM setting generalize the
corresponding results for the classical ROM bipartite matching setting, it is not clear that a result for a known stochastic graph in the ROM setting implies the same result  for the stochastic unknown and known i.i.d. settings. 

In deriving our results, we clarify and expand upon previous offline benchmarks, relative to which one defines an appropriate definition of the competitive ratio.
In particular, we introduce a new LP relaxation which upper
bounds the performance of ``an  ideal benchmark''. 

\end{abstract}

\newpage

\tableofcontents

\newpage

\section{Introduction}

Stochastic probing problems are part of the larger area of decision making under uncertainty and more specifically, stochastic optimization. Unlike more standard forms of stochastic optimization, it is not just that there is some stochastic uncertainty in the set of inputs, stochastic probing problems involve inputs that cannot be determined without probing (at some cost and/or within some constraint). Applications of stochastic probing occur naturally in many settings, such as in matching problems where compatibility cannot be determined without some trial or investigation (for example, in online dating and kidney exchange applications).
There is by now an extensive literature for stochastic matching problems. For space efficiency, we will give an extended overview of related work in Appendix \ref{appendix:extended_related_works}. Research most directly relating to this paper will appear as we proceed.          

The \textit{stochastic matching problem}\footnote{Unfortunately, the term ``stochastic matching'' is also  used to refer to more standard optimization where the input (i.e., edges or vertices) are drawn from some known or unknown distributions but no probing is involved.}  was introduced by Chen et al. \cite{Chen}. In this problem, we are given an adversarially generated {\it stochastic graph} $G = (V, E)$ with a probability $p_e$ associated with each edge $e$ and a patience (or timeout) parameter $\ell_v$ associated with each vertex $v$. An algorithm probes edges in $E$ within the constraint that at most $\ell_v$ edges are probed incident to any particular vertex $v \in V$. The patience constraint can be viewed as a simple budgetary constraint, where each probe has unit cost and the patience constraint is the budget.  When an edge $e$ is probed, it is guaranteed to exist with probability exactly $p_e$. If an edge $(u,v)$ is found to exist, it is added to the matching and then $u$ and $v$ are no longer available.
The goal is to maximize the expected size of a matching constructed in this way. This problem can be generalized to offline vertices or edges having weights and then the objective is to maximize the expected weight of the matching. Notably, in Chen et al., the algorithm knows the entire stochastic graph in advance.   

In addition to generalizing the setting of the results of Chen et al., Bansal et al. \cite{BansalGLMNR12} introduced an i.i.d.  bipartite version of the problem where nodes on one side of the partition arrive online and edges adjacent to that node are then probed. In their model, each online vertex (and its adjacent edges) is drawn independently and identically from a known distribution. That is,  the possible ``type'' of each online node (i.e., the adjacent edge probabilities and edge weights) is known and the input sequence is then determined i.i.d. from this  known distribution, where the type of a node is presented to the algorithm upon arrival. In both the Chen et al. and Bansal et al. models, each offline node has unlimited patience, whereas each online node specifies its patience upon arrival. As in other online bipartite matching problems, the match for an online node must be made before the next online arrival.  In both of these models, if an edge is probed and confirmed to exist, then it must be included in the current matching (if possible).   
This problem is referred to as the {\it online stochastic matching problem}\footnote{The online stochastic matching problem is sometimes meant to imply unit patience but we will mainly be interested in arbitrary patience values.}  (with patience) and also referred to as the {\it stochastic rewards problem}, though we avoid the latter terminology. In various settings, we will study the online stochastic matching problem. More specifically, we will consider online settings where the algorithm knows the adversarially determined stochastic graph, where a stochastic (type) graph and a distribution on the online vertices is known and online nodes are generated i.i.d. from this distribution, and in the random order model (ROM) when the stochastic graph is both known and unknown\footnote{In a related paper, we establish a $1-1/e$ competitive ratio in the setting where the stochastic graph is unknown and the offline vertices are weighted.  This setting is simpler and allows for a completely combinatorial deterministic algorithm. It is interesting to note that in this setting the same algorithm can be derived using the LP based approach of this paper but the combinatorial method is conceptually and computationally simpler \cite{borodin2020greedy}.}. 
Amongst other applications, the online stochastic matching problem notably models online advertising where the probability of an edge can correspond to the probability of a  purchase in online stores or to pay per click revenue in online searching.

We note that these stochastic matching models generalize the  corresponding classical  non-stochastic models where edges adjacent to an online node are known upon arrival and do not need to be probed. It follows that any inapproximation results in the classical setting apply to the corresponding  stochastic setting.

\section{Preliminaries and Techniques}\label{sec:prelim}
\label{sec:prelim}

The online stochastic matching problem generalizes the classical online bipartite setting as follows. For each $e \in E$ in the stochastic (bipartite) graph $G=(U,V,E)$, there is a fraction $0 \leq p_e \leq 1$ associated with $e$ that gives the probability of existence of the edge $e$. More precisely, each edge $e \in E$ is associated with an independent Bernoulli random variable of parameter $p_{e}$, which we denote by $\st(e)$, corresponding to the \textbf{state} of the edge. 
If $\st(e)=1$, then we say that $e$ is \textbf{active}, and otherwise we say that
$e$ is \textbf{inactive}. It will be convenient to hereby assume that $E = U \times V$. In this way,
if we wish to exclude a pair $(u,v) \in U \times V$ from existing as an edge in $G$, then we may set $p_{u,v}=0$,
thus ensuring that $(u,v)$ is always inactive.

When an online node $v \in V$ arrives, the \textbf{online probing algorithm} sees all the adjacent edges and associated probabilities but must perform a \textbf{probing operation} on the edge to reveal/expose its state, $\st(e)$. As in the classical problem, an online algorithm must decide on a possible match for an online node $v$ before seeing the next online node. The algorithm can be \textbf{non-greedy} and not match a given $v \in V$ even though some $u \in U$ is still unmatched. The online stochastic matching problem simplifies to the classical setting in one of two ways: (1) if $p_e \in \{0,1\}$ for all edges $e$, or (2) if the algorithm is allowed to probe all edges adjacent to an online node $v$ before determining which, if any, node to match to $v$. To make stochastic probing problems meaningful, we either must have a cost for probing or some kind of 
\textbf{commitment} upon probing an input item. Specifically, if an edge $e = (u,v)$ is probed in the stochastic matching problem and turns out to be active, then $e$ must be added to the current matching, provided $u$ and $v$ are both currently unmatched.
We say that the online probing algorithm \textbf{respects commitment} or is \textbf{commital}, provided it satisfies this property.
Furthermore, in the stochastic matching problem, for each online node $v$, there is a known \textbf{patience parameter} (also called timeout)  $\ell_v$ that bounds the number of probes that can be made to edges adjacent to $v$.  The classical online bipartite matching problems (unweighted, vertex weighted, or edge weighted) for  adversarial, ROM, and i.i.d.
input sequences  all generalize to the stochastic  matching setting. We emphasize that the online stochastic matching problem generalizes the classical online problem, even when restricted to the case of \textbf{unit patience} (i.e., $\ell_v = 1$ for all $v \in V$).

Clearly, in the classical adversarial or ROM settings, if the algorithm knew the input graph $G$, the online algorithm could compute an optimal solution before seeing the online sequence and use  that optimal solution to  determine an optimal matching online. But similar to knowing the type graph  in the classical setting with i.i.d. inputs, an algorithm still lacks the ability to know the states of the edges of $G$, namely $(\st(e))_{e \in E}$, so that the stochastic matching problem is interesting, whether the stochastic graph $G$ is known or unknown to the algorithm. We are left then with a wide selection of problems, depending on whether or not the stochastic graph is known, how input sequences are determined, and whether or not edges or vertices
are weighted. In the classical i.i.d.  bipartite matching problem, competitive bounds for an unknown distribution follow from the corresponding ROM problem by a result of Karande et al. \cite{KarandeMT11}. The same result (using the same argument) holds in the stochastic matching setting, provided the stochastic graph is unknown. It is unclear to us whether this reduction continues to
hold for the ROM setting when the stochastic graph is known, as we expand upon in Section \ref{sec:stochastic_known_iid}.   
We will focus on the following settings: a known stochastic graph with adversarial and ROM inputs, the i.i.d. setting with a known stochastic type graph and distribution on the online vertices, and an unknown stochastic graph with ROM arrivals and weighted edges.

What is the benchmark against which  we measure the competitive performance of an online algorithm in the stochastic matching problem?   
In the classical online setting, we compare the value of the online algorithm to that of an optimal matching of the graph. If the inputs are drawn from a distribution, we then compare the expected value of the algorithm to the expected value of an optimum matching. For stochastic probing problems, it is easy to see we cannot hope to obtain a reasonable competitive bound for this type of comparison; that is, 
if we are comparing the expected value of an online probing algorithm to the expected value of an optimum matching of the stochastic graph. 
For example, consider a single online vertex with patience $1$, and $n$ offline (unweighted) vertices where each edge $e$ has probability $\frac{1}{n}$ of being present. The expectation of an online probing algorithm will be at most $\frac{1}{n}$ while the expected size of an optimal matching (over all instantiations of the edge probabilities) will be $1 - (1-\frac{1}{n})^n \rightarrow 1 -\frac{1}{e}$. This example clearly shows that no constant ratio is possible if the patience is sublinear (in $n = |U|$).  

A reasonable approach is to force the benchmark to adhere to the commitment and patience requirements of $G$ that the online algorithm satisfies. Following previous work and the explicit reasoning in Brubach et al.\cite{Brubach2019}, an ideal benchmark is the following: knowing the stochastic graph $G$ (or $G$ and the type graph in the stochastic i.i.d. setting) and the patience requirements of the online nodes, the benchmark can probe edges in any {\it adaptive} order but must satisfy the commitment and patience requirements of the online vertices. By adaptive order, we mean that the next edge to be probed will depend on all the edges that have been currently revealed and the current matching. We emphasize that this benchmark is not restricted to any ordering of the online vertices. In particular, we note that after probing some edge $(u_1,v_1)$, the next probed edge can be $(u_2,v_2)$ where $u_2$ and $v_2$ each may be distinct from $u_{1}$ and $v_{1}$, respectively.  As in online probing algorithms, the goal of the benchmark is to build
a matching whose weight is as large as possible in expectation. We refer to this benchmark as the \textbf{committal benchmark}, and denote the expected value of its matching by $\OPT(G)$.

We also consider a \textit{stronger} benchmark
which still must adaptively probe edges subject to patience constraints, but isn't restricted by commitment;
that is, it may decide upon which subset of edges to match \textit{after} all its probes have been made. Once again,
the probes of this benchmark need not respect any ordering on the online nodes, and the benchmark's goal is
to build a matching of maximum expected weight. We refer to this benchmark as the \textbf{non-committal benchmark}, and denote the expected value of the matching it constructs by $\OPT_{non}(G)$. Observe that in the case
of \textbf{full patience} (i.e., $\ell_{v} = |U|$ for all $v \in V$), the benchmark may probe all the edges of $G$,
and thus corresponds to the expected weight of the optimum matching of the stochastic graph.

Following what is standard in the stochastic matching literature, we prove our results against the committal benchmark,
though in Appendix \ref{appendix:non_committal_benchmark} we prove that our results also hold against the non-committal benchmark
for the case of unit or full patience values.

\subsection{A Review of Our Technical Contributions}

Suppose we are presented a stochastic bipartite graph $G=(U,V,E)$, with edge probabilities $(p_{e})_{e \in E}$,
edge weights $(w_{e})_{e \in E}$ and patience values $(\ell_{v})_{v \in V}$. Here $V$ is the set of online vertices and $U$ is the set of offline vertices. We assume $U$ is known apriori to an online algorithm and the vertices in $V$ arrive online. 
In both the committal and non-committal offline benchmarks, it is not clear how to compute these optimal benchmark values.
As such, one instead resorts to an appropriate LP upper bound on their value.

The most prevalent (standard) LP used in the literature
was introduced by Bansal et al. \cite{BansalGLMNR12}, where each pair $(u,v)$
for $u \in U$ and $v \in V$ has a variable $x_{u,v}$, corresponding to the probability
that the committal benchmark probes $(u,v)$. 

\begin{align}\label{LP:standard_definition}
\tag{LP-std}
&\text{maximize} & \sum_{u \in U, v \in V} w_{u,v} \cdot p_{u, v} \cdot x_{u, v} \\
&\text{subject to} & \sum_{v \in V} p_{u, v} \cdot x_{u, v} & \leq 1 && \forall u \in U \label{eqn:standard_offline_matching_constraint}\\
& &\sum_{u \in U} p_{u,v} \cdot x_{u, v} & \leq 1 && \forall v \in V  \\
& &\sum_{u \in U} x_{u, v} & \leq \ell_v && \forall v \in V  \\
& &0 \leq x_{u, v} &\leq 1 && \forall u \in U, v \in V.
\end{align}
Bansal et al. \cite{BansalGLMNR12} observed that if $\LPOPT_{std}(G)$ denotes the value of an optimal solution to \ref{LP:standard_definition}, then it is a \textbf{relaxation} of the committal benchmark \OPT(G); that is,
\begin{equation}\label{eqn:standard_relaxation}
	\OPT(G) \le \LPOPT_{std}(G).
\end{equation}
To date, the most common technique for proving guarantees against the committal benchmark involves comparing the performance
of one's probing algorithm to $\LPOPT_{std}(G)$, as opposed to $\OPT(G)$ directly. In fact, when $G$ is \textit{known} to the
algorithm, one can leverage a solution to \ref{LP:standard_definition} to determine which probes one
should make (see \cite{BansalGLMNR12, Adamczyk15, BrubachSSX16, BrubachSSX20, BavejaBCNSX18}). 

This is especially effective in the case of unit patience, as an optimum solution
to \ref{LP:standard_definition}, say $(x_{u,v})_{u \in U, v \in V}$, induces a distribution for each vertex $v \in V$.
Rather,
\[
	\sum_{u \in U} x_{u,v} \le 1,
\]
for each $v \in V$. Thus, when processing an arriving online node $v$, one can choose to probe
$u \in U$ with probability $x_{u,v}$ (where one \textit{passes} on $v$ with probability $1 - \sum_{u \in U} x_{u,v}$).
While this ignores the issue of the online nodes \textit{colliding} (i.e., multiple vertices of $V$ attempting
a match to $u \in U$), \eqref{eqn:standard_offline_matching_constraint} ensures each vertex of $U$ is matched
at most once \textit{in expectation}. In the case of offline vertex weights, this is sufficient to prove a guarantee
of $1-1/e$ against $\LPOPT_{std}(G)$, no matter how the vertices of $V$ are presented to the probing algorithm.
Similarly, in the case of edge weights, this approach suffices to achieve a guarantee of $1/2$,
albeit requiring the vertices of $V$ to arrive in random order. 
We emphasize that $G$ must be known to the probing algorithm in order to implement these strategies. 
Both these arguments (in the more general setting of arbitrary patience) are discussed in detail in Section \ref{sec:known_stochastic}. 

If one now moves to the case when $v$ has arbitrary patience $\ell_v$, then the values
$(x_{u,v})_{u \in U}$ satisfy the following inequalities:
\begin{equation}
	\sum_{u \in U} x_{u,v} \le \ell_v \text{\; and \;} \sum_{u \in U} p_{u,v} \cdot x_{u,v} \le 1.
\end{equation}
For a fixed $v$, the techniques in \cite{BansalGLMNR12, Adamczyk15, BrubachSSX16, BavejaBCNSX18,BrubachSSX20} involve first drawing a random ordering $\pi$ of $U$ (where different distributions are used in different papers).
Once this is done, a random subset $P \subseteq U$ is drawn using the GKSP algorithm of Gandhi et al. \cite{GandhiGKSP06}, with
the guarantee that $u \in P$ with probability $x_{u,v}$, and that $|P| \le \ell_v$.
The edges $(u_i,v)_{i=1}^{|P|}$  are then probed in the order $u_{1}, \ldots ,u_{|P|}$ induced by $\pi$.
Since the algorithm must respect commitment, $(u_i,v)$  is probed with probability $\prod_{j=1}^{i-1}(1-p_{u_j,v})$.

Clearly this GKSP rounding approach is problematic, as a fixed vertex $u \in U$ may get probed with probability sufficiently less than $x_{u,v}$,
depending on how highly prioritized it is in the ordering $\pi$. The solutions in the literature involve drawing $\pi$ in such
a way that each vertex $u \in U$ is probed with probability as \textit{close} to $x_{u,v}$ as possible, but it is clear
that no approach is without loss for \textit{every} vertex of $U$.

While this describes the major challenge with generalizing to arbitrary patience, in theory it does not preclude a probing
algorithm from existing which matches the ratio of $1-1/e$, as attainable in the vertex weighted unit patience case.
Unfortunately, the ratio between $\OPT(G)$ and $\LPOPT_{std}(G)$ can become quite small, depending on the values of $(\ell_{v})_{v \in V}$ and the instance $G$. In \cite{Brubach2019}, Brubach et al. define the \textbf{stochasticity gap} of this LP as the infimum of this ratio across all stochastic graphs, namely $\inf_{G}\OPT(G)/ \LPOPT_{std}(G)$. 
Clearly, the notion of a stochasticity gap can be extended to any LP relaxation of the committal benchmark,
as we shall later do in the context of our new LP.

Brubach et al. also consider the following example, thus providing a negative result for the stochasticity
gap of \ref{LP:standard_definition}. 

\begin{example}[\cite{Brubach2019}] \label{example:stochastic_gap}
Fix $n \ge 1$, and construct an unweighted graph  $G_{n}=(U,V,E)$. Suppose that $|U|=|V|=n$ and $\ell_{v} = n$ for all $v \in V$. Set $E:= U \times V$, and define $p_{u,v}:=1/n$ for each $(u,v) \in E$. Observe that $G_n$ corresponds to the Erdős–Rényi random
graph $\mb{G}_{n,n,1/n}$. In this case,
\[
    \mb{E}[ \OPT(G_n)] \le 0.544 \cdot (1+ o(1))\, \LPOPT_{std}(G_n),
\]
where 
the asymptotics are over $n \rightarrow \infty$ \footnote{The example in Brubach et al \cite{Brubach2019} can clearly be extended
to the case when $G_{n}$ has linearly sized patience, that is when $\min_{v \in V} \ell_{v} = \Omega(n)$,
at the expense of the strength of their negative result (the constant $0.544$).}.  
\end{example}

Thus, any probing algorithm which attains a guarantee against $\LPOPT_{std}(G)$ has a provable competitive ratio of at most $0.544$.

In order to get around this limitation, Gamlath et al. \cite{Gamlath2019} consider an LP in the setting of full patience,
which imposes exponentially many constraints, in addition to those of \ref{LP:standard_definition}.
Specifically, for each $v \in V$ and $S \subseteq U$, they ensure that
\begin{equation}\label{eqn:svensson_constraints}
	\sum_{u \in S} p_{u,v} \cdot x_{u,v} \le 1 - \prod_{u \in S}(1 - p_{u,v}).
\end{equation}
Observe that in the variable interpretation of \ref{LP:standard_definition}, the left-hand side corresponds
to the probability a probing algorithm makes a match to a vertex of $S$, and the right-hand side corresponds
to the probability an edge between $v$ and $S$ exists\footnote{The LP considered by Gamlath et al. in \cite{Gamlath2019}
also places the analogous constraints of \eqref{eqn:svensson_constraints} on the vertices of $U$. That being said, these additional constraints are not used anywhere in the work of Gamlath et al., so we do not consider
them when we provide a generalization of their LP in \ref{appendix:committal_benchmark}}. The goal of these additional constraints is thus to force the LP to better capture the behavior of the committal benchmark\footnote{In fact, the results of Gamlath et al.
are proven against the optimum expected matching of $G$, which is equivalent to the non-committal benchmark, as they
work exclusively in the setting of full patience.}.

Using a polynomial time oracle, Gamlath et al. argue that their LP remains poly-time solvable, despite having exponentially
many constraints. As in the setting \ref{LP:standard_definition}, they solve their LP to attain a solution $(x_{u,v})_{u \in U, v \in V}$ for $G$. Each time an online vertex $v \in V$ then arrives, they draw a random permutation $\pi_v$ on a random subset of $U$, which indicates both the probes they intend to make, and the order they intend to make them in. By following $\pi_v$, their procedure is \textit{lossless}; that is, they are able to probe $u \in U$ with probability exactly $x_{u,v}$\footnote{The results of Costello et. al \cite{costello2012matching}
also consider the full patience non-bipartite stochastic matching problem, though without edge weights.
They derive a probing strategy for a fixed vertex $v \in V$ of $G=(V,E)$ and its neighbourhood $N(v)$,
which attains the same guarantee as that of Gamlath et al. \cite{Gamlath2019} through combinatorial techniques.}.

Unfortunately, the results of Gamlath et al., as well as related techniques of Costello et. al \cite{costello2012matching},
do not seem to naturally extend to arbitrary patience. For instance, even the correct modification
of \eqref{eqn:svensson_constraints} is not clear to us. However, we are able to provide a reasonable extension of the Gamlath et al. LP in Appendix \ref{appendix:committal_benchmark}.

\subsection{Defining a New LP}

In this section, we design a new LP for the problem of designing a fixed vertex \textit{lossless} probing algorithm,
which works no matter the edge probabilities, edge weights and patience values $(\ell_{v})_{v \in V}$ 
of $G=(U,V,E)$. Instead of attempting to find the appropriate constraints on the variables of \ref{LP:standard_definition},
we take a different approach. Specifically, we ensure our LP has polynomially many constraints, while allowing
it exponentially many variables to better indicate how the committal benchmark make decisions.

For each $i \ge 1$, denote $U^{(i)}$ as the collection of  tuples of length $i$ constructed from $U$
whose entries are all distinct.
Moreover, set $U^{(\le i)}:= \cup_{j=1}^{i} U^{(j)}$.

For each $v \in V$, $1 \le k \le \ell_v$ and $\bm{u} \in U^{(k)}$, define
\[
	g^{i}_{v}(\bm{u}):= p_{u_i,v} \cdot \prod_{j=1}^{i-1} (1 - p_{u_j,v}),
\]
where $\bm{u}=(u_1, \ldots , u_{k})$, and $i \in [k]$ (here $[k]:=\{1, \ldots ,k\}$). 
Observe that if one reveals the edge states $(\st(u_{i},v))_{i=1}^{k}$
in order, then $g^{i}_{v}(\bm{u})$ corresponds to the probability that $(u_{i},v)$
is the first active edge revealed.

We also define a variable, denoted $x_{v}(\bm{u})$, which may loosely be
interpreted as the likelihood the committal benchmark probes the vertices
in the order specified by $\bm{u}=(u_{1}, \ldots , u_{k})$.
These definitions lead to the following LP:

\begin{align}\label{LP:new}
\tag{LP-new}
&\text{maximize} &  \sum_{v \in V} \sum_{\bm{u} \in U^{( \le \ell_{v})}} \left( \sum_{i=1}^{|\bm{u}|} w_{u_i,v} \, g^{i}_{v}(\bm{u}) \right)\cdot x_{v}(\bm{u}) \\
&\text{subject to} & \sum_{v \in V} \sum_{i=1}^{\ell_v}\sum_{\substack{ \bm{u}^* \in U^{ ( \le \ell_{v})}: \\ u_{i}^*=u}} 
g_{v}^{i}(\bm{u}^*) \cdot x_{v}( \bm{u}^*)  \leq 1 && \forall u \in U  \label{eqn:relaxation_efficiency_matching}\\
&& \sum_{\bm{u} \in U^{( \le \ell_{v})}} x_{v}(\bm{u}) \le 1 && \forall v \in V,  \label{eqn:relaxation_efficiency_distribution} \\
&&x_{v}( \bm{u}) \ge 0 && \forall v \in V, \bm{u} \in U^{(\le \ell_v)}
\end{align}

\ref{LP:new} is a relaxation of the committal benchmark. However, unlike many of the LP formulations
in the stochastic matching literature, we are not aware of an immediate proof of either of these
facts. We instead must introduce a related stochastic probing problem, known as the
\textbf{relaxed stochastic matching problem}, which is exactly encoded by \ref{LP:new}, and whose optimum value upper bounds the committal benchmark. We provide the relevants definitions in Appendix \ref{appendix:committal_benchmark}, where we also prove the following theorem:

\begin{theorem}\label{thm:new_LP_relaxation}

For any stochastic graph $G$, an optimum solution to \ref{LP:new} upper bounds
$\OPT(G)$, the value of the committal benchmark on $G$.

\end{theorem}

Not only is \ref{LP:new} a relaxation of the committal benchmark,
it also can be solved efficiently.
To see this, we first take its dual:

\begin{align}\label{LP:new_dual}
\tag{LP-new-dual}
&\text{minimize} &  \sum_{u \in U} \alpha_{u} + \sum_{v \in V} \beta_{v}  \\
&\text{subject to} & \beta_{v} + \sum_{j=1}^{|\bm{u}^*|} g_{v}^{j}(\bm{u}^*) \cdot \alpha_{u_j^*} \ge \sum_{j=1}^{|\bm{u}^*|} w_{u_j^*,v} \cdot g_{v}^{j}( \bm{u}^*) &&
\forall v \in V, \bm{u}^* \in U^{ ( \le \ell_v)} \\
&&  \alpha_{u} \ge 0 && \forall u \in U\\
&& \beta_{v} \ge 0 && \forall v \in V
\end{align}
In Appendix \ref{appendix:efficient_probing_algorithms}, we argue that \ref{LP:new_dual} has a polynomial time separation oracle,
by solving an optimization problem similar to the one considered by Brubach et al. \cite{Brubach2019}. By standard duality
techniques involving the ellipsoid algorithm \cite{Groetschel,GartnerM}, this allows us to find a solution to \ref{LP:new} in polynomial time, no matter the patience values of $G$ (see \cite{Williamson,Adamczyk2017,Lee2018} for similar examples).
That being said, \ref{LP:new} clearly always has an optimum solution, which can be found efficiently when $\ell_{max}:=\max_{v \in V} \ell_{v}$
is a constant, independent of the size of $|U|$. Moreover, our results are all in the context of competitive analysis, and so the ratios we present in the various online stochastic matching models all hold, independently of the fact that \ref{LP:new} can be solved in poly-time.

Suppose now that we are presented a feasible solution, 
say $(x_{v}(\bm{u}))_{\bm{u} \in U^{(\le \ell_v)}, v \in V}$,
to \ref{LP:new}, for the stochastic graph $G=(U,V,E)$.
For each $v \in V$ and $u \in U$, define
\begin{equation}
	\til{x}_{u,v}:= \sum_{i=1}^{\ell_{v}} \sum_{\substack{\bm{u}^* \in U^{( \le \ell_{v})}: \\ u_{i}^{*} = u}} \frac{ g^{i}(\bm{u}^*) \cdot x_{v}( \bm{u}^*)}{p_{u,v}}. 
\end{equation}
In order to simplify our notation in the later sections,
we refer to the values $(\til{x}_{u,v})_{u \in U,v \in V}$
as the \textbf{(induced) edge variables}
of the solution $(x_{v}(\bm{u}))_{v \in V,\bm{u} \in U^{(\le \ell_v)}}$.

If we now fix $s \in V$, then we can easily
leverage constraint \eqref{eqn:relaxation_efficiency_distribution} to
argue that the edge variables $(\til{x}_{u,s})_{u \in U}$ can be probed without loss. Specifically,
we may execute the following \textit{fixed vertex} probing algorithm, which we refer to
as \textsc{VertexProbe}:

\begin{algorithm}[H]
\caption{VertexProbe} 
\label{alg:vertex_probe}
\begin{algorithmic}[1]
\Statex Input the stochastic graph $G=(U,V,E)$,
a fixed node $s \in V$ and the variables $(x_{s}(\bm{u}))_{\bm{u} \in U^{ (\le \ell_s)}}$ associated
to $s$ in a solution to \ref{LP:new} for $G$.

\State Initialize $\scr{M} \leftarrow \emptyset$.
\State Return $\scr{M}$ with probability $1 - \sum_{\bm{u} \in U^{ (\le \ell_s)}} x_{s}(\bm{u})$
\Comment{\textit{pass} with a certain probability.}. 
\State Draw $\bm{u}^*$ from $U^{( \le \ell_s)}$ with
probability $x_{s}(\bm{u}^*)$ (see \eqref{eqn:relaxation_efficiency_distribution}).
\State Denote $\bm{u}^* = (u_{1}^*, \ldots ,u_{k}^*)$ for $k := |\bm{u}^*|$.
\For{$i=1, \ldots ,k$}
\State Probe $(u_{i}^*,s)$.
\If{$\st(u_i^*,s)=1$}
\State  Set $\scr{M}(s) \leftarrow u_{i}^*$ and return $\scr{M}$. 
\EndIf
\EndFor
\State Return $\scr{M}$.

\end{algorithmic}
\end{algorithm}


Observe the following claim, which follows immediately from the definition
of the edge variables, $(\til{x}_{u,v})_{u \in U, v \in V}$:

\begin{lemma}\label{lem:fixed_vertex_probe}

Let $G=(U,V,E)$ be a stochastic graph with \ref{LP:new} solution $(x_{v}(\bm{u}))_{v \in V, \bm{u} \in U^{ (\le \ell_v)}}$, and whose induced edge variables we denote $(\til{x}_{u,v})_{u \in U, v \in V}$.
If the \textsc{VertexProbe} algorithm is passed a fixed node $s \in V$, then each node
$u \in U$ is probed with probability $\til{x}_{u,s}$.

Moreover, the edge $(u,s)$ is returned by the algorithm with probability $p_{u,v} \cdot \til{x}_{u,s}$.

\end{lemma}

\begin{remark}
We say that \textsc{VertexProbe} \textbf{commits} to the edge $(u,s)$,
provided the algorithm outputs this edge when executing on the fixed node $s \in V$.
\end{remark}

Before providing an overview of our results, we provide a general algorithmic template which unifies how
all of our online probing algorithms are implemented\footnote{In Section \ref{sec:known_stochastic},
we must use a modified \textsc{VertexProbe} subroutine in the ROM setting with edge weights
to improve the competitive ratio from $1/2$ to $1-1/e$.}.
%

Let $G=(U,V,E)$ be an adversarially generated stochastic graph, with arbitrary patience values, edge weights and edge probabilities.
We may always assume that the online probing algorithm has access to the offline vertices $U$, but its information regarding
$V$ and $E$ is limited, depending on the online model we work in.
In general however, we always assume that the online vertices of $V$ are presented in some order to the online probing algorithm, say $v_{1}, \ldots ,v_{n}$ where $n=|V|$, either through an adversarial, ROM or i.i.d  arrival process. We refer to a vertex $v_{t}$ as \textbf{arriving} at \textbf{time} $1 \le t \le n$. We then follow the general high level template for defining online probing algorithms:

\begin{algorithm}[H]
\caption{General Template} 
\label{alg:general_template}
\begin{algorithmic}[1]
\Statex Input $U$ of $G=(U,V,E)$, as well as the remaining information regarding $G$, depending on the online model.
\State $\scr{M} \leftarrow \emptyset$.
\For{$t=1, \ldots , n$} 
\State Using $U$ and the probing decisions of the previous arrivals (together with $G=(U,V,E)$ if it is known), 
compute a stochastic graph $H_t =(U_t, V_t, E_t)$ which contains arrival $v_t$,
and satisfies $U_{t} \subseteq U$. \label{eqn:online_arrival_process}
\State Compute an optimum solution of \ref{LP:new} for $H_{t}$, say $(x_{v}(\bm{u}))_{v \in V_{t}, \bm{u} \in U_{t}^{ (\le \ell_v)}}$.
\State Set $e_t \leftarrow \textsc{VertexProbe}(H_t, (x_{v_t}(\bm{u}))_{\bm{u} \in U_{t}^{ (\le \ell_{v_t})}}, v_t)$.
\State If $e_t \neq \emptyset$, then denote $e_{t}=(u_t,v_t)$, and if $u_t$ is currently unmatched, set $\scr{M}(v_t) \leftarrow u_t$.
\EndFor

\State Return $\scr{M}$.

\end{algorithmic}
\end{algorithm}

We once again emphasize that the online probing algorithm has no control over the order of the arrivals,
so step \eqref{eqn:online_arrival_process} is the only place this algorithmic template can be modified. 
Choosing the ``correct'' choice of $H_t$ depends on how much information we are privy to (e.g., is the
stochastic graph known or unknown), and whether we wish our probing algorithm to execute \textbf{adaptively} - 
that is, depend upon the \textit{probing outcomes} of the previous nodes, $v_{1}, \ldots , v_{t-1}$ -- or not,
that is, execute \textbf{non-adapatively}\footnote{When processing
an online node arrival, say $v_{t}$, a \textbf{non-adaptive} online probing algorithm can base its probes of $v_t$ on
the identities of the previous vertex arrivals, say $v_{1}, \ldots , v_{t-1}$,
as well as their edge weights, edge probabilities and patience values (as well as $G$ itself if it is known). The probes of $v_t$
\textit{cannot} however depend upon the previously probed edge states of $(\st(u,v_{k}))_{u \in U, k \in [t-1]}$.}. We note that in our results, we will execute Algorithm \ref{alg:general_template} non-adaptively and the resulting algorithms will be non-greedy. In the subsequent sections, we investigate these issues in detail, and attempt to attain or approach the same  competitive ratios
one can get in the classical (non-stochastic) online matching settings (when there is a meaningful generalization). 
All of our results are proven by comparing the performance of the relevant online probing algorithm to \ref{LP:new}. We are typically
able to make use of the classical techniques in the literature (with some key modifications), and so we again emphasize that one of
our main technical contributions is in generalizing to arbitrary patience from the more tractable unit/full patience
settings and the new LP that is used to derive our results. We also argue that many of the results in the stochastic matching literature actually hold against the non-committal benchmark, as we discuss in detail in Appendix \ref{appendix:non_committal_benchmark} for the case of
general (i.e., not necessarily bipartite) stochastic graphs. We argue this by proving
that \ref{LP:standard_definition} is a relaxation of the non-committal benchmark. 


\subsection{An Overview of Our Results}

With these definitions in mind, 
we now reiterate and point ahead to our main results as first stated in our abstract. 
All of our results apply to arbitrary patience and the competitive ratios are with respect to the committal benchmark.

\begin{enumerate}

\item Theorem \ref{thm:known_vertex_weights} shows that Algorithm \ref{alg:known_stochastic_graph} is an online algorithm with competitive ratio $1-\frac{1}{e}$ in the following stochastic setting:

\begin{itemize}
\item There is a known stochastic graph
\item Online vertices are given adversarially 
\item Offline vertices have weights 
\end{itemize}

This result shows that the $.544$ inapproximation bound against the LP relaxation in Bansal et al. \cite{BansalGLMNR12} does not hold with respect to our new LP relaxation (see Example \ref{example:stochastic_gap}). 

\item Theorem \ref{thm:modified_known_stochastic_graph} shows that Algorithm \ref{alg:modified_known_stochastic_graph} is 
an online algorithm with competitive ratio $1-\frac{1}{e}$ in the following setting:

\begin{itemize}
\item There is a known stochastic graph
\item Online vertices are presented in an order determined by a uniform at random permutation of the online vertices in the stochastic graph (i.e., stochastic random order model) 
\item Edges have weights

\end{itemize}

This algorithm generalizes\footnote{In Appendix \ref{appendix:non_committal_benchmark},
we provide a reasonable generalization of the Gamlath et al. LP, and show that it attains the same value as \ref{LP:new}.} the online probing algorithm considered by Gamlath et. al \cite{Gamlath2019} in what they refer to as the \textit{query-commit model}.

\item Theorem \ref{thm:iid} shows that Algorithm \ref{alg:known_iid} is an
online algorithm with competitive ratio $1-\frac{1}{e}$ in the following stochastic i.i.d. setting (improving
upon the previously best ratio of $0.46$ in \cite{Brubach2019}):

\begin{itemize}
\item There is a known stochastic (type) graph
\item Online vertices are drawn independently and identically from a distribution on the online vertices (with their adjacent stochastic edges)
\item Edges have weights

\end{itemize} 

In the classical i.i.d. setting with non-integral arrival rates, 
Manshadi et al. \cite{ManshadiGS12} present an example that shows that $1-1/e$ is optimal for \textbf{classically non-adaptive}\footnote{Manshadi et al. \cite{ManshadiGS12} use the terminology non-adaptive
to mean that a (classical) online algorithm in the known i.i.d. setting uses
only the \textit{type} of the arriving node to determine its matching decisions. Observe that this restriction
is sufficient to ensure that an online probing algorithm is non-adaptive (by our definition) in the stochastic matching setting.} algorithms.  
Our algorithm fits this classical definition and applies to non-integral arrival rates and hence our algorithm has an optimal competitive ratio amongst this restricted class of probing algorithms. 

\item Theorem \ref{thm:ROM_edge_weights} shows that Algorithm \ref{alg:ROM_edge_weights} is an online algorithm with (tight) competitive ratio $\frac{1}{e}$ in the following setting:

\begin{itemize}
\item The stochastic graph is not known to the algorithm
\item Online vertices are
given in random order
\item Edges have weights

\end{itemize}
This generalizes the classical non-stochastic result of Kesselheim et al.
\cite{KRTV2013}.

\end{enumerate}

All of our probing algorithms are randomized and implemented \textit{non-adaptively}. 
In Appendix \ref{appendix:adaptivity_results} we discuss the implications 
of the non-adaptivity by considering the relevant \textbf{adaptivity gaps} of the online stochastic matching problems
we consider. Roughly speaking, an adaptivity gap is the worst case ratio of performance between the \textit{optimum} non-adaptive probing algorithm, and the committal benchmark\footnote{We provide a more precise definition of adaptivity gaps in Appendix \ref{appendix:adaptivity-results}.}.

\section{Known Stochastic Graphs: Adversarial and ROM Arrivals}\label{sec:known_stochastic}
\label{sec:benchmark-comparisons}

In this section, we restrict our attention to online bipartite  stochastic matching in the  setting where the stochastic graph is known to the algorithm. We use our  new LP to guide the sequence of probes for each of the online vertices. We first prove Theorem \ref{thm:known_vertex_weights}, showing that Algorithm \ref{alg:known_stochastic_graph} achieves a $1-\frac{1}{e}$ competitive ratio for the setting of offline vertex weights and adversarial online arrivals.

We then consider Algorithm \ref{alg:known_stochastic_graph} in the case of arbitrary edge weights, under the assumption that the online nodes arrive in random order, thus attaining a competitive ratio of $1/2$ (which we show is tight for this algorithm).
By considering a modification of Algorithm \ref{alg:known_stochastic_graph}, we can improve this competitive ratio to $1-1/e$
using the techniques of Ehsani et al. \cite{Ehsani2017} and Gamlath et al. \cite{Gamlath2019}.
This extends the recent work of \cite{Gamlath2019} to arbitrary patience in what they
refer to as the \textit{query-commit model}.

The results of this section also yield a lower bound (positive result) on the adaptivity gap of the bipartite stochastic
matching problem with one-sided patience.

\subsection{Defining the Probing Algorithm}

We now consider the probing algorithm which is the subject of Theorems \ref{thm:known_vertex_weights}
and \ref{thm:known_ROM_edge_weights}. 

\begin{algorithm}[H]
\caption{Known Stochastic Graph} \label{alg:known_stochastic_graph}
\begin{algorithmic}[1]
\Statex Input $G=(U,V,E)$, a stochastic graph with edge probabilities $(p_{e})_{e \in E}$, edge weights $(w_e)_{e \in E}$ and patience parameters $(\ell_{v})_{v \in V}$    
\State Set $\scr{M} \leftarrow \emptyset$.

\State Solve \ref{LP:new}, and find an optimal solution $(x_{v}(\bm{u}))_{v \in V, \bm{u} \in U^{( \le \ell_v)}}$. 

\For{$t=1, \ldots , |V|$}
\State Process $v_t$ (vertex arriving at time $t$)
\State Set $(u_t,v_t) \leftarrow \textsc{VertexProbe}(G, (x_{v_t}(\bm{u}))_{\bm{u} \in U^{( \le \ell_{v_t})}}, v_t)$.
\If{$(u_t,v_t) \neq \emptyset$ and $u_t$ is unmatched}
\State Set $\scr{M}(v_t) = u_t$.
\EndIf
\EndFor
\State Return $\scr{M}$.
\end{algorithmic}
\end{algorithm}

\subsection{Adversarial Arrivals}

We first consider the known stochastic online matching problem in the case of arbitrary patience, offline vertex weights and adversarial online vertex arrivals. Specifically, we provide a proof of Theorem \ref{thm:known_vertex_weights}.

\begin{theorem}
\label{thm:known_vertex_weights}

If Algorithm \ref{alg:known_stochastic_graph} is passed a stochastic graph $G = (U,V,E)$ with offline vertex weights $(w_{u})_{u \in U}$ (that is, $w_{u, v} = w_u$ for all $(u, v) \in E$) and arbitrary patience, then
\[
	\mb{E}[ \val( \scr{M})] \ge \left(1 - \frac{1}{e} \right) \cdot \OPT(G).
\]
Thus, the competitive ratio of this algorithm (when the stochastic graph and order of online vertices is chosen by an adversary) is $1 - 1/e$ against the committal benchmark.
\end{theorem}

\begin{proof}

Let us now denote $\val(\scr{M})$ as the value of the matching returned by Algorithm \ref{alg:known_stochastic_graph}. Observe that
\[
	\mb{E}[ \val( \scr{M})] = \sum_{u \in U} w_{u} \,\mb{P}[ \text{$u$ is matched by the algorithm}].
\]
As such, for each fixed $u \in U$, we may focus on lower bounding the probability that the algorithm matches it.

Recall that associated with the solution $(x_{v}(\bm{u}))_{v \in V, \bm{u} \in U^{(\ell_{v})}}$
are the edge variables $(\til{x}_{u,v})_{u \in U, v \in V}$, as defined following \ref{LP:new} in Section \ref{sec:prelim}.

Observe now that $u \in U$ is matched by the algorithm, if and only if there exists some $v \in V$ which commits to $u$ while
executing \textsc{VertexProbe}.  For a fixed $v$, we denote this event by $C(u,v)$. Using Lemma \ref{lem:fixed_vertex_probe}, we get that
\[
	\mb{P}[C(u,v)] =  p_{u,v} \, \til{x}_{u,v},
\]
for each $v \in V$.

Now, since the executions of \textsc{VertexProbe} are independent, so
are the events $\{ \neg C(u,v)\}_{v \in V}$, and so
\[
  \mb{P}[ \text{$u$ is not matched}] = \prod_{v \in V} ( 1 - p_{u,v} \, \til{x}_{u,v} ).
\]
As a result,
\begin{align*}
  \mb{P}[ \text{$u$ is matched} ] &= 1 - \prod_{v \in V}(1 - p_{u,v} \, \til{x}_{u,v})   \\
                  &\ge 1 - \prod_{v \in V} \exp\left( - p_{u,v} \, \til{x}_{u,v} \right) \\
                  &= 1 - \exp\left( -\sum_{v \in V} p_{u,v} \, \til{x}_{u,v} \right),
\end{align*}

as $1- z \le \exp(-z)$ for all $z \in \mb{R}$ (here we use $\exp(z):= e^{z}$ for notational clarity).

Now $(x_{v}(\bm{u}))_{v \in V, \bm{u} \in U^{( \le \ell_v)}}$ is a feasible solution to \ref{LP:new}, and so
\[
\sum_{v \in V} p_{u,v} \, \til{x}_{u,v} \le 1.
\]
We may therefore conclude that
\[
\mb{P}[ \text{$u$ is matched}] \ge ( 1 - \exp(-1))\sum_{v \in V} p_{u,v} \, \til{x}_{u,v},
\]
since $1 - \exp(-z) \ge (1 - \exp(-1)) \, z$ for all $0 \le z \le 1$. 

Thus,
\begin{align*}
  \mb{E}[\val(\scr{M})] &= \sum_{u \in U} w_{u} \, \mb{P}[ \text{$u$ is matched}] \\
            &\ge ( 1 - \exp(-1))\sum_{u \in U} \sum_{v \in V} w_{u} \,  p_{u,v} \, \til{x}_{u,v} \\
            & = (1 - \exp(-1)) \, \LPOPT_{new}(G),
\end{align*}

as $(x_{v}(\bm{u}))_{v \in V, \bm{u} \in U^{( \le \ell_v)}}$ is an optimal solution to \ref{LP:new}. By Theorem \ref{thm:new_LP_relaxation}, $\OPT(G) \le \LPOPT_{new}(G)$, and so the proof is complete.

\end{proof}

Suppose that we fix an ordering $\pi$ of $V$. We can then define $\OPT(G,\pi)$ to be
the largest expected value an \textit{online} probing algorithm can attain on $G$, provided
it is presented the vertices $V$ in the order $\pi$. 
With this definition, we can define the \textbf{order gap} of $G$ as
the worst case ratio between $\OPT(G, \pi_{1})$ and $\OPT(G, \pi_{2})$, over all orderings
$\pi_{1}$ and $\pi_{2}$ of $V$; that is, the ratio $\text{order}(G)$, where
\begin{equation}\label{eqn:order_gap}
	 \text{order}(G) := \frac{\inf_{\pi_{1}}\OPT(G, \pi_{1})}{\sup_{\pi_{2}}\OPT(G, \pi_{2})}.
\end{equation}
If $\mathscr{C}$ corresponds to the collection
of all vertex weighted stochastic graphs,
then we can define the \textbf{order gap} of $\mathscr{C}$ as the minimum 
value of $\text{order}(G)$ over all $G \in \mathscr{C}$; that is,
the value $\inf_{G \in \mathscr{C}} \text{order}(G)$.

Since Algorithm \ref{alg:known_stochastic_graph} achieves a competitive ratio of $1-1/e$ no matter the order the online vertices
are presented to it, we observe the following corollary:

\begin{corollary}
The order gap of the collection of vertex weighted stochastic graphs is no smaller than $1-1/e$. Note that this is a positive result. 
\end{corollary}

We contrast this observation with an upper bound (negative result) on the order gap.

\begin{example}\label{example:order-gap}
Let us consider the bipartite graph $G = (U, V, E)$ where $U = \{u_1, u_2\}$ is the set of offline vertices and $V = \{v_1, v_2\}$ is the set of online vertices each of which has unit patience. Denote $E = \{(u_1, v_1), (u_2, v_1), (u_2, v_2)\}$ as the set of edges of $G$, with edge probabilities $p_{u_1, v_1} = \frac{1}{2}$, $p_{u_2, v_1} = 1$ and $p_{u_2, v_2} = \frac{1}{2}$. 
The order gap of $G$ is at most $0.8$. 
\end{example}

\begin{proof}
If we process vertex $v_1$ before vertex $v_2$, we either probe $(u_1, v_1)$ or $(u_2, v_1)$ while processing $v_1$. If we probe the former, the only other edge that we can probe is $(u_2, v_2)$ and the expected size of the constructed  matching  is $\frac{1}{2} + \frac{1}{2} = 1$, as both of these edges are active with probability $\frac{1}{2}$. If instead we probe edge $(u_2, v_1)$, then it is always active and there are no more edges that we can probe. This also gives an expected matching size of $1$. Thus, the maximum expected size of the matching is 1  when $v_1$ is processed before $v_2$.

If instead we process vertex $v_2$ before $v_1$, we may probe $(u_2, v_2)$, which is active with probability $\frac{1}{2}$. If the edge is active, we next probe edge $(u_1, v_1)$, which is active with probability $\frac{1}{2}$. In this case, expected size of matching found is $\frac{3}{2}$. If $(u_2, v_2)$ is inactive, we instead probe $(u_2, v_1)$, which is always active and so in this case, the expected size of matching found is $1$. Thus, the maximum expected size of matching when $v_2$ is probed before $v_1$ is $\frac{1}{2}\frac{3}{2} + \frac{1}{2}1 = \frac{5}{4}$. 
\end{proof}

While it would be interesting to know the precise value of the order gap (even just for the case of offline vertex weights),
$1-1/e$ is the limitation of our techniques for proving positive results, as demonstrated by the following example:

\begin{example}\label{example:unit_patience_unweighted}
Consider a graph $G$ with a single offline node $u$ and a collection of $n$ online nodes $V$. For each edge $e=(u,v)$ with $v \in V$, set $p_{u,v}:= 1/n$. As the example is in the setting of unit patience, \ref{LP:new} and \ref{LP:standard_definition} are equivalent,
and in particular, $\LPOPT_{new}(G) = \LPOPT_{std}(G)$. Thus, we describe the remainder of the example with respect
to the definition of \ref{LP:standard_definition} for simplicity.

Observe that the LP solution $x_{u,v}:=1$ for each $v \in V$ satisfies the constraints
of \ref{LP:standard_definition}. Moreover, it evaluates to an objective value of $1$. Thus, $LPOPT_{std}(G) \ge 1$.

Observe now that if we consider an arbitrary probing algorithm, then its only option is to probe the edges of $u$ in some arbitrary order (or not at all). Of course, each edge is active with probability $1/n$, so we observe that
\[
  \mb{P}[\text{$G$ has a least one active edge}] = 1 - (1 - 1/n)^{n} =(1 +o(1))\left(1 - \frac{1}{e} \right),
\] 
as we allow $n \rightarrow \infty$.

As a result, 
\[
  \inf_{G}\frac{\OPT(G)}{\LPOPT_{new}(G)}  \le \left(1 - \frac{1}{e} \right),
\]
and so in particular, Algorithm \ref{alg:known_stochastic_graph} achieves the best possible bound against $\LPOPT_{new}(G)$.
\end{example}

\subsection{Random Order Arrivals}

We now consider the known stochastic matching problem in the case of arbitrary edges weights and ROM arrivals. We first prove Theorem \ref{thm:known_ROM_edge_weights}, which shows that Algorithm \ref{alg:ROM_edge_weights} gets a competitive ratio of $1/2$.
After arguing that the analysis is tight, we introduce a modified \textsc{VertexProbe} algorithm, which
we refer to as \textsc{VertexProbe-S}. By replacing \textsc{VertexProbe} with the subroutine \textsc{VertexProbe-S}
in Algorithm \ref{alg:known_stochastic_graph}, we are able to improve the competitive guarantee to $1-1/e$.

\begin{theorem}\label{thm:known_ROM_edge_weights}
In the ROM input model, if Algorithm \ref{alg:known_stochastic_graph} is passed a stochastic graph $G = (U,V,E)$ with arbitrary edge weights $(w_e)_{e\in E}$ and patience $(\ell_{v})_{v \in V}$, then
\[
	\mb{E}[ \val(\scr{M})] \ge \frac{1}{2} \cdot \OPT(G),
\]
Thus, against the committal benchmark, the competitive ratio of this algorithm (when the stochastic graph is chosen by an adversary and order of online vertices is determined uniformly at random) is 1/2.
\end{theorem}

We include the proof of Theorem \ref{thm:known_ROM_edge_weights} in Appendix \ref{appendix:deferred_proofs},
as it has a relatively simple analysis and helped motivate the improvement to $1-1/e$. 
We now consider the following example, which confirms the performance guarantee of Algorithm \ref{alg:known_stochastic_graph} is tight:

\begin{example} \label{example:edge_weight_random_order}
Let $G = (U, V, E)$ be a bipartite graph with a single offline node $u$,
online vertices $V = \{v_1, v_2\}$ and edges $E = \{(u, v_1), (u, v_2)\}$.
We assume that the online nodes have unit patience.

Fix $0 < \epsilon < 1$, and define the edge probabilities 
$p_{(u, v_1)} := \epsilon$ and $p_{u, v_2} := 1 - \epsilon$.
Moreover, define the weights of the edges as $w_{u, v_1} := 1/\epsilon$ and $w_{u, v_2} = \epsilon/(1 - \epsilon)$. 

For this instance, if we allow $\epsilon \rightarrow 0$,
then the expected weight of matching returned by Algorithm \ref{alg:known_stochastic_graph} 
in the ROM setting is at most half that of $\OPT(G)$.
\end{example}

\begin{proof}
Since we work in the unit patience setting for $G$, 
we express the relevant linear program as in the setting of \ref{LP:standard_definition}:
\begin{align}
&\text{maximize}&   w_{u, v_1} \cdot p_{u, v_1} \cdot x_{u, v_1} + w_{u, v_2} \cdot p_{u, v_2} \cdot x_{u, v_2}\\
&\text{subject to}&  \: p_{u, v_1} \cdot x_{u, v_1} + p_{u, v_2} \cdot x_{u, v_2} \leq 1\\
& &     0 \leq x_{u, v_1} \leq 1 \\
& &     0 \leq x_{u, v_2} \leq 1
\end{align}

The optimal solution to this LP corresponds to $x_{u, v_1} = x_{u, v_2} = 1$, and the optimal value is $1 + \epsilon$.

Now, when considering the order in which $v_1$ arrives before $v_2$, if we probe the edge $e \in E$ with probability $x_e$, the expected value of matching returned is 
\[
w_{u, v_1} \, x_{u, v_1} \, p_{u, v_1} + (1 - x_{u, v_1} \, p_{u, v_1}) \, w_{u, v_2} \, x_{(u, v_2)}\, p_{(u, v_2)} = 1 + (1 - \epsilon) \, \epsilon. 
\]
Similarly, when considering the order in which $v_2$ arrives before $v_1$, if we probe the edge $e\in E$ with probability $x_e$, then the expected value of matching returned is 
\[
w_{u, v_2} \, x_{u, v_2} \, p_{u, v_2} + (1 - x_{u, v_2} \, p_{u, v_2}) \, w_{u, v_1} \, x_{u, v_1} \, p_{u, v_1} = 2 \, \epsilon.
\]
Thus, as the order of arrivals is determined uniformly at random,
the expected value of the matching returned is
\[
 \frac{1}{2}(2\epsilon + 1 + \epsilon - \epsilon^2),
\] 
which tends to $1/2$ as $\epsilon$ tends to $0$.
Moreover, the optimum value of the LP tends to $1$ as $\epsilon$ tends to $0$,
so the ratio of these values tends to $1/2$.
Moreover, for this specific choice of $G$, 
\[
	\LPOPT_{new}(G) = \OPT(G),
\]
and so the proof is complete.
\end{proof}

We remark that there is an online algorithm that achieves the optimum expectation since the stochastic graph is known. That is, the online algorithm would simply not probe $(u,v_1)$, if $v_{1}$ is the first arrival.

\subsubsection{Improving Upon the Competitive Ratio} \label{sec:known_ROM_improvement}

In \cite{Gamlath2019}, Gamlath et al. showed, among other things,
that when $G=(U,V,E)$ has full patience, there exists a probing algorithm
which achieves an approximation ratio of $1-1/e$. This algorithm
in fact executes in the online setting, in which $G$ is known and the probing algorithm
respects a vertex order that is generated uniformly at random.

In addition to using their new LP relaxation (as discussed in Section \ref{sec:prelim}),
Gamlath et al. adapted the techniques of Ehsani et al. \cite{Ehsani2017} from the \textit{prophet secretary problem} to get a competitive ratio of $1-1/e$ in the case of full patience. We now generalize their algorithm to attain the same competitive ratio, while handling arbitrary patience constraints. We emphasize that our analysis
proceeds almost identically, though this is only made possible by our definition of \ref{LP:new} \footnote{In Appendix \ref{appendix:committal_benchmark}, we provide a reasonable generalization of the Gamlath et al. LP, and show that it attains the same value as \ref{LP:new}.}.

Given an arbitrary stochastic graph $G=(U,V,E)$, let us suppose we are presented
an optimum solution to \ref{LP:new}, denoted $(x_{v}(\bm{u}))_{v \in V, \bm{u} \in U^{ ( \le \ell_v)}}$,
whose edge variables we denote by $(\til{x}_{u,v})_{u \in U, v\in V}$. In this case, define
\[
	c_{u} := \sum_{v \in V} w_{u,v} \, p_{u,v} \, \til{x}_{u,v}
\]
for each $u \in U$. We can view $c_{u}$ as corresponding to the contribution of $u$
to the evaluation of $(x_{v}(\bm{u}))_{v \in V, \bm{u} \in U^{ ( \le \ell_v)}}$ as a solution
to \ref{LP:new}. Specifically, observe that
\begin{equation}
\sum_{u \in U} c_{u} = \sum_{u \in U, v \in V} w_{u,v} \, p_{u,v} \, \til{x}_{u,v} = \LPOPT_{new}(G).
\end{equation}

Let us now return to the ROM setting, though we describe it in a slightly different way for the stochastic graph $G=(U,V,E)$. As in Devanur et al. \cite{DJK2013}, 
for each $v \in V$, draw $Y_{v} \in [0,1]$ independently and uniformly at random. We assume that the vertices of $V$ are presented to the algorithm in an increasing order, based on the values of $(Y_{v})_{v \in V}$. 
In this way, we say that vertex $v \in V$ \textit{arrives at time $Y_{v}$}. 
Observe that the vertices of $V$ are presented to the algorithm in a uniformly at random order, so this interpretation is equivalent to the ROM setting.

We now describe a modification of Algorithm \ref{alg:known_stochastic_graph} that is more selective
as to which of the edges returned by $\textsc{VertexProbe}$ we are willing to accept.
Specifically, when $v$ arrives at time $Y_{v}$, \textsc{VertexProbe} is executed as before.
However, if \textsc{VertexProbe} returns the edge $e=(u,v)$, then we only add $e$ to the matching
provided $u$ is unmatched \textit{and} $w_{e} \ge (1 - e^{ Y_{v} -1}) \cdot c_{u}$. Of course,
this high level description of the online probing algorithm clearly does \textit{not} respect commitment,
but fortunately we can run a \textit{simulated} version of \textsc{VertexProbe},
which we refer to as \textsc{VertexProbe-S}.

\begin{algorithm}[H]
\caption{VertexProbe-S} 
\label{alg:modified_vertex_probe}
\begin{algorithmic}[1]
\Statex Input the stochastic graph $G=(U,V,E)$, a fixed node $s \in V$, the variables $(x_{s}(\bm{u}))_{\bm{u} \in U^{ (\le \ell_s)}}$ associated to $s$ in a solution to \ref{LP:new} for $G$, and $0 \le z \le 1$.
\State Initialize $\scr{M} \leftarrow \emptyset$.
\State Return $\scr{M}$ with probability $1 - \sum_{\bm{u} \in U^{ (\le \ell_s)}} x_{s}(\bm{u})$
\Comment{\textit{pass} with a certain probability.} 
\State Draw $\bm{u}^*$ from $U^{( \le \ell_s)}$ with
probability $x_{s}(\bm{u}^*)$ (see \eqref{eqn:relaxation_efficiency_distribution}).
\State Denote $\bm{u}^* = (u_{1}^*, \ldots ,u_{k}^*)$ for $k := |\bm{u}|$.
\For{$i=1, \ldots ,k$}
\If{$w_{u_{i}^*,s} \ge (1 - e^{z - 1}) \cdot c_{u_i^*}$}
\State Probe $(u_{i}^*,s)$.
\If{$\st(u_i^*,s)=1$}
\State  Set $\scr{M}(s) \leftarrow u_{i}^*$ and return $\scr{M}$. 
\EndIf
\Else{ draw $Z \sim \Ber(p_{u_{i}^*,v})$ independently.} \Comment{a Bernoulli of parameter $p_{u_{i}^*,v}$.}
\If{$Z = 1$}
\State Set $\scr{M}(s) \leftarrow u_{i}^*$ and return $\scr{M}$. \Comment{drawing $Z$ simulates an edge probe.}
\EndIf
\EndIf
\EndFor
\State Return $\scr{M}$.

\end{algorithmic}
\end{algorithm}

\begin{remark}
Observe that \textsc{VertexProbe-S} makes a probe to the edge $(u,s)$,
only if 
\[
w_{u,s} \ge (1 - e^{z - 1}) \cdot c_{u}.
\]
If this condition is not satisfied, then it still may return the edge $(u,s)$, however $(u,s)$
will not be probed. We make sure to return $(u,s)$, so that \textsc{VertexProbe-S} can be
coupled with \textsc{VertexProbe}, as this will simplify the proof of Theorem \ref{thm:modified_known_stochastic_graph}.

\end{remark}

We say that \textsc{VertexProbe-S} \textbf{commits} to the edge $(u,s)$,
provided it returns this edge (even if it doesn't actually probe $(u,s)$), when executed
with the parameter $0 \le z \le 1$.
Observe that \textsc{VertexProbe-S} returns $(u,s)$ with the same
probability as \text{VertexProbe}, so we can make use of Lemma \ref{lem:fixed_vertex_probe}
to get an analogous guarantee.

\begin{lemma} \label{lem:fixed_vertex_probe_sim}

Suppose $G=(U,V,E)$ is a stochastic graph with fixed node $s \in V$
and \ref{LP:new} solution $(x_{v}(\bm{u}))_{v \in V, \bm{u} \in U^{ (\le \ell_v)}}$, whose induced edge variables we denote by $(\til{x}_{u,v})_{u \in U, v \in V}$.

If \textsc{VertexProbe-S} is passed the fixed node $s$, then $(u,s)$ is returned
with probability $p_{u,v} \cdot \til{x}_{u,s}$ for each $u \in U$, no matter which
value of $0 \le \alpha \le 1$ is presented to \textsc{VertexProbe-S}. Moreover, 
the edge $(u,s)$ is probed only if $w_{u,s} \ge (1 - e^{z - 1}) \cdot c_{u}$.
\end{lemma}

We now can implement a modified version of Algorithm \ref{alg:known_stochastic_graph}
which executes identically to the high level modification we just described, while
respecting commitment.

\begin{algorithm}[H]
\caption{Modified Known Stochastic Graph} \label{alg:modified_known_stochastic_graph}
\begin{algorithmic}[1]
\Statex Input $G=(U,V,E)$, a stochastic graph with edge probabilities $(p_{e})_{e \in E}$, edge weights $(w_e)_{e \in E}$ and patience parameters $(\ell_{v})_{v \in V}$.    
\State Set $\scr{M} \leftarrow \emptyset$.

\State Solve \ref{LP:new}, and find an optimal solution $(x_{v}(\bm{u}))_{v \in V, \bm{u} \in U^{( \le \ell_v)}}$. 
\State For each $v \in V$, draw $Y_{v} \in [0,1]$ independently and uniformly at random.

\For{$v \in V$ in increasing order of $Y_v$}
\State Set $(u,v) \leftarrow \textsc{VertexProbe-S}(G, v, (x_{v}(\bm{u}))_{\bm{u} \in U^{( \le \ell_{v})}}, Y_{v})$. 
\If{$(u,v) \neq \emptyset$, and $w_{u,v} \ge (1 - e^{Y_v - 1}) \cdot c_{u}$} 
\If{$u$ is unmatched} 
\State Set $\scr{M}(v) = u$. 				\Comment{$(u,v)$ is matched only if $(u,v)$ is probed and $\st(u,v)=1$.}
\EndIf
\Else
\State Pass on $(u,v)$.
\EndIf
\EndFor
\State Return $\scr{M}$.
\end{algorithmic}
\end{algorithm}

\begin{theorem}\label{thm:modified_known_stochastic_graph}
In the ROM input model, if Algorithm \ref{alg:modified_known_stochastic_graph} is passed a stochastic graph $G = (U,V,E)$ with arbitrary edge weights $(w_e)_{e\in E}$ and patience $(\ell_{v})_{v \in V}$, then
\[
	\mb{E}[ \val(\scr{M})] \ge \left(1 - \frac{1}{e} \right) \cdot \OPT(G),
\]
Thus, the competitive ratio of this algorithm is $1-1/e$ against the committal benchmark.
\end{theorem}

The analysis of Theorem \ref{thm:modified_known_stochastic_graph} follows very closely
the full patience proof presented in Gamlath et al. \cite{Gamlath2019}, and hence is
mainly motivated by the single item prophet secretary problem of Ehsani et al. \cite{Ehsani2017}. 
However, for sake of completeness, we include the argument.

\begin{proof}[Proof of Theorem \ref{thm:modified_known_stochastic_graph}]

For each offline node $u \in U$, denote $\val(\scr{M}(u))$ as the weight
of the edge assigned to $u$ (which is zero, if $u$ remains unmatched).

Observe then that
\[
	\mb{E}[ \val(\scr{M})] = \sum_{u \in U} \mb{E}[ \val(\scr{M}(u))].
\]
Thus, in order to complete the proof it suffices to show
that
\begin{equation} \label{eqn:approximate_contribution}
	\mb{E}[ \val(\scr{M}(u))] \ge \left(1 - \frac{1}{e} \right) \cdot c_{u}
\end{equation}
for each $u \in U$, as we know that $\sum_{u \in U} c_{u} = \LPOPT_{new}(G) \ge \OPT(G)$ (by Theorem \ref{thm:new_LP_relaxation}).

As such, let us suppose $u \in U$ is fixed for the remainder of the proof.
The remaining computations follow Ehsani et al. \cite{Ehsani2017} for the single
item prophet secretary problem, though we must make use Lemma \ref{lem:fixed_vertex_probe_sim}
and constraint \eqref{eqn:relaxation_efficiency_matching} of \ref{LP:new}.

Let us now define the random variables $N_{u}$
and $M_{u,v}$ where:
\begin{enumerate}
	\item $N_{u} := \sum_{v \in V} (1 - e^{Y_{v} -1}) \cdot c_{u} \cdot \bm{1}_{[ \scr{M}(u) = v]}$,
	\item $M_{u,v} := \left( w_{u,v} - (1 - e^{Y_{v} -1}) \cdot c_{u} \right) \cdot \bm{1}_{[ \scr{M}(u) = v]}$.
\end{enumerate}
That is, if $u$ is matched to $v$, then $N_{u}$ is assigned the value $(1 - e^{Y_{v} -1}) \cdot c_{u}$
and $M_{u,v}$ is assigned $\left( w_{u,v} - (1 - e^{Y_{v} -1}) \cdot c_{u} \right)$.
Proving \eqref{eqn:approximate_contribution} thus reduces to showing that
\[
	\mb{E}[N_{u}] + \sum_{v \in V} \mb{E}[ M_{u,v}] \ge \left(1 - \frac{1}{e} \right) \cdot c_{u}.
\]
In other words, the cumulative amount assigned to $u$ and the vertices of $V$ is at least $\left(1 - 1/e \right) \cdot c_{u}$ in expectation\footnote{Ehsani et al. \cite{Ehsani2017} and Gamlath et al. \cite{Gamlath2019} provide a utility/revenue interpretation
of these variables.}.

We first focus on lower bounding $\mb{E}[N_{u}]$. Let us define $T_{u}$ as the arrival time of the online vertex which matches to $u$ (which is $1$, if no
such vertex exists). For convenience, define the functions $r=r(t), F=F(t)$ and $\alpha = \alpha(t)$, where for $t \in [0,1]$.
\[
	r(t) := \mb{P}[ T_{u} \ge t], \: F(t):= 1 - r(t), \: \text{and} \: \alpha(t):= 1 - e^{t-1}.
\]
Define $C(u,v)$ as the event in which $v$ commits to $u$
when $\textsc{VertexProbe-S}$ is executed using $v \in V$.
Observe that the events $\{C(u,v)\}_{v \in V}$ are independent,
as the executions of $\textsc{VertexProbe-S}$ in Algorithm \ref{alg:modified_known_stochastic_graph}
are themselves 
independent.

For each $v\in V$, let us now define the function $\psi_{u,v} = \psi_{u,v}(t)$
where $\psi_{u,v}(t):= \mb{P}[\text{$Y_{v} < t$ and $w_{u,v} \ge \alpha(t) \cdot c_{u}$}]$ for $t \in [0,1]$.
If we now apply Lemma \ref{lem:fixed_vertex_probe_sim},
and additionally make use of the independence of the variables $(Y_{v})_{v \in V}$,
then we may conclude that
\begin{align*}
	r(t) &= \mb{P}[ T_{u} \ge t] \\
		 &= 1 - \prod_{v \in V}(1 - \psi_{u,v}(t) \cdot \mb{P}[C(u,v)]) \\
		 &= 1 - \prod_{v \in V}(1 - \psi_{u,v}(t) \cdot \til{x}_{u,v})
\end{align*}
for each $t \in [0,1]$.

Thus, since the functions $(\psi_{u,v})_{v \in V}$ are clearly all continuously differentiable on
$(0,1)$, $r$ is continuously differentiable as well.
On the other hand, observe that $F$ is the c.d.f of $T_u$, so we have that
\begin{align*}
	\mb{E}[ N_{u}] &= \int_{0}^{1} \alpha(t) \cdot c_{u} \cdot dF(t) \\
				   &=\int_{0}^{1} \alpha(t) \cdot c_{u} \cdot F'(t) \, dt \\
				   &= - \int_{0}^{1} \alpha(t) \cdot c_{u} \cdot r'(t) \, dt
\end{align*}
Thus,
\[
	\mb{E}[ N_{u}]  = - \int_{0}^{1} \alpha(t) \cdot c_{u} \cdot r'(t) \, dt,
\] 
and so we may apply integration by parts to get that
\begin{align*}
	\int_{0}^{1} \alpha(t) \cdot c_{u} \cdot r'(t) \, dt &= c_{u} \, \left( [ r(t) \cdot \alpha(t)]_{t=0}^{1} - \int_{0}^{1} r(t) \cdot \alpha'(t) \, dt \right) \\
		&= c_{u} \left( (1 - 1/e) + \int_{0}^{1} r(t) \cdot \alpha'(t) \, dt \right).
\end{align*}
To conclude, 
\begin{equation}\label{eqn:revenue_lower_bound}
	\mb{E}[ N_{u}] = c_{u} \left( (1 - 1/e) + \int_{0}^{1} r(t) \cdot \alpha'(t) \, dt \right).
\end{equation}

Let us now focus on lower bounding $\sum_{v \in V} \mb{E}[ M_{u,v}]$.
First observe that if $0 \le t \le 1$ satisfies $w_{u,v} \ge \alpha(t) \cdot c_{u}$, then
\[
	\mb{E}[ M_{u,v} \, | \, Y_{v} = t] = p_{u,v} \, \til{x}_{u,v} \cdot (w_{u,v} - \alpha(t) \cdot c_{u} ) \cdot \mb{P}[T_{u} \ge t \, | \, Y_{v} =t],
\]
as $v$ is matched to $u$ with probability $p_{u,v} \, \til{x}_{u,v}$, given $u$ is unmatched at time $t$. Moreover,
if $0 \le t \le 1$ satisfies $w_{u,v} < \alpha(t) \cdot c_{u}$, then $\mb{E}[ M_{u,v} \, | \,  Y_{v} = t] = 0$.
Thus, for all $0 \le t \le 1$,
\begin{equation} \label{eqn:edge_lower_bound_utitlity}
	\mb{E}[ M_{u,v} \, | \, Y_{v} = t] \ge p_{u,v} \, \til{x}_{u,v} (w_{u,v} - \alpha(t) \cdot c_{u} ) \cdot \mb{P}[T_{u} \ge t \, | \, Y_{v} =t].
\end{equation}

On the other hand, it is clear that $\mb{P}[T_{u} \ge t \, | \, Y_{v} =t] \ge \mb{P}[T_{u} \ge t ]$.
Thus, after applying  \eqref{eqn:edge_lower_bound_utitlity} and observing $r(t) = \mb{P}[T_{u}  \ge t]$, we get that
\begin{align*}
	\sum_{v \in v} \mb{E}[ M_{u,v} \, |  \, Y_{v} = t] &\ge \sum_{v \in V} p_{u,v} \, \til{x}_{u,v} \cdot ( w_{u,v} - \alpha(t) \cdot c_{u} ) \cdot r(t) \\
												&= \left( c_{u} - \alpha(t) \cdot c_{u} \cdot \sum_{v \in V} p_{u,v} \, \til{x}_{u,v} \right) \cdot r(t).
\end{align*}
Now, we know that $\sum_{v \in V} p_{u,v} \, \til{x}_{u,v} \le 1$ by constraint \eqref{eqn:relaxation_efficiency_matching} of \ref{LP:new}.
Thus,
\[
	\sum_{v \in v} \mb{E}[ M_{u,v} \, | \, Y_{v} = t] \ge  c_{u} \cdot (1 - \alpha(t))  \cdot r(t),
\]
and so since the random variables $(Y_{v})_{v \in V}$ are uniformly distributed, we get that
\[
	\sum_{v \in v} \mb{E}[ M_{u,v} ] \ge  c_{u} \cdot \int_{0}^{1} (1 - \alpha(t))  \cdot r(t) \, dt.
\]
By combining this equation with \eqref{eqn:revenue_lower_bound}, it follows that
\begin{align*}
	\mb{E}[ N_{u}] + \sum_{v \in v} \mb{E}[ M_{u,v} ]  &\ge c_{u} \cdot \left( (1 - 1/e) + \int_{0}^{1} r(t) \cdot \alpha'(t) \, dt \right) + c_{u} \cdot \int_{0}^{1} (1 - \alpha(t))  \cdot r(t) \, dt   \\		
							&= c_{u} \cdot (1 - 1/e) + c_{u} \, \int_{0}^{1} r(t) \cdot (1 - \alpha(t) + \alpha'(t)) \, dt \\
							&= (1 -1 /e) \cdot c_{u}, 													   
\end{align*}	
as $1 - \alpha(t) + \alpha'(t) = 0$ for all $t \in [0,1]$.

As this argument holds for each $u \in U$, the proof is complete.

\end{proof}

\section{Online Stochastic Matching in the Known I.I.D. Model} \label{sec:stochastic_known_iid}
\label{iid-patience1}
In this section, we consider a generalization of the classical known i.i.d. matching problem (introduced in Feldman et al. \cite{FeldmanMMM09}) to the stochastic matching setting (as first studied in Bansal et al. \cite{BansalGLMNR12}). Once we review the relevant framework and terminology, we introduce an online probing algorithm which achieves a competitive ratio of $1-1/e$ for arbitrary patience and edge weights, thereby proving
Theorem \ref{thm:iid}. Our algorithm generalizes the algorithm of \cite{BrubachSSX16}
to arbitrary patience, in which Brubach et al. proved a competitive ratio of $1-1/e$ for the unit patience setting. This allows us to
improve upon  the previously best known competitive ratio of $0.46$ for the case of arbitrary patience,
as presented in \cite{BrubachSSX20}.

We again note that our positive results from Section \ref{sec:known_stochastic} do \textit{not} seem to 
immediately imply the same positive results in the known i.i.d. stochastic matching problem, as we explain in more detail
after reviewing the model.

\subsection{The Known I.I.D. Stochastic Setting}
Let us suppose that $G = (U,V,E)$ is a stochastic graph with edges weights $(w_{e})_{e \in E}$, edge probabilities $(p_{e})_{e \in E}$ and offline patience values $(\ell_{v})_{v \in V}$ associated with it. In the known i.i.d. setting, we refer to $G$ as a \textbf{stochastic type graph} (or \textbf{type graph} when clear),
and the vertices of $V$ as the \textbf{type nodes} of $G$.

Now, fix a parameter $n \ge 1$ (which need not be equal to $|V|$), indicating the number of \textbf{rounds} or \textbf{arrivals} to occur. Moreover, consider $\bm{r}=(r_{v})_{v \in V}$, where $r_v > 0$ for each $v \in V$, and $\sum_{v \in V} r_{v} = n$. We refer to $r_{v}$ as the (fractional) \textbf{arrival rate} of type node $v \in V$. An input to the stochastic known i.i.d. matching problem then consists of the tuple $(G, \bm{r}, n)$, which
we refer to as a \textbf{known i.i.d. input} with \text{fractional arrival rates}.

An \textbf{online probing algorithm} $\scr{A}$ is given access to $(G, \bm{r},n)$ as part of its input. For each $t=1, \ldots ,n$, \textit{vertex arrival} $v_{t} \in V$ is drawn independently in round $t$ using the distribution $\bm{r}/n$, at which point $v_t$ is said to be of type $v \in V$, provided
$v_{t}=v$. We emphasize that the edge states of $v_t$ are statistically independent from
the edge states of all the previously drawn nodes (even if $v_t$ is not the first vertex of
type $v$ to arrive).

Using all past available information regarding the outcomes of the probes involving $v_{1}, \ldots ,v_{t-1}$, together with the edge probabilities $(p_{u,v_t})_{u \in U}$, weights $(w_{u,v_t})_{u \in U}$ and patience value $\ell_{v_t}$, $\scr{A}$ may
probe up to $\ell_{v_t}$ edges adjacent to $v_t$. The algorithm is again restricted by commitment, in that $v_{t}$ may only be matched
to the first $u \in U$ for which the probe to $(u,v_{t})$ confirms that the edge is active.
%

Observe that while the type graph $(G, \bm{r},n)$ is passed as input to $\scr{A}$, the stochastic graph $\scr{A}$ actually executes on is in fact randomly generated, and unknown to $\scr{A}$. Let us denote this
(random) stochastic graph by $\hat{G}=(U, \hat{V}, \hat{E})$. Here,
$\hat{V}$ consists of the random arrival nodes of $V$ presented to the algorithm, and
$\hat{E}$ includes all the relevant edges between  $U$ and $\hat{V}$ (since the same node
from $V$ can arrive multiple times, $\hat{V}$ and $\hat{E}$ are multisets). We assume that $\hat{G}$
also encodes all the edge weights, probabilities and patience values induced from the arrival
nodes of $\hat{V}$.

We refer to $\hat{G}$ as the \textbf{instantiated stochastic graph} or simply the \textbf{instantiated graph} when clear.
Observe that since $(G ,\bm{r}, n)$ encodes the distribution of $\hat{G}$, we say that $\hat{G}$ is distributed according to the known i.i.d. input $(G, \bm{r},n)$, which we denote by $\hat{G} \sim (G, \bm{r},n)$.

Denote $\val( \scr{A}( \hat{G}))$ as the value of the matching $\scr{A}$ constructs when passed the instantiated graph $\hat{G}$. Our performance measure for $\scr{A}$ then involves averaging over
all the possible instantiations of $\hat{G}$. Specifically, we wish to maximize
\[
	\mb{E}[ \val(\scr{A}(\hat{G}))],
\]
where the expectation is over the randomness in drawing $\hat{G}$ from $(G,\bm{r},n)$, together with the inherent randomness in the states of the edges of $\hat{G}$, as well as any randomized decisions $\scr{A}$ may make.

For each randomly drawn $\hat{G} \sim (G,\bm{r},n)$, we can consider
the committal benchmark, and the evaluation it takes on $\hat{G}$,
namely $\OPT(\hat{G})$. This yields a committal probing strategy,
which we refer to as the \textbf{committal benchmark} for the stochastic type
graph $(G,\bm{r},n)$. We denote the expected performance of the committal benchmark
by $\OPT(G,\bm{r},n)$. Observe that
\[
	\OPT(G,\bm{r},n)= \mb{E}[ \OPT(\hat{G})], 
\]
where the expectation is over the randomness in generating $\hat{G}$. 
We can define the \textbf{non-committal benchmark} for the stochastic
type graph $(G,\bm{r},n)$ analogously, which we denote by $\OPT_{non}(G,\bm{r},n)$.

The standard in the literature (see \cite{Adamczyk15, BansalGLMNR12,BrubachSSX20}) is to prove
competitive ratios against the committal benchmark. More precisely, the
goal is to find an online probing algorithm $\scr{A}$ for which the (strict) competitive ratio
\[
	\inf_{(G, \bm{r}, n)} \frac{ \mb{E}[ \val(\scr{A}(\hat{G})) ]}{ \OPT(G, \bm{r},n)}
\]
is as close to $1$ as possible.

Before continuing, we emphasize that there does \textit{not} seem
to be an obvious reduction from the known i.i.d. stochastic matching problem
to the known stochastic matching problem with ROM arrivals. Specifically,
suppose we are presented an online probing algorithm $\scr{A}$ which achieves
competitive ratio $0 <c \le 1 $ in the known stochastic matching problem with ROM arrivals.
In this case, let us now fix a stochastic type graph $(G,\bm{r},n)$, and imagine
trying to use $\scr{A}$ to design a probing algorithm for the i.i.d. matching problem.
If we consider the instantiated graph $\hat{G}$ drawn from $(G,\bm{r},n)$,
then the online vertices of $\hat{G}$ will indeed be presented to $\scr{A}$ in
a random order. That being said, in order for $\scr{A}$
to attain to attain a competitive guarantee of $c \cdot \OPT(\hat{G})$, it needs to be presented the entire description of $\hat{G}$ as well. However, an online probing algorithm in the known i.i.d. setting is only given
access to the type graph, $(G,\bm{r},n)$, \textit{not} the instantiated graph $\hat{G}$.
Moreover, $\hat{G}=(U,\hat{V}, \hat{E})$ may be a very different stochastic graph than $G$; for instance,
type node $v \in V$ may appear multiple times in $\hat{V}$, or perhaps not at all.
As such, it is unclear how to modify $\scr{A}$ to obtain the same competitive ratio of $c$
against $\OPT(G, \bm{r}, n)$.

\subsection{Defining an LP Relaxation}
Given an input $(G, \bm{r}, n)$ to the known i.i.d. matching problem, it is challenging to directly compare the performance of an online probing algorithm to that of the committal benchmark; that is, the value $\OPT(G,\bm{r},n)$. Instead, we once again focus on LP based approaches for upper bounding this quantity.

Let us now review the LP introduced in \cite{BansalGLMNR12,BrubachSSX20}, as defined for $(G, \bm{r},n)$,
specialized to the case of one-sided patience.

\begin{align}\label{LP:standard_definition_iid}
\tag{LP-std-iid}
&\text{maximize} & \sum_{u \in V, v \in V} w_{u,v} \, p_{u,v} \, y_{u,v}  \\
&\text{subject to} & \sum_{v \in V} p_{u, v} \, y_{u,v} & \leq 1 && \forall u \in U \\
& & \sum_{u \in U} p_{u,v} \, y_{u,v} & \le r_v && \forall v \in V \\
& & \sum_{u \in U} y_{u, v} & \leq r_v \cdot \ell_v && \forall v \in V \\
& & 0 \le  y_{u, v} &  \le r_v  && \forall u \in U, v \in V
\end{align}

If $\LPOPT_{std-iid}( G, \bm{r}, n)$ denotes the value of the optimal solution to \ref{LP:standard_definition_iid}, then
it was shown by Bansal et al. \cite{BansalGLMNR12} to be a relaxation of the committal benchmark; that is, 
\[
	\OPT(G,\bm{r},n) \le \LPOPT_{std-iid}( G, \bm{r}, n).
\]
Unfortunately, \ref{LP:standard_definition_iid} suffers the same issues as \ref{LP:standard_definition}, as Example \ref{example:stochastic_gap} continues to apply, as can be seen by setting $r_{v}=1$ for $v \in V$
and $n= |V|$. As such, we introduce a new LP for $(G, \bm{r},n)$, using the same ideas as in the derivation of \ref{LP:new}. 
The essential difference in this LP being that we incorporate the
arrival rates of $(G, \bm{r},n)$, as can be seen below in constraint \eqref{eqn:relaxation_efficiency_distribution_iid}.

\begin{align}\label{LP:new_iid}
\tag{LP-new-iid}
&\text{maximize} &  \sum_{v \in V} \sum_{\bm{u} \in U^{( \le \ell_{v})}} \left( \sum_{i=1}^{|\bm{u}|} w_{u_i,v} \, g^{i}_{v}(\bm{u}) \right) y_{v}(\bm{u}) \\
&\text{subject to} & \sum_{v \in V} \sum_{i=1}^{\ell_v}\sum_{\substack{ \bm{u}^* \in U^{ ( \le \ell_{v})} \\ u_{i}^*=u}} 
g_{v}^{i}(\bm{u}^*) \, y_{v}( \bm{u}^*)  \leq 1 && \forall u \in U  \label{eqn:relaxation_efficiency_matching_iid}\\
&& \sum_{\bm{u} \in U^{( \le \ell_{v})}} y_{v}(\bm{u}) \le r_v && \forall v \in V,  \label{eqn:relaxation_efficiency_distribution_iid} \\
&&y_{v}( \bm{u}) \ge 0 && \forall v \in V, \bm{u} \in U^{(\le \ell_v)}
\end{align}
Given a feasible solution to \ref{LP:new_iid}, say $(y_{v}(\bm{u})_{v \in V, \bm{u} \in U^{( \le \ell_v)}}$ we define the \textbf{induced edge variables}, hereby denoted $(\til{y}_{u,v})_{u \in U, v \in V}$, as in the case of \ref{LP:new};
that is,
\[
	\til{y}_{u,v} := \sum_{i=1}^{\ell_v} \sum_{\substack{\bm{u}^{*} \in U^{ (\le \ell_v)}: \\ u_{i}^{*} = u}} \frac{g_{v}^{i}(\bm{u}^*) \, y_{v}(\bm{u})}{p_{u,v}}. 
\]

Let us now denote $\LPOPT_{new-iid}(G, \bm{r},n)$ as the value of an optimal solution to \ref{LP:new_iid}.
We claim that \ref{LP:new_iid} is a relaxation of the committal benchmark.
This follows from a conditioning argument involving an application
of Theorem \ref{thm:new_LP_relaxation} to $\hat{G} \sim (G, \bm{r},n)$, so we defer the details
to Appendix \ref{appendix:deferred_proofs}. We remark that the techniques used in this
proof constitute a general method for extending LP relaxations to the stochastic known
i.i.d. setting, and so they may be of independent interest.

\begin{lemma}\label{lem:LP_validity_known_iid}
For any input $(G, \bm{r}, n)$ of the known i.i.d stochastic matching problem,
\[
	\OPT(G,\bm{r},n) \le \LPOPT_{new-iid}(G, \bm{r},n).
\]
\end{lemma}

The procedure for solving \ref{LP:new_iid} efficiently follows the same approach
as that of \ref{LP:new}, so we defer the details to Appendix \ref{appendix:efficient_probing_algorithms}.

\subsection{Defining a Known I.I.D. Probing Algorithm}
We now consider an online probing algorithm for the known i.i.d. stochastic matching problem,
which generalizes the unit patience probing algorithm of Brubach et al. \cite{BrubachSSX16}. 

Given $(G, \bm{r},n)$, suppose that we consider a feasible
solution to \ref{LP:new_iid}, which we denote by $(y_{v}( \bm{u}))_{v \in V, \bm{u} \in U^{(\le \ell_v)}}$.
If we fix $v \in V$, then the values $(y_{v}(\bm{u})/ r_v)_{\bm{u} \in U^{( \le \ell_v)}}$
satisfy,
\begin{equation}
	\sum_{ \bm{u} \in U^{( \le \ell_v)}} \frac{ y_{v}(\bm{u})}{r_v} \le 1,
\end{equation}
as a result of constraint \eqref{eqn:relaxation_efficiency_distribution_iid}.
As such, given the input $(U, (y_{v}(\bm{u})/ r_v)_{\bm{u} \in U^{( \le \ell_v)}}, v)$
for a fixed $v \in V$, we can execute \textsc{VertexProbe}. In particular, we can apply Lemma \ref{lem:fixed_vertex_probe} in the
known i.i.d. setting to get the following lemma:

\begin{lemma}\label{lem:commitment_probability}
Fix $u \in U$ and $v \in V$. For each $t=1, \ldots ,n$, denote
$C(u,v_t)$ as the event in which Algorithm \ref{alg:known_iid} commits $v_t$ to $u$ in one of
its probes. In this case,
\[
	\mb{P}[ C(u,v_t) \, | \, v_t = v] = \frac{\til{y}_{u,v} \, p_{u,v}}{r_v},
\]

where $(\til{y}_{u,v})_{u \in U, v \in V}$ are the induced edge variables
of the solution $(y_{v}(\bm{u}))_{v \in V, \bm{u} \in U^{( \le \ell_v)}}$.

\end{lemma}

\begin{proof}

When Algorithm \ref{alg:known_iid} processes $v_{t}$, it 
executes $\textsc{VertexProbe}(U, (y_{v_t}( \bm{u}^*)/r_{v_t})_{\bm{u}\in U^{( \le \ell_{v_t})}}, v_t)$.
If we condition on the event in which $v_{t}=v$, then this corresponds to executing
\textsc{VertexProbe} using the input $(U, (y_{v}( \bm{u}^*)/r_{v})_{\bm{u}\in U^{( \le \ell_{v})}}, v)$.
As a result, an application of Lemma \ref{lem:fixed_vertex_probe} ensures that
\[
	\mb{P}[ C(u,v_t) \, | \, v_t = v] = \frac{\til{y}_{u,v} \, p_{u,v}}{r_v},
\]
thus completing the proof.

\end{proof}

We now adapt Algorithm \ref{alg:known_stochastic_graph} to the known i.i.d. setting,
leading to the following algorithm:

\begin{algorithm} 
\caption{Known I.I.D.}\label{alg:known_iid}
\begin{algorithmic}[1] 
\Statex Input $G=(U,V,E)$, an arbitrary stochastic type graph. 
\Statex Input $n \ge 1$, the number of arriving vertices, and the arrivals rates of $V$, $\bm{r}=(r_{v})_{v \in V}$.

\State Set $\scr{M} \leftarrow \emptyset$.

\State Solve \ref{LP:new_iid}, and find an optimal solution $(y_{v}( \bm{u}))_{v \in V, \bm{u} \in U^{(\le \ell_v)}}$.
\For{$t=1, \ldots , n$}
\State Let $v_t$ be the vertex that arrives at time t.
\State Identify the type of $v_t$ in $V$, and the corresponding values $(y_{v_t}( \bm{u}^*)/r_{v_t})_{\bm{u}^* \in U^{( \le \ell_{v_t})}}$
\State Set $(u_{t},v_t) \leftarrow \textsc{VertexProbe}(U, (y_{v_t}( \bm{u}^*)/r_{v_t})_{\bm{u}^*\in U^{( \le \ell_{v_t})}}, v_t)$.
\If{$(u_t,v_t) \neq \emptyset$ and $u_t$ is unmatched}
\State Set $\scr{M}(v_t) = u_t$.
\EndIf
\EndFor
\State Return $\scr{M}$.
\end{algorithmic}
\end{algorithm}

\begin{theorem}
\label{thm:iid}
Algorithm \ref{alg:known_iid} achieves a competitive ratio of $1-1/e$ against
the committal benchmark, for arbitrary edge weights and patience values.
\end{theorem}

\begin{proof}

Let us fix $u \in U$ and $v \in V$, where $G=(U,V,E)$. While Algorithm \ref{alg:known_iid} executes
on the instantiated graph $\hat{G}=(U,\hat{V},\hat{E})$, let us say that the algorithm
matches the edge $e=(u,v) \in E$, provided there exists some $1 \le t \le n$
for which $v_t = v$ and $\scr{M}(v_t) = u$ (here $v_{1}, \ldots ,v_{n}$ are the ordered
arrivals of the vertices of $\hat{V}$). Observe then that
\[
	\mb{E}[ \val( \scr{M})] = \sum_{e \in E} w_{e} \, \mb{P}[\text{$e$ is matched}].
\]
As such, we focus on lower bounding $\mb{P}[\text{$e$ is matched}]$
for each $e \in E$.

Observe now that
\[
	\mb{P}[\text{$e$ is matched}] = \sum_{t=1}^{n} \mb{P}[ \scr{M}(v_t)=u \, | \, v_{t} =v ] \cdot \mb{P}[v_t =v]. 
\]
Moreover, if $R_{t} \subseteq U$ denotes the unmatched vertices of $U$ after
vertices $v_{1}, \ldots ,v_{t-1}$ arrive, then

\begin{align*}
	\mb{P}[ \scr{M}(v_t)=u \, | \, v_{t} =v ] &= \mb{P}[ C(u,v_t) \cap \{u \in R_t\} \, | \, v_{t} =v ]	\\
										&= \mb{P}[C(u,v_t) \, | \, v_{t} =v] \cdot \mb{P}[ u \in R_{t} \, | \,v_{t} = v],
\end{align*}
as the events $C(u,v_t)$ and $u \in R_{t}$ are conditionally independent given $v_{t}=v$,
since the algorithm decides upon the probes of $v_t$ independently from those of $v_{1}, \ldots ,v_{t-1}$.

Moreover, the event $u \in R_{t}$ can be determined from the probes of the vertices $v_{1}, \ldots ,v_{t-1}$,
and is therefore independent from the event $v_{t}=v$. Thus,
\[
	\mb{P}[ \scr{M}(v_t)=u \, | \, v_{t} =v ] = \mb{P}[C(u,v_t) \, | \, v_{t} =v] \cdot \mb{P}[ u \in R_{t}],
\]
and so
\[
	\mb{P}[ \scr{M}(v_t)=u \, | \, v_{t} =v ] = \til{y}_{u,v} \, p_{u,v} \, \mb{P}[ u \in R_{t}],
\]
after applying Lemma \ref{lem:commitment_probability}.

It suffices to lower bound $\mb{P}[ u \in R_{t}]$.
Observe that for each $k=1, \ldots ,n-1$,
\[
	\mb{P}[ u \in R_{k+1}] = \mb{P}[ \cap_{j=1}^{k}  \neg C(u,v_j) ] = \mb{P}[ \neg C(u,v_{k})] \cdot \mb{P}[ u \in R_{k}]
\]
as the probes of $v_k$ are drawn independently from those of $v_{1}, \ldots ,v_{k-1}$.

Yet,
\begin{align*}
	\mb{P}[ C(u,v_k)] &= \sum_{v \in V} \mb{P}[  C(u,v_k) \, | \, v_{k} = v] \cdot \mb{P}[v_{k}=v]	\\
						   &= \sum_{v \in V} \frac{\til{y}_{u,v} \, p_{u,v}}{r_v} \, \frac{r_{v}}{n} \\
						   &= \sum_{v \in V} \frac{\til{y}_{u,v} \, p_{u,v}}{n} \\
						   &\le \frac{1}{n},
\end{align*}
by Lemma \ref{lem:commitment_probability} and the constraints of \ref{LP:new_iid}.
Thus,
\begin{equation}
	\mb{P}[ u \in R_{t}] \ge \left(1 - \frac{1}{n} \right)^{t-1}
\end{equation}
after applying the above recursion.

As a result,
\[
	\mb{P}[ \scr{M}(v_t)=u \, | \, v_{t} =v ] \ge p_{u,v} \, \til{y}_{u,v} \left(1 - \frac{1}{n} \right)^{t-1},
\]
and so
\begin{align*}
	\mb{P}[ \text{$(u,v)$ is matched}] &= \sum_{t=1}^{n} \mb{P}[ \scr{M}(v_t)=u \, | \, v_{t} =v ] \cdot \mb{P}[v_t =v] \\
								   &\ge \sum_{t=1}^{n} \frac{\til{y}_{u,v} \, p_{u,v}}{r_{v}} 
								   \left(1 - \frac{1}{n} \right)^{t-1} \frac{r_v}{n} \\
								   &= \sum_{t=1}^{n} \left(1 - \frac{1}{n} \right)^{t-1} \frac{\til{y}_{u,v} \, p_{u,v}}{n}.
\end{align*}
Now, $\sum_{t=1}^{n}\frac{1}{n}\left(1 - \frac{1}{n} \right)^{t-1} \ge 1- \frac{1}{e}$,
so
\[
	\mb{P}[ \text{$(u,v)$ is matched}] \ge \left(1 - \frac{1}{e} \right) \til{y}_{u,v} \, p_{u,v}
\]
for each $u \in U, v \in V$. As such
\[
	\mb{E}[ \val( \scr{M})] \ge \sum_{u \in U, v \in V} w_{u,v} \, \til{y}_{u,v} \, p_{u,v} \left(1 - \frac{1}{e} \right).
\]
Since $(y_{v}( \bm{u}))_{v \in V, \bm{u} \in U^{(\le \ell_v)}}$ is an optimum solution to \ref{LP:new_iid},
the algorithm is $1 - 1/e$ competitive by Lemma \ref{lem:LP_validity_known_iid}, thus completing the proof.
\end{proof}

\section{Online Stochastic Matching with ROM Arrivals: The Case of an Unknown Stochastic Graph} \label{sec:unknown_edge}

In this section, we consider the unknown stochastic matching problem in
the setting of arbitrary edge weights. Specifically, we
employ the LP based techniques of the previous section to design a randomized
probing algorithm which generalizes the approach of Kesselheim et al. \cite{KRTV2013}. As in \cite{KRTV2013}, we make the added assumption that the number of vertex arrivals
is known to the online probing algorithm ahead of time. We are then able
to prove a best possible asymptotic competitive ratio of $1/e$,
though unlike the work of Kesselheim et al. \cite{KRTV2013}, our online algorithm requires randomization.

Let us suppose that $G=(U,V,E)$ is a stochastic graph with arbitrary edge weights, probabilities
and patience values. We assume that $n:=|V|$, and that the online nodes of $V$ are denoted
$v_{1}, \ldots ,v_{n}$, where the order is generated uniformly at random.

Since $G$ is unknown to us in the current setting, we cannot directly solve \ref{LP:new} to define a probing algorithm. As such, we must adjust which LP we attempt to solve.

Let us suppose that $S$ is a non-empty subset of the nodes of $V$. We can then denote $G[S]$ as the \textbf{induced stochastic graph} of $G$ on $S$. This is constructed by taking the induced graph of $G$ on the partite sets $U$ and $S$, and restricting the edge weights and probabilities to $(p_{u,s})_{u \in U, s \in S}$ and $(w_{u,s})_{u \in U, s \in S}$ respectively, as well as the patience values to $(\ell_s)_{s \in S}$.

From now on, denote $V_{t}$ as the set of first $t$ arrivals of $V$; that is, $V_{t}:= \{v_{1}, \ldots ,v_{t}\}$. Moreover, set $G_{t}:= G[V_t]$, and $\LPOPT_{new}(G_t)$ as the value of an optimum solution to \ref{LP:new} (this is a random variable, as $V_{t}$ is a random subset of $V$). The following inequality then holds:

\begin{lemma} \label{lem:random_induced_subgraph}
For each $t \ge 1$, 
\[
    \mb{E}[ \LPOPT_{new}(G_{t}) ] \ge \frac{t}{n} \, \LPOPT_{new}(G).
\]
\end{lemma}

In light of this observation, we design an online probing algorithm which makes use of $V_{t}$, the currently known nodes, to derive an optimum LP solution with respect to $G_{t}$. As such, each time an online node
arrives, we must compute an optimum solution for the LP associated to $G_{t}$, distinct from the solution computed
for that of $G_{t-1}$. 

\begin{algorithm}[H]
\caption{Unknown Stochastic Graph ROM} 
\label{alg:ROM_edge_weights}
\begin{algorithmic}[1]

\Statex Input $U$, $n:=|V|$, and $0 \le \alpha \le 1$.
\State Set $\scr{M} \leftarrow \emptyset$.
\State Set $G_{0} = (U, \emptyset, \emptyset)$

\For{$t=1, \ldots , |V|$}
\State Input $v_{t}$, with $(w_{u,v_t})_{u \in U}$, $(p_{u,v_t})_{u \in U}$ and $\ell_{v_t}$.  
\State Compute $G_{t}$, by updating $G_{t-1}$ to contain $v_{t}$ and its edges into $U$,
as well its edge weights, probabilities and patience.

\If{ $t < |V|\alpha$}
\State Pass on $v_{t}$.
\Else
\State Solve \ref{LP:new} for $G_{t}$ and find an optimum solution.
\State Encode this (new) optimum solution as $(x_{v}(\bm{u}))_{v \in V_{t}, \bm{u} \in U^{( \le \ell_{v})}}$.
\State Process $v_t$, and set $(u_t,v_t) \leftarrow \textsc{VertexProbe}(G_t, (x_{v_t}(\bm{u}))_{\bm{u} \in U^{( \le \ell_{v_t})}}, v_t)$.
\If{$(u_t,v_t) \neq \emptyset$ and $u_t$ is unmatched}
\State Set $\scr{M}(v_t) = u_t$.
\EndIf
\EndIf
\EndFor

\State Return $\scr{M}$.
\end{algorithmic}
\end{algorithm}

\begin{theorem}\label{thm:ROM_edge_weights}

When $\alpha$ is set to $1/e$, Algorithm \ref{alg:ROM_edge_weights} achieves an asymptotic competitive ratio\footnote{The asymptotic competitive ratio for an online probing algorithm $\scr{A}$ in the ROM setting is defined as
$\liminf_{\OPT(G) \rightarrow \infty}  \frac{ \mb{E}[ \val( \scr{A}(G))]}{\OPT(G)}$.} of $1/e$ against the committal
benchmark.

\end{theorem}

\begin{proof}

Observe that by definition, Algorithm \ref{alg:ROM_edge_weights} does not
probe any of the neighbours of $v_{t}$ for $1 \le t \le \alpha n -1$. As such,
these online vertices do not contribute to the matching returned by the algorithm,
and so we hereby fix $t$ and assume that $t \ge \alpha n$. We emphasize
that the value of $x_{v}(\bm{u})$ corresponds to
this fixed value of $t$, for each $v \in V_t$ and $\bm{u} \in U^{(\le \ell_v)}$.

Let us now define $e_{t}:=(u_t,v_{t})$, where $u_{t}$ is the vertex $u \in U$ which $v_{t}$ commits to (recall that $(u_t,v_t) = \emptyset$ if $v_t$ remains uncommitted after its probes). We now define the random variable
\[
	\val(e_t):= w_{e_t} \bm{1}_{[ e_t \neq \emptyset]},
\]
which indicates the weight of the edge $v_t$ commits to (which is
zero, provided $v_t$ remains uncommitted).

For each $u \in U$, denote $C(u,v_t)$ as the event in which
$v_t$ commits to $u$. 
Let us now condition on the random subset $V_{t}$, as well as the 
random vertex $v_{t}$. 
In this case, 
\[
    \mb{E}[ \val(e_{t})  \, | \, V_{t}, v_t] = \sum_{u \in U} w_{u,v_t} \, \mb{P}[ C(u,v_t)  \, | \, V_{t}, v_{t}].
\]
Observe however that once we condition on $V_{t}$ and $v_{t}$, Algorithm \ref{alg:ROM_edge_weights} corresponds to executing \textsc{VertexProbe} on the instance $(G_t, (x_{v_t}(\bm{u}))_{\bm{u} \in U^{( \le \ell_{v_t})}}, v_t)$. Thus,
Lemma \ref{lem:fixed_vertex_probe} implies that
\[
\mb{P}[ C(u,v_t)  \, | \, V_{t}, v_{t}]  =  p_{u,v_t} \, \til{x}_{u,v_t},
\]
where $\til{x}_{u,v_t}$ is the induced edge variable associated with the solution $(x_{v}(\bm{u}))_{v \in V_t, \bm{u} \in U^{(\ell_{v})}}$.
As such,
\[
	\mb{E}[ \val(e_{t})  \, | \, V_{t}, v_t] = \sum_{u \in U} w_{u,v_t} \, p_{u,v_t} \, \til{x}_{u,v_t}.
\]
On the other hand, if we condition on \textit{solely} $V_{t}$, then $v_{t}$ remains distributed uniformly
at random amongst the vertices of $V_{t}$. Moreover, once we condition on $V_t$, the graph $G_t$ is determined, and thus so are
the values $(x_{v}(\bm{u}))_{v \in V_t, \bm{u} \in U^{(\ell_{v})}}$ of \ref{LP:new}. These observations together imply that
\begin{equation}\label{eqn:conditional_expectation_value}
\mb{E}[ w_{u,v_t} \, p_{u,v_t} \,  \til{x}_{u,v_t}  \, | \, V_{t}] = \frac{\sum_{v \in V_t} w_{u,v} \, p_{u,v} \, \til{x}_{u,v}}{t}
\end{equation}
for each $u \in U$ and $\alpha n \le t \le n$.

If we now take expectation over $v_{t}$, then using the law of iterated
expectations,
\begin{align*}
    \mb{E}[ \val(e_t)  \, | \, V_{t}] &= \mb{E}[ \,  \mb{E}[ \val(e_t) \, | \, V_t, v_t] \,   \, | \, V_{t} ]  \\
                               &= \mb{E}\left[ \sum_{u \in U} w_{u,v_t} \, p_{u,v_t} \, \til{x}_{u,v_t}   \, | \, V_{t} \right] \\
                               &= \sum_{u \in U} \mb{E}[ w_{u,v_t} \, p_{u,v_t} \, \til{x}_{u,v_t}  \, | \, V_{t}] \\
                               &= \sum_{u \in U} \sum_{v \in V_t} \frac{w_{u,v} p_{u,v} \, \til{x}_{u,v}}{t},
\end{align*}
where the final equation follows from \eqref{eqn:conditional_expectation_value}.

Observe however that
\[
    \LPOPT_{new}(G_t)=\sum_{v \in V_t} \sum_{u \in U} w_{u,v_t} \, p_{u,v_t} \, \til{x}_{u,v_t},
\]
as $(x_{v}(\bm{u}))_{v \in V_{t}, \bm{u} \in U^{( \le \ell_{v_t})}}$ is an optimum solution to
\ref{LP:new} for $G_t$. As a result,
\[
    \mb{E}[ \val(e_t)  \, | \, V_{t}]  = \frac{\LPOPT_{new}(G_t)}{t},
\]
and so
\[
	\mb{E}[ \val(e_t)]  = \frac{\mb{E}[\LPOPT_{new}(G_t)]}{t},
\]

after taking taking expectation over $V_{t}$.

On the other hand, Lemma \ref{lem:random_induced_subgraph} implies that
\[
	\frac{\mb{E}[\LPOPT_{new}(G_t)]}{t} \ge  \frac{\LPOPT_{new}(G)}{n}.
\]
Thus,
\begin{equation}\label{eqn:edge_value_lower_bound}
\mb{E}[ \val(e_t)] \ge \frac{\LPOPT_{new}(G)}{n},
\end{equation}
provided $ \alpha n \le t \le n$.

Let us now consider the matching $\scr{M}$ returned by the algorithm,
as well as its value, which we denote by $\val(\scr{M})$. For each $\alpha n \le t \le n$,
define $R_{t}$ as the \textit{remaining vertices} of $U$ when vertex $v_{t}$ arrives (these are the unmatched vertices of $U$, after $v_{1}, \ldots ,v_{t-1}$
are processed). With this notation, we have that
\begin{equation}\label{eqn:value_of_matching}
    \val(\scr{M}) = \sum_{t=\alpha n}^{n} \val(u_t, v_t) \, \bm{1}_{[u_t \in R_t]}.
\end{equation}
Moreover, we have the following lemma, whose proof we defer until afterwards.

\begin{lemma} \label{lem:availability_lower_bound}

If $f(t,n):= \alpha n/(t-1)$, then
\[
    \mb{P}[ u_{t} \in R_t  \, | \, V_{t},v_t] \ge f(t,n),
\]

for $t \ge \alpha n$.

\end{lemma}

Now, $\val(u_t, v_t)$ and $\{ u_{t} \in R_t \}$ are conditionally independent given
$(V_{t},v_t)$, as the probes of $v_{t}$ are independent from those of $v_{1}, \ldots ,v_{t-1}$. Thus,
\[
\mb{E}[ \val(u_t,v_t) \, \bm{1}_{[u_t \in R_t]}  \, | \, V_{t}, v_{t}] = \mb{E}[ \val(u_t,v_t)  \, | \, V_t, v_t] \cdot \mb{P}[ u_t \in R_t  \, | \, V_{t}, v_t].
\]
Moreover, for each $t \ge \alpha n$, Lemma \ref{lem:availability_lower_bound} implies that
\[    
   \mb{E}[ \val(u_t,v_t)  \, | \, V_t, v_t] \cdot \mb{P}[ u_t \in R_t  \, | \, V_{t}, v_t] \ge 
   \mb{E}[ \val(u_t,v_t)  \, | \, V_{t}, v_t] \, f(t,n),
\]
and so
\[
\mb{E}[ \val(u_t,v_t) \, \bm{1}_{[u_t \in R_t]}  \, | \, V_{t}, v_{t}] \ge \mb{E}[ \val(u_t,v_t)  \, | \, V_{t}, v_t] \, f(t,n).
\]
Thus, by applying the law of iterated expectations,
\begin{align*}
\mb{E}[ \val(u_t,v_t) \bm{1}_{[u_t \in R_t]} ] &= \mb{E}[ \, \mb{E}[ \val(u_t,v_t) \, \bm{1}_{[u_t \in R_t]}  \, | \, V_{t}, v_{t}] \, ] \\
			&\ge \mb{E}[ \, \mb{E}[ \val(u_t,v_t)  \, | \, V_{t}, v_t] \, f(t,n) \, ] \\
			&= f(t,n) \, \mb{E}[ \val(u_t,v_t)],
\end{align*}
for each $t \ge \alpha n$.

As a result, using \eqref{eqn:value_of_matching}, we get that
\begin{align*}
	\mb{E}[\val(\scr{M})] &= \sum_{t=\alpha n}^{n} \mb{E}[ \val(u_t, v_t) \, \bm{1}_{[u_t \in R_t]} ] \\
						  &\ge \sum_{t=\alpha n}^{n} f(t,n) \, \mb{E}[ \val(u_t,v_t)].
\end{align*}

We may thus conclude that
\[
    \mb{E}[ \val(\scr{M})] \ge \LPOPT_{new}(G) \sum_{t=\alpha n}^{n} \frac{ f(t,n)}{n},
\]
after applying \eqref{eqn:edge_value_lower_bound}.

As $\sum_{t=\alpha n}^{n} f(t,n)/n = (1 + o(1)) 1/e$ when $\alpha =1/e$ (where the asymptotics
are as $n \rightarrow \infty$), the result holds.
\end{proof}

In order to complete the proof of Theorem \ref{thm:ROM_edge_weights}, we must prove Lemma \ref{lem:availability_lower_bound}. Up until now, when Algorithm \ref{alg:ROM_edge_weights}
solves \ref{LP:new} for $G_t$, we have been able to notate
the induced edge variables as $(\til{x}_{u,v})_{u \in U, v \in V_t}$ without
ambiguity, despite the dependence on $\alpha n \le t \le n$. In the proof below, it
is necessary to be more explicit in our notation, so we denote $\til{x}_{u,v}$
as $\til{x}^{(t)}_{u,v}$ to indicate that we are working with an edge variable
from the relevant LP solution involving $G_t$. 

\begin{proof}[Proof of Lemma \ref{lem:availability_lower_bound}]

In what follows, let us assume that $ \alpha n \le t \le n$ is fixed. We wish to prove that for each $u \in U$, 
\[
\mb{P}[u \in R_{t}  \, | \, V_{t},v_t] \ge \frac{ \alpha n}{t-1}.
\]
As such, we must condition on $(V_t,v_t)$ throughout the remainder of the proof.
To simplify the argument, we abuse notation slightly and remove $(V_t,v_t)$
from the subsequent probability computations, though it is understood to implicitly
appear.

Given arriving node $v_{j}$ for $j=1, \ldots ,n$, once again denote $C(u,v_j)$
as the event in which $v_j$ commits to $u \in U$. As $R_{t}$ denotes the
unmatched nodes after the vertices $v_{1}, \ldots , v_{t-1}$ are processed
by Algorithm \ref{alg:ROM_edge_weights}, observe that
$u \in R_{t}$ if and only if $\neg C(u,v_j)$ occurs for each $j=1, \ldots , t-1$.
As a result,
\[
	\mb{P}[ u \in R_{t} ] = \mb{P}[\cap_{j=1}^{t-1} \neg C(u,v_j)].
\]
We therefore focus on lower bounding $\mb{P}[\cap_{j=1}^{t-1} \neg C(u,v_j) ]$ in order to prove
the lemma.

First observe that for $j=1, \ldots , \alpha n -1$, the algorithm passes on all the
trials of $v_j$ by definition. As such, we may focus on lower bounding
\[
	\mb{P}[\cap_{j= \alpha n}^{t-1} \neg C(u,v_j)],
\]
which depends only on the vertices of $V_{t-1} \setminus V_{\alpha n -1}$.
We denote $\bar{t}:= t - \alpha n$ as the number of vertices within this set.

Let us first consider the vertex $v_{t-1}$,
and the induced edge variable $\til{x}^{(t-1)}_{u,v}$ for each $v \in V_{t-1}$.
Observe that after applying Lemma \ref{lem:fixed_vertex_probe},
\begin{align*}
	\mb{P}[ C(u,v_{t-1})] &= \sum_{ v \in V_{t-1}} \mb{P}[ C(u,v_{t-1})  \, | \, v_{t-1} = v] \cdot \mb{P}[v_{t-1}=v] \\
											  &= \frac{1}{t-1} \sum_{v \in V_{t-1}} \til{x}_{u,v}^{(t-1)} \, p_{u,v},
\end{align*}
as once we condition on $(V_{t}, v_t)$, $v_{t-1}$ is uniformly distributed amongst $V_{t-1}$.
On the other hand, the values $(\til{x}_{u,v}^{(t-1)})_{u \in U, v \in V_{t-1}}$ are derived
from a solution to \ref{LP:new} for $G_{t-1}$, and so
\[
	\sum_{v \in V_{t-1}} \til{x}_{u,v}^{(t-1)} \, p_{u,v} \le 1.
\]
We therefore get that
\[
	\mb{P}[ C(u,v_{t-1}) ]  \le \frac{1}{t-1}.
\]
Similarly, if we fix $1 \le k \le \bar{t}$, then we can generalize the above
argument by conditioning
on the identities of all the vertices preceding $v_{t-k}$, as well as the probes
they make; that is, $(u_{t-1},v_{t-1}), \ldots ,(u_{t-(k-1)},v_{t-(k-1)})$ (in addition to $V_t$ and $v_t$ as always). 

In order to simplify the resulting indices, let us reorder the vertices of $V_{t-1} \setminus V_{\alpha n -1}$.
Specifically, define $\bar{v}_{k}:=v_{t-k}, \bar{u}_k := u_{t-k}$ and $\bar{e}_k:=e_{t-k}$
for $k=1,\ldots ,\bar{t}$. With this notation, we denote
$\scr{H}_{k}$ as encoding the information available based on the vertices $\bar{v}_1, \ldots , \bar{v}_{k}$ and the edges they (potentially) committed to, namely $\bar{e}_{1}, \ldots ,\bar{e}_{k}$\footnote{Formally, $\scr{H}_k$ is the sigma-algebra
generated from $V_t,v_t$ and $\bar{e}_{1}, \ldots ,\bar{e}_{k}$.}. By convention,
we define $\scr{H}_{0}$ as encoding the information regarding $V_t$ and $v_t$. 

An analogous computation to the above case then implies that
\[
	\mb{P}[ C(u,\bar{v}_{k})  \, | \, \scr{H}_{k-1}] = \sum_{ v \in V_{t-k}} \til{x}_{u,v}^{(t-k)} \, p_{u,v} \, \mb{P}[\bar{v}_k=v] \le \frac{1}{t-k},
\]
for each $k=1,\ldots , \bar{t}$, where
$\til{x}_{u,v}^{(t-k)}$ is the edge variable for $v \in V_{t-k}$.

Observe now that in each step, we condition on strictly more information;
that is, $\scr{H}_{k-1} \subseteq \scr{H}_{k}$
for each $k=2, \ldots , \bar{t}$. On the other hand, observe that
if we condition on $\scr{H}_{k-1}$ for $1 \le k \le \bar{t}-1$,
then the event $C(u,\bar{v}_{j})$ can be determined from $\scr{H}_{k-1}$
for each $1 \le j \le k-1$.

Using these observations, for $1 \le k \le \bar{t}$, the following recursion holds:
\begin{align*}
\mb{P}[ \cap_{j=1}^{k} \neg C(u,\bar{v}_j)] 
  &= \mb{E}\left[ \, \mb{E}\left[ \prod_{j=1}^{k} \bm{1}_{ [\neg C(u,\bar{v}_j)]}  \, | \, \scr{H}_{k-1}\right] \right]	\\
&= \mb{E}\left[ \,  \prod_{j=1}^{k-1} \bm{1}_{ [\neg C(u,\bar{v}_j)]} \, \mb{P}[\neg C(\bar{v}_k, u)  \, | \, \scr{H}_{k-1}]\right] \\
&\ge \left(1 - \frac{1}{t-k}\right) \mb{P}[ \cap_{j=1}^{k-1} \neg C(u,\bar{v}_j)]
\end{align*}

It follows that if $k= t - \alpha n$, then applying the above recursion implies that
\[
	\mb{P}[\cap_{j= \alpha n}^{t-1} \neg C(u,v_j)] \ge \prod_{k=1}^{t- \alpha n} \left(1 - \frac{1}{t-k}\right).
\]
Thus, after cancelling the pairwise products,
\[
	 \mb{P}[\cap_{j= \alpha n}^{t-1}\neg C(u,v_j)] \ge  \frac{\alpha n}{t-1},
\]
and so 
\[
		\mb{P}[ u \in R_{t}] = \mb{P}[\cap_{j= \alpha n}^{t-1}\neg C(u,v_j)] \ge \frac{\alpha n}{t-1}
\]

thereby completing the argument.

\end{proof}

\section{Conclusion and Open Problems}
\label{sec:conclusion}

We discussed the online stochastic matching problem in various settings and gave new and improved results with respect to a new LP relaxation which upper bounds the performance of the committal benchmark. We use our LP to create fractional solutions which can then be rounded to determine a non-adaptive sequence of edge probes. Our LP has a better stochasticity gap, as compared to the linear programs discussed in previous papers. We considered the ROM input model in the unknown stochastic graph setting, and adversarial, ROM and i.i.d. input models in the known stochastic graph setting. All of our results hold for arbitrary patience values and we consider both offline vertex weights and the more general  edge weights in determining the stochastic reward. 

Our results leave unsettled many interesting questions. We can view many open problems in terms of one basic issue: 
When (if ever) is there a provable difference between the classical online bipartite matching problem and the corresponding stochastic matching problem?    
What negative 
(i.e., inapproximation) results (if any) can be strengthened beyond what is known in the corresponding classical settings? 




One of the questions we have left open is whether our competitive ratios can be seen to hold against the
non-committal benchmark, or whether we must allow them more power. For instance, if our online probing algorithms execute without needing to respect commitment, is it clear that a competitive ratio of $1-1/e$ is attainable against the non-committal benchmark?
What if we enforce commitment, but allow our algorithms to execute adaptively? We believe that these open questions
highlight the difficulty of having to design probing algorithms which work for arbitrary patience constraints.

Another direction is to improve the linear program for the case of unit patience since a stochasticity gap of $1-1/e$ holds here,
and our linear program is equivalent to those linear programs discussed in earlier papers, when restricted to unit patience. This seems to be a bottleneck in proving positive results in any model. It would also be interesting to look at other methods to prove competitive ratios without using linear programs at all (i.e., by combinatorial methods). Or when (if ever) is $1 - 1/e$ an optimal competitive ratio? 

We are also interested in whether our results for stochastic matching can be extended so that offline, as well as online vertices, have patience constraints. Another extension would be to generalize the patience constraints so that now online vertices have budgets, and edges have non-uniform probing costs. The constraint would be  that the cost of probes adjacent to an online vertex is limited to its budget. And finally, we are interested in whether we can obtain improved competitive ratios, for special cases, such as when the edge probabilities are decomposable or vanishingly small as studied in Goyal and Udwani \cite{Goyal2020}.

\vspace{.3in}

\textbf{Acknowledgements}

We would like to thank Denis Pankratov for his helpful comments.

%
%
%

\bibliographystyle{plain}

{\footnotesize
\bibliography{bibliography}}

\appendix

\section{Relaxing the Committal Benchmark} \label{appendix:committal_benchmark}

In this section, we consider the committal benchmark,
as defined in Section \ref{sec:prelim}. In particular, we prove that \ref{LP:new}
is a relaxation of the committal benchmark (Theorem \ref{thm:new_LP_relaxation}).
It is  convenient to extend our  definition of an online probing algorithm
to the offline setting (as has previously been implicitly suggested by the committal and
non-committal benchmarks).

Suppose that we are given an arbitrary stochastic
graph $G=(U,V,E)$. We define an \textbf{(offline) probing algorithm} as an algorithm which adaptively reveals
the edge states $G$, while respecting the patience values of $G$. Notably, we do \textit{not}
restrict a probing algorithm to any specific ordering of the edges of $G$. The goal of a probing
algorithm is again to return a matching of active edges of large expected weight, though it
must \textbf{respect commitment}. That is, it has the property that if it makes a probe which
yields an active edge, then this edge must be included in the current matching (if possible).
The value of the committal benchmark on $G$, denoted $\OPT(G)$, simply corresponds to
the largest expected value a probing algorithm can attain on $G$.

Let us now restate \ref{LP:new} for convenience:
\begin{align}\label{LP:new_restatement}
\tag{LP-new}
&\text{maximize} &  \sum_{v \in V} \sum_{\bm{u} \in U^{( \le \ell_{v})}} \left( \sum_{i=1}^{|\bm{u}|} w_{u_i,v} \, g^{i}_{v}(\bm{u}) \right) \cdot x_{v}(\bm{u}) \\
&\text{subject to} & \sum_{v \in V} \sum_{i=1}^{\ell_v}\sum_{\substack{ \bm{u}^* \in U^{ ( \le \ell_{v})}: \\ u_{i}^*=u}} 
g_{v}^{i}(\bm{u}^*) \cdot x_{v}( \bm{u}^*)  \leq 1 && \forall u \in U  \label{eqn:relaxation_efficiency_matching_restatement}\\
&& \sum_{\bm{u} \in U^{( \le \ell_{v})}} x_{v}(\bm{u}) \le 1 && \forall v \in V,  \label{eqn:relaxation_efficiency_distribution_restatement} \\
&&x_{v}( \bm{u}) \ge 0 && \forall v \in V, \bm{u} \in U^{(\le \ell_v)}
\end{align}
In order to prove that \ref{LP:new_restatement} is a relaxation of the committal benchmark,
we must show that $\OPT(G) \le \LPOPT_{new}(G)$. Observe that in the above terminology,
this is equivalent to showing that for each (offline) probing algorithm $\scr{A}$, 
$\mb{E}[ \val(\scr{A}(G))] \le \LPOPT_{new}(G)$, where $\scr{A}(G)$ is the matching returned by $\scr{A}$.

Suppose now that we define $x_{v}(\bm{u})$ to be the probability that $\scr{A}$
probes the edges $(u_{i},v)_{i=1}^{|\bm{u}|}$ in order, where $v \in V$ and $\bm{u} \in U^{(\le \ell_v)}$. 
Let us suppose that $\scr{A}$ has the following property:
\begin{enumerate}
\item For each $v \in V$, the edge probes involving $v$ are made independently of the edge states $(\st(u,v))_{u \in U}$. \label{eqn:probing_algorithm_assumptions}
\end{enumerate}
Observe that by using \eqref{eqn:probing_algorithm_assumptions},
the expected value of the edge assigned to $v$ is
\[
	\sum_{\bm{u} \in U^{( \le \ell_{v})}} \left( \sum_{i=1}^{|\bm{u}|} w_{u_i,v} \cdot g^{i}_{v}(\bm{u}) \right) \cdot x_{v}(\bm{u}).
\] 
By additionally arguing that $(x_{v}(\bm{u}))_{v \in V, \bm{u} \in U^{( \le \ell_v}}$
is a feasible solution to \ref{LP:new_restatement}, we can then upper bound $\mb{E}[\val(\scr{A}(G)]$ by $\LPOPT_{new}(G)$.

That being said, since we make no assumption on how $\scr{A}$ moves between edge probes,
it is not clear that we can assume it satisfies \eqref{eqn:probing_algorithm_assumptions} without loss of generality. 
As such, the natural interpretation of the
variables of \ref{LP:new_restatement} does not seem to easily lend itself to a proof
of Theorem \ref{thm:new_LP_relaxation}.

In order to get around these issues, we introduce a new stochastic probing problem
for a stochastic graph $G=(U,V,E)$, known as the \textbf{relaxed stochastic matching problem}. This problem is closely related to the stochastic matching problem, however it places fewer restrictions on how many times each vertex $u \in U$
may be matched by a probing algorithm. As such, it has an optimum solution which upper bounds $\OPT(G)$. Interestingly, \ref{LP:new} exactly encodes this relaxed matching problem, which implies Theorem \ref{thm:new_LP_relaxation} as a corollary.

We now introduce the definition of a \textbf{relaxed probing algorithm}, described in the following way:

A relaxed probing algorithm probes edges of $G$,
while respecting the patience constraints of the online nodes of $V$.
It is allowed to arbitrarily move between the edges of $G$,
and it must return $\scr{M}$, a subset of its probes which yielded
active edges. The goal of
the relaxed probing algorithm is to maximize the expected weight of $\scr{M}$,
while ensuring that the following properties are satisfied:
\begin{enumerate}
\item Each $v \in V$ appears in at most one edge of $\scr{M}$.
\item For each $u \in U$, the \textit{expected} number of edges which contain
$u$ is at most one.
\end{enumerate}
We refer to $\scr{M}$ as a \textbf{one-sided matching} for the online nodes.
In a slight abuse of terminology, we say that a relaxed probing algorithm \textbf{matches}
the edge $e$, provided $e$ is included in $\scr{M}$.

A relaxed probing algorithm must \textbf{respect commitment}. That is, it has the property that if a probe to $e=(u,v)$ yields an active edge, then the edge is included in $\scr{M}$ (provided $v$ is currently not in $\scr{M}$). Observe that
this requires the relaxed probing algorithm to include $e$, even if $u$ is already adjacent to some
element of $\scr{M}$.

We define the \textbf{relaxed benchmark} as the optimum relaxed probing algorithm
on $G$, and denote $\OPT_{rel}(G)$ as the value this benchmark attains  $G$. Observe that by definition,
\[
	\OPT(G) \le \OPT_{rel}(G),
\]
where $\OPT(G)$ is the value of the committal benchmark on $G$.

Finally, we say that a relaxed probing algorithm is \textbf{non-adaptive},
provided its edge probes are statistically independent
from the edge states of $G$; that is, the random variables $(\st(e))_{e \in E}$. We
emphasize that this is equivalent to specifying an ordering $\lambda$ on a subset of $E$, where each vertex $v \in V$
appears in at most $\ell_v$ edges of $\lambda$. Once $\lambda$ is generated (potentially
using randomness), the edges specified by $\lambda$ are probed in order, and an edge
is added to the matching, provided its online node is unmatched.

Unlike the committal benchmark, $\OPT_{rel}(G)$ can be attained
by a non-adaptive relaxed probing algorithm.

\begin{theorem} \label{thm:non_adaptive_optimum}
There exists a relaxed probing algorithm which is non-adaptive
and attains value $\OPT_{rel}(G)$ in expectation.
\end{theorem}

We defer the proof of Theorem \ref{thm:non_adaptive_optimum} for now,
and instead show how it allows us to prove Theorem \ref{thm:new_LP_relaxation}.
In fact, we prove that \ref{LP:new_restatement} encodes the value
of the relaxed benchmark exactly, thus implying Theorem \ref{thm:new_LP_relaxation}
since $\OPT(G) \le \OPT_{rel}(G)$. 

\begin{theorem}

For any stochastic graph $G$, an optimum solution to \ref{LP:new_restatement} has value
equal to $\OPT_{rel}(G)$, the value of the relaxed benchmark on $G$.

\end{theorem}

\begin{proof}

Suppose we are presented a solution $(x_{v}(\bm{u}))_{v \in V, \bm{u} \in U^{ (\le \ell_v)}}$
to \ref{LP:new_restatement}. We observe then the following relaxed probing algorithm:
\begin{enumerate}
\item $\scr{M} \leftarrow \emptyset$.
\item For each $v \in V$, set $e \leftarrow \textsc{VertexProbe}(G, (x_{v}(\bm{u}))_{\bm{u} \in U^{ (\le \ell_v)}}, v).$
\item If $e \neq \emptyset$, then let $e=(u,v)$ and set $\scr{M}(v) = u$.
\item Return $\scr{M}$.
\end{enumerate}
Using Lemma \ref{lem:fixed_vertex_probe}, it is clear that
\[
	\mb{E}[ \val(\scr{M})] = \sum_{v \in V} \sum_{\bm{u} \in U^{( \le \ell_{v})}} \left( \sum_{i=1}^{|\bm{u}|} w_{u_i,v} \, g^{i}_{v}(\bm{u}) \right) \cdot x_{v}(\bm{u}).
\]
Moreover, each vertex $u \in U$ is matched by $\scr{M}$ at most once in expectation, as a consequence of
\eqref{eqn:relaxation_efficiency_matching_restatement}.

In order to complete the proof, it remains to show that if $\scr{A}$ is
an optimum relaxed probing algorithm,
then there exists a solution to \ref{LP:new_restatement} whose value is equal to
$\mb{E}[\val(\scr{A}(G))]$ (where $\scr{A}(G)$ is the one-sided matching returned by $\scr{A}$).
In fact, by Theorem \ref{thm:non_adaptive_optimum}, we may assume that $\scr{A}$
is non-adaptive.

Observe then that for each $v \in V$ and $\bm{u} =(u_{1}, \ldots ,u_{k}) \in U^{( \le \ell_v)}$ with $k=|\bm{u}|$ we can define
\[
	x_{v}(\bm{u}):= \mb{P}[\text{$\scr{A}$ probes the edges $(u_{i},v)_{i=1}^{k}$ in order}]. 
\]
Setting $\scr{M}= \scr{A}(G)$ for convenience, observe that 
if $\val(\scr{M}(v))$ corresponds to the weight of the edge
assigned to $v$ (which is $0$ if no assignment is made),
then
\[
	\mb{E}[ \val(\scr{M}(v))] = \sum_{\bm{u} \in U^{( \le \ell_{v})}} \left( \sum_{i=1}^{|\bm{u}|} w_{u_i,v} \, g^{i}_{v}(\bm{u}) \right) \cdot x_{v}(\bm{u}),
\]
as $\scr{A}$ is non-adaptive.

Moreover, for each $u \in U$,
\[
	\sum_{v \in V} \sum_{i=1}^{\ell_v} \sum_{\substack{ \bm{u}^* \in U^{ ( \le \ell_{v})}: \\ u_{i}^*=u}} 
g_{v}^{i}(\bm{u}^*) \cdot x_{v}( \bm{u}^*)  \leq 1,
\]
by once again using the non-adaptivity of $\scr{A}$. The proof is therefore complete.

\end{proof}

\subsection{Non-adaptivity in the Relaxed Stochastic Matching Problem}

We now argue that in the relaxed stochastic matching problem, the relaxed benchmark
does not require adaptivity. Our approach is to first generalize the LP used by Gamlath et al. \cite{Gamlath2019}
to the case of arbitrary patience. In doing so, we argue that this generalized
LP also encodes the relaxed stochastic matching problem. Moreover,
given an optimum solution to this LP, we can employ the techniques of 
Gamlath et al. \cite{Gamlath2019} and Costello et al. \cite{costello2012matching}
to recover an optimum probing algorithm which is non-adaptive. These observations
immediately imply Theorem \ref{thm:non_adaptive_optimum}, which allows
us to complete the proof of Theorem \ref{thm:new_LP_relaxation}.

Suppose that $G=(U,V,E)$ is an arbitrary stochastic graph.
For each $S \subseteq U$ and $v \in V$, we first define 
\[
	p(S,v):= 1 - \prod_{u \in S}(1 -p_{u,v}),
\]
which corresponds to the probability that an edge between $v$ and $S$ is active. Moreover,
for $1 \le k \le |U|$, denote $\binom{U}{k}$ as the collection of subsets of $U$ of size $k$,
and $\binom{U}{\le k} = \cup_{i=1}^{k} \binom{U}{i}$ as the collection of (non-empty) subsets of $U$
of size no greater than $k$.

For each $v \in V$ and $R \in \binom{U}{\le \ell_v}$,
we interpret the variable $\alpha_{v}(R)$ as the probability that the relaxed benchmark
probes the edges $\{v\} \times R$. Moreover, for $u \in U$,
we interpret the variable $z_{u,v}(R)$ as corresponding to the probability that the relaxed probing algorithm
probes the edges $\{v\} \times R$ \textit{and} matches the edge $(u,v)$. In this way, the edge probability $p_{u,v}$ is
implicitly encoded in this variable.
\begin{align}\label{LP:relaxed_benchmark}
\tag{LP-rel}
& \text{maximize}  & \sum_{u \in U, v \in V} \sum_{R \in \binom{U}{\le \ell_v}} w_{u,v} \cdot z_{u,v}(R) \\
&\text{subject to} &\, \sum_{v \in V} \sum_{R \in \binom{U}{\le \ell_v}} z_{u,v}(R) &\leq 1 &&\forall u \in U   \label{eqn:matching_constraint} \\
&  &\sum_{u \in S}  z_{u,v}(R) & \le  p(S,v) \cdot \alpha_{v}(R)  && \forall v \in V, \, R \in \binom{U}{\le \ell_v}, S \subseteq R	\label{eqn:costello_general_constraint}\\ 
&  & \sum_{R \in \binom{U}{\le \ell_v}}\alpha_{v}(R) & \le 1 && \forall v \in V \\
&  &z_{u,v}(R) & = 0   && \forall  u \in U \setminus R, v \in V, R \in \binom{U}{\le \ell_v}  \\
&  &z_{u,v}(R) & \ge 0	&& \forall u \in U, v \in V, R \in \binom{U}{\le \ell_v} \\
&  & \alpha_{v}(R) &\ge 0   && \forall v \in V, R \in \binom{U}{\le \ell_v} 
\end{align}

\begin{remark}
When $G=(U,V,E)$ has full patience, $\alpha_{v}(U)$ may be set to $1$ for each
$v \in V$. As a result, only the $z_{u,v}(U)$ variables become relevant 
for $u \in U$, and so \ref{LP:relaxed_benchmark} generalizes the LP 
of Gamlath et al. in \cite{Gamlath2019} (see \eqref{eqn:svensson_constraints} of Section \ref{sec:prelim}).
\end{remark}

If $\LPOPT_{rel}(G)$ corresponds to the optimum value of this LP, then this value
encodes $\OPT_{rel}(G)$ exactly. Moreover, an optimum solution to \ref{LP:relaxed_benchmark}
induces a relaxed probing algorithm which is non-adaptive:

\begin{theorem}\label{thm:relaxed_benchmark_characterization}

For any stochastic graph $G$, we have that
\[
	\LPOPT_{rel}(G) = \OPT_{rel}(G).
\]
Moreover, there exists a relaxed probing algorithm which is non-adaptive
and optimum, thereby proving Theorem \ref{thm:non_adaptive_optimum}.

\end{theorem}

Before proving this theorem, we first
state a key result from the work of Gamlath et al. \cite{Gamlath2019},
which helps motivate constraint \eqref{eqn:costello_general_constraint} of \ref{LP:relaxed_benchmark}.
We mention that an almost identical guarantee is also proven by Costello et al. in \cite{costello2012matching}.

\begin{theorem}[\cite{Gamlath2019}] \label{thm:costello_svensson_guarantee}
Suppose that $G=(U,V,E)$ is a stochastic graph, and $v \in V$ and $R \subseteq U$ are fixed.
Assume that there are non-negative values $(y_{u,v})_{u \in R}$, such that
\begin{equation} \label{eqn:svensson_requirement}
	\sum_{u \in S} y_{u,v} \le p(v,S)
\end{equation}
for each $S \subseteq R$. Under these assumptions,
there exists a non-adaptive and committal probing algorithm, say $\scr{B}_{v}(R)$,
which processes the single online node $v$ while executing on the stochastic sub-graph $G[\{v\} \cup R]$.
Moreover, it has the guarantee that
\[
	\mb{P}[\text{$\scr{B}_{v}(R)$ matches $v$ to $u$}] = y_{u,v}
\]
for each $u \in R$.
\end{theorem}

\begin{proof}[Proof of Theorem \ref{thm:relaxed_benchmark_characterization}]

We first argue that the relaxed benchmark corresponds to a
solution of \ref{LP:relaxed_benchmark}. Observe that since the relaxed benchmark
respects commitment, we may assume that for each $v \in V$, if a probe involving $v$
yields an active edge, then $v$ is never probed again.

For each $v \in V$ and $R \subseteq U$, where $k:=|R|$ satisfies $k \le \ell_v$,
define $\alpha_{v}(R)$ as the probability that the relaxed benchmark
probes the $k$ edges between $R$ and $v$. 

If we now take $u \in U$, then define $z_{u,v}(R)$ as the probability that the relaxed benchmark
probes the edges $R \times \{v\}$, \textit{and} matches $u$ to $v$.
For convenience, we also define $z_{u,v}$ as the probability that the relaxed benchmark matches
$u$ to $v$. We then get that
\[
	z_{u,v} = \sum_{R \in \binom{U}{\le \ell_v}} z_{u,v}(R)
\]
for each $u \in U$ and $v \in V$. If $\scr{M}$ corresponds to the one-sided matching
returned by the relaxed benchmark, then observe
that
\begin{equation}\label{eqn:relaxed_benchmark_expected_value}
	\mb{E}[ \val(\scr{M})] = \sum_{u \in U, v \in V} w_{u,v} \cdot z_{u,v}.
\end{equation}

We now claim that $(z_{u,v}(R))_{u \in U, v \in V, R \in \binom{U}{\le \ell_v}}$ together with
$(\alpha_{v}(R))_{v \in V, R \in \binom{U}{\le \ell_v}}$ corresponds to a feasible solution
to \ref{LP:relaxed_benchmark}. Given $v \in V$, $R \subseteq U$, and $S \subseteq R$,
we focus on proving that \eqref{eqn:costello_general_constraint} holds,
as the other constraints are easily seen to hold. For notational simplicity, we focus on the case of this constraint
when $S=R$, however the general case follows identically.

Let us define $E_{v}(R)$ as the event in which
the relaxed benchmark probes the edges $\{v\} \times R$.
Observe now that if the relaxed benchmark matches $v$ to
some vertex of $R$, then
one of the edges of $\{v\} \times R$
must have been active. As a result,
\begin{equation} \label{eqn:relaxed_benchmark_negative_correlation}
	\sum_{u \in R}  z_{u,v}(R) \le \mb{P}[E_{v}(R) \cap \bigcup_{u \in R} \{\st(u,v) =1\}]. 
\end{equation}
On the other hand, we claim
that the events $E_{v}(R)$ and $\cup_{u \in R} \{\st(u,v) =1\}$ are negatively
correlated. That is,
\begin{equation}\label{eqn:negative_correlation}
	\mb{P}[E_{v}(R) \cap \bigcup_{u \in R} \{\st(u,v) =1\}] \le \mb{P}[E_{v}(R)] \cdot \mb{P}[\cup_{u \in R} \{\st(u,v) =1\}] =
		\alpha_{v}(R) \cdot p(v,R),
\end{equation}
where $p(v,R)=  1 - \prod_{u \in R}(1 -p_{u,v})$.
To see this, for each $i \in [\ell_v]$ define $X_{v}^{i}$ as the offline vertex of
the $i^{th}$ edge involving $v$ which is probed by the relaxed benchmark. By convention, if no such vertex exists,
then $X_{v}^{i} := \emptyset$. Set $k= |R|$,
and define $R^{(k)}$ as the set of $k$-length tuples of $R$ whose
coordinates are all distinct. Observe that since the relaxed benchmark stops probing edges
adjacent to $v$ as soon as it witnesses an active edge involving $v$, $E_{v}(R) \cap \bigcup_{u \in R} \{\st(u,v) =1\}$
occurs if and only if there exists some $\bm{u} \in R^{(k)}$ such that
\[
	\{\text{$\st(u_{k},v) =1$ and $X_{v}^{k}  = u_k$}\} \cap \bigcap_{i=1}^{k-1} \{\text{$X_{v}^{i} = u_i$ and $\st(u_i,v) = 0$}\}.
\]
Moreover, for any $\bm{u} \in R^{(k)}$, 
the relaxed benchmark must decide whether to probe $(u_{k},v)$
before observing $\st(u_{k},v)$. Thus,
\[
	\mb{P}[\{\text{$\st(u_{k},v) =1$ and $X_{v}^{k}  = u_k$}\} \cap \bigcap_{i=1}^{k-1} \{ \text{$X_{v}^{i} = u_i$ and $\st(u_i,v) = 0$}\}]
\]
is equal to
\[
p_{u_{k},v} \cdot \mb{P}[\{X_{v}^{k}  = u_k \cap \bigcap_{i=1}^{k-1} \{\text{$X_{v}^{i} = u_i$ and $\st(u_i,v) = 0$}\}],
\]
which itself is upper bounded by
\[
p(v,R) \cdot \mb{P}[\{X_{v}^{k}  = u_k \} \cap \bigcap_{i=1}^{k-1} \{\text{$X_{v}^{i} = u_i$ and $\st(u_i,v) = 0$} \}].
\]
Thus, \eqref{eqn:negative_correlation} holds after summing over all $\bm{u} \in R^{(k)}$. 
As such, combined with \eqref{eqn:relaxed_benchmark_negative_correlation},
it follows that constraint \eqref{eqn:costello_general_constraint} must be satisfied. 
Since the remaining constraints of \ref{LP:relaxed_benchmark},
are easily seen to hold, we may apply \eqref{eqn:relaxed_benchmark_expected_value},
to conclude that $\OPT_{rel}(G) \le \LPOPT_{rel}(G)$.

Let us now suppose that we are presented a solution to \ref{LP:relaxed_benchmark}, which we
denote by $(z_{u,v}(R))_{u \in U, v \in V, R \in \binom{U}{\le \ell_v}}$ and
$(\alpha_{v}(R))_{v \in V, R \in \binom{U}{ \le \ell_v}}$. Using this
solution, we can derive relaxed probing algorithm which returns a
one-sided matching whose expected value is equal to the LP solution's value.

Let us first fix a vertex  $v \in V$
and consider the values $(\alpha_{v}(R))_{R \in \binom{U}{\le \ell_v}}$
and $(z_{u,v}(R))_{u \in U, R \in\binom{U}{\le \ell_v}}$, where $\alpha_{v}(R) \neq 0$.
Observe that if we fix $R \subseteq U$, $1 \le |R| \le \ell_v$,
then the values $(z_{u,v}(R)/\alpha_{v}(R))_{u \in R}$ satisfy the relevant
inequalities of \eqref{eqn:svensson_requirement}
thanks to constraint \eqref{eqn:costello_general_constraint} of \ref{LP:relaxed_benchmark}. 
Theorem \ref{thm:costello_svensson_guarantee} thus guarantees that there exists a committal probing algorithm for $G[ R \cup \{v\}]$,
say $\scr{B}_{v}(R)$, such that
\[
	\mb{P}[\text{$\scr{B}_{v}(R)$ matches $v$ to $u$}] = \frac{z_{u,v}(R)}{\alpha_{v}(R)}
\]
for each $u \in U$. Moreover, $\scr{B}_{v}(R)$ is non-adaptive;
that is, the probes of $\scr{B}_{v}(R)$ are statistically independent from the edge states, $(\st(u,v))_{u \in R}$.

This suggests the following relaxed probing algorithm, which we denote by $\scr{B}$:

\begin{enumerate}
\item Set $\scr{M} = \emptyset$.
\item For each $v \in V$, pass on $v$ with probability $1 - \sum_{R \in \binom{U}{\le \ell_v}} \alpha_{v}(R)$
\item Otherwise, draw $P \subseteq U$ with probability  $\alpha_{v}(P)$.
\item Execute $\scr{B}_{v}(P)$, and match $v$ to whichever vertex of $U$ (if any) $v$ is matched to by $\scr{B}_{v}(P)$. 
\end{enumerate}

Observe now that
\[
	\mb{P}[\text{$\scr{B}$ matches $v$ to $u$}] = \sum_{R \in \binom{U}{\le \ell_v}} \alpha_{v}(R) \cdot \frac{z_{u,v}(R)}{\alpha_{v}(R)} = z_{u,v},
\]
Thus,
\[
	\mb{E}[ \val(\scr{M})] = \sum_{u \in U, v \in V} w_{u,v} \cdot z_{u,v}.
\]
Moreover, if $N_{u}$ counts the number of vertices of $V$ which match to $u$,
then
\[
	\mb{E}[N_{u}] = \sum_{v \in V} z_{u,v} \le 1,
\]
as the values $(z_{u,v})_{v \in V}$ satisfy \eqref{eqn:matching_constraint} by assumption.

Finally, we observe that $\scr{B}$ respects commitment and is non-adaptive.

\end{proof}

We remark that \ref{LP:relaxed_benchmark}, coupled with the probing procedure
of Theorem \ref{thm:costello_svensson_guarantee}, can
be used to construct online probing algorithms with the same competitive guarantees as derived using \ref{LP:new}
and \textsc{VertexProbe}. That being said, \ref{LP:relaxed_benchmark} does not seem to be
poly-time solvable, at least for arbitrary patience\footnote{If $\max_{v \in V} \ell_v$
is upper bounded by a constant, independent of the size of $U$, then the LP is solvable using the separation
oracle presented by Gamlath et al. \cite{Gamlath2019}. A similar statement is true if all the
patience values are close to $|U|$.}. Introducing \ref{LP:new} allows us to derive online probing
algorithms which are poly-time solvable and which do not require appealing to the subroutine
involved in Theorem \ref{thm:costello_svensson_guarantee}

\section{Solving \ref{LP:new} Efficiently} \label{appendix:efficient_probing_algorithms}
In this section, we prove the following result:

\begin{theorem} \label{thm:new_LP_efficient}

An optimum solution to \ref{LP:new} can be found in polynomial time
in the size of the stochastic graph, $G=(U,V,E)$.

\end{theorem}

In order to prove this claim, it suffices to show that \ref{LP:new_dual} has a (deterministic) polynomial time separation
oracle, as a consequence of how the ellipsoid algorithm \cite{Groetschel,GartnerM} executes (see \cite{Williamson,Adamczyk2017,Lee2018} for applications). As such, we restate the dual of \ref{LP:new} for convenience:

\begin{align}\label{LP:new_dual_restatement}
\tag{LP-new-dual}
&\text{minimize} &  \sum_{u \in U} \alpha_{u} + \sum_{v \in V} \beta_{v}  \\
&\text{subject to} & \beta_{v} + \sum_{j=1}^{|\bm{u}^*|} g_{v}^{j}(\bm{u}^*) \cdot \alpha_{u_j^*} \ge \sum_{j=1}^{|\bm{u}^*|} w_{u_j^*,v} \cdot g_{v}^{j}( \bm{u}^*) &&
\forall v \in V, \bm{u}^* \in U^{ ( \le \ell_v)} \\
&&  \alpha_{u} \ge 0 && \forall u \in U\\
&& \beta_{v} \ge 0 && \forall v \in V
\end{align}
Suppose now that we are presented a particular selection dual variables,   
say $( (\alpha_{u})_{u \in U}, (\beta_{v})_{v \in V})$, which may or may not
be a feasible solution to \ref{LP:new_dual_restatement}. Our
separation oracle must determine efficiently whether these variables
satisfy all the constraints of \ref{LP:new_dual_restatement}. In the case
in which the solution is \textit{infeasible}, the oracle must additionally
return a constraint which is violated.

It is clear that we can accomplish this for the non-negativity constraints,
so we hereby assume that $\alpha_{u} \ge 0$ and $\beta_{v} \ge 0$ for all $u \in U$
and $v \in V$.

Let us now fix a particular $v \in V$ in what follows. 
We wish to determine whether there exists some
$\bm{u}^* \in U^{( \le \ell_v)}$  such that
\[
\beta_{v} + \sum_{j=1}^{|\bm{u}^*|} g_{v}^{j}(\bm{u}^*) \cdot \alpha_{u_j^*} < \sum_{j=1}^{|\bm{u}^*|} w_{u_j^*,v} \cdot g_{v}^{j}( \bm{u}^*).
\]
To make such a determination, we consider the function $\phi$, where
\begin{equation}\label{eqn:oracle_optimization}
	\phi(\bm{u}^{*}) :=\sum_{j=1}^{|\bm{u}^*|} ( w_{u_j^*,v} - \alpha_{u_{j}^*} )  \cdot g_{v}^{j}(\bm{u}^*),
\end{equation}

for $\bm{u}^* \in U^{( \le \ell_v)}$.
Our goal is to verify whether
there exists some $\bm{u}^* \in U^{( \le \ell_v)}$ such that $\phi(\bm{u}^*) > \beta_{v}$.  
If we can efficiently check this for a fixed $v \in V$,
then we can iterate the same procedure for all $v \in V$, thus yielding a polynomial
time separation oracle for \ref{LP:new_dual_restatement}. Thus, in order
to prove Theorem \ref{thm:new_LP_efficient}, we only need to prove the following
statement:

\begin{proposition}\label{prop:oracle_optimization}
There exists an efficient deterministic algorithm which checks whether there exists some $\bm{u}^* \in U^{( \le \ell_v)}$ such that $\phi(\bm{u}^{*}) > \beta_{v}$. Moreover, if a tuple with this property exists, then this algorithm will return such a tuple in polynomial time.

\end{proposition}

In \cite{Brubach2019}, Brubach et al. consider the setting 
in which one is presented \textit{non-negative} edge weights $(\bar{w}_{u^*,v})_{u^{*} \in U}$,
and the function
\begin{equation}\label{eqn:brubach_optimization}
	\psi(\bm{u}^*) := \sum_{j=1}^{|\bm{u}^*|} \bar{w}_{u_j^*,v} \cdot g_{v}^{j}( \bm{u}^*),
\end{equation}

for $\bm{u}^* \in U^{(\le \ell_v)}$. They show that one can maximize this function in polynomial
time using a deterministic algorithm based on dynamical programming techniques.

\begin{theorem}[\cite{Brubach2019}] \label{thm:brubach_optimization}

For any  $v \in V$ with patience $\ell_v$ and non-negative selection of weights, $(\bar{w}_{u^*,v})_{u^{*} \in U}$,
the function $\psi= \psi(\bm{u}^*)$ in \eqref{eqn:brubach_optimization} can be maximized in polynomial time using
a deterministic procedure.

\end{theorem}

We can apply Theorem \ref{thm:brubach_optimization} to prove Proposition \ref{prop:oracle_optimization}.

\begin{proof}[Proof of Proposition \ref{prop:oracle_optimization}]

Let us first define $\bar{w}_{u,v} := w_{u,v} - \alpha_{u}$ for all $u \in U$.
Denote $P$ as those $u \in U$ such that $\bar{w}_{u,v} \ge 0$.
First note that if $P$ is empty, then clearly $\phi(\bm{u}^*) \le 0$
for all $\bm{u} \in  U^{( \le \ell_v)}$, so since $\beta_{v} \ge 0$ by assumption,
there is nothing to prove.

Let us therefore assume that $P \neq \emptyset$. Observe then that
we can restrict our attention to those $\bm{u}^* \in U^{ ( \le \ell_v)}$
whose entries all lie in $P$; namely, $P^{(\le \ell_v)}$. By applying Theorem \ref{thm:brubach_optimization},
we are guaranteed a deterministic procedure
for maximizing $\phi$ on $P^{(\le \ell_v)}$ in polynomial time. 
Let us denote the outcome of this procedure by $\bm{u}_{max}$. Observe then that 
either
\[
	\beta_{v} \ge \phi(\bm{u}_{max}) \ge  \phi(\bm{u}^*)
\]
for all $\bm{u}^* \in U^{ ( \le \ell_v)}$, or
$\bm{u}_{max}$ satisfies 
\[
	\phi(\bm{u}_{max}) > \beta_{v}.
\]
In either case, the procedure satisfies the requirements of Proposition \ref{prop:oracle_optimization},
and so the proof is complete.

\end{proof}

We conclude the section by considering the stochastic known i.i.d. setting.
Specifically, consider a stochastic type 
graph, say $G=(U,V,E)$, with $n$ arrivals draw from the fractional rates, $\bm{r}=(r_{v})_{v \in V}$.
In this case, we can solve \ref{LP:new_iid} efficiently
by presenting a separation oracle for its dual. The reduces to the same
maximization problem just considered, with the caveat that for a fixed
type node $v \in V$, we compare the maximized value to $\beta_{v} \cdot r_{v}$
(where $\beta_v$ is the dual variable associated to $v \in V$).

\section{The Non-committal Benchmark} \label{appendix:non_committal_benchmark}

In this section, we extend our definition of a probing algorithm to general (i.e., not necessarily
bipartite) stochastic graphs, as well as our definitions
of the committal and non-committal benchmarks. We then review \ref{LP:standard_definition}, and show that 
it is in fact a relaxation of the non-committal benchmark. This allows us to also prove 
that \ref{LP:standard_definition_iid} is a relaxation of the non-committal benchmark,
as defined for the stochastic known i.i.d. matching problem. As a corollary, we argue
that many of the results in the stochastic matching literature hold against this stronger benchmark.

We then discuss \ref{LP:new}, and discuss the restricted settings in which is it also
a relaxation of the non-committal benchmark (when once again restricted to online
stochastic matching problems). In the case of arbitrary patience values, we show
that \ref{LP:new} is \textit{not} a relaxation of the non-committal benchmark,
and discuss the limitations of our techniques. In particular, we explain the difficulty
in designing probing algorithms which respect commitment and also
attain large competitive ratios against the non-committal benchmark.

\subsection{Relaxing the Non-committal Benchmark via \ref{LP:standard_definition}}

Let us suppose that $G=(V,E)$ is a general stochastic graph
with edge weights $(w_{e})_{e \in E}$, edge probabilities
$(p_{e})_{e \in E}$ and patience values $(\ell_{v})_{v \in V}$.
We emphasize that $G$ need not be bipartite.

An \textbf{(offline) probing algorithm} must satisfy the requirement
that most $\ell_v$ probes are made to the neighbouring
edges of $v$ for each $v \in V$. We say that a probing algorithm is
\textbf{committal}, provided it satisfies the following
property when constructing its matching: if an edge $e \in E$ is probed, then $e=(u,v)$
must be added to the current matching, provided $u$ and $v$ are currently unmatched.
Alternatively, a probing algorithm is said to be \textbf{non-committal}, provided
the matching it constructs is done in the following manner: if the algorithm
(adaptively) reveals the edges $B \subseteq E$ to be active, then an
optimum matching constructed from $B$ is returned. 

We can also extend our definitions of the committal and non-committal benchmarks to $G$.
Specifically, the committal benchmark corresponds to the optimum committal probing algorithm, 
whose expected value we denote by $\OPT(G)$. Similarly,
the non-committal benchmark corresponds to the optimum non-committal probing algorithm,
whose expected value we denote by $\OPT_{non}(G)$.

We now generalize \ref{LP:standard_definition} to the non-bipartite
case, as originally presented in \cite{BansalGLMNR12} by Bansal et al.:

\begin{align}\label{LP:standard_general}
\tag{LP-std}
& \text{maximize}  & \sum_{e \in E} w_{e} \cdot p_{e} \cdot x_e \\
&\text{subject to} &\, \sum_{\substack{ e \in E: \\ v \in e}}  x_{e} &\leq \ell_v &&\forall v \in V\\
&  &\sum_{\substack{ e \in E: \\ v \in e}} p_{e} \cdot x_{e}  &\leq 1 && \forall v \in V\\	\\
&  &0 \leq x_{e} &\leq 1 && \forall e \in E
\end{align}

Bansal et al. showed that \ref{LP:standard_general} is a relaxation of the committal
benchmark; that is, $\OPT(G) \le \LPOPT_{std}(G)$, where $\LPOPT_{std}(G)$ denotes the
optimum value of \ref{LP:standard_general}.
We claim that this \ref{LP:standard_general} is also a relaxation
of the non-committal benchmark:

\begin{theorem}\label{thm:standard_LP_non_committal_relaxation}

For any (general) stochastic graph $G=(V,E)$, it holds that
\[
	\OPT_{non}(G) \le \LPOPT_{std}(G).
\]
\end{theorem}

As a result, if an (offline) probing algorithm attains an approximation ratio of $0 \le c \le 1$ against \ref{LP:standard_general}, then it attains this approximation ratio against \textit{non-committal}
benchmark as well. Notably, this implies that the competitive ratios
of the probing algorithms considered in \cite{BansalGLMNR12, Adamczyk15, BavejaBCNSX18} all in fact hold against the 
non-committal benchmark.

We can also consider the known i.i.d. setting, in which we are presented
a known i.i.d. instance $(G, \bm{r},n)$. In this case, an analogous statement of
Theorem \ref{thm:standard_LP_non_committal_relaxation} holds regarding \ref{LP:standard_definition_iid}:

\begin{theorem}\label{thm:standard_LP_non_committal_relaxation_iid}

For any input $(G, \bm{r}, n)$ of the known i.i.d. stochastic matching problem,
\[
	\OPT_{non}(G,\bm{r},n) \le \LPOPT_{std-iid}(G, \bm{r},n).
\]
\end{theorem}

Since the works of \cite{BansalGLMNR12, Adamczyk15, BrubachSSX16,BrubachSSX20} all prove competitive ratios
against \ref{LP:standard_definition_iid}, Theorem \ref{thm:standard_LP_non_committal_relaxation_iid} implies
that these competitive ratios all hold against the non-committal benchmark.

In the remainder of the section, we prove Theorem \ref{thm:standard_LP_non_committal_relaxation},
which we then argue can be used to imply Theorem \ref{thm:standard_LP_non_committal_relaxation_iid}.
It will be convenient to instead work
with a modified version of \ref{LP:standard_general},
where each edge $e \in E$ is instead associated with two variables, namely
$x_{e}$ and $z_{e}$. We interpret the former variable as the probability
that the non-committal benchmark probes the edge $e$, whereas the latter
variable corresponds to the probability that $e$ is included
in the matching constructed by the non-committal benchmark.

\begin{align}\label{LP:standard_general_non}
\tag{LP-std-non}
& \text{maximize}  & \sum_{e \in E} w_{e} \cdot z_{e} \\
&\text{subject to} &\, \sum_{\substack{ e \in E \\ v \in e}}  z_{e} &\leq 1 &&\forall v \in V\\
&  &\sum_{\substack{ e \in E \\ v \in e}} x_{u,v}  &\leq \ell_v && \forall v \in V\\
&	&z_{e} &\le p_{e} \cdot x_{e} && \forall e \in E	\\
&  &x_{e} & \le 1	&& \forall e \in E	\\
&  &x_{e}, z_{e} &\ge 0 && \forall e \in E
\end{align}
We denote $\LPOPT_{std-non}(G)$ as the value of an optimum solution to \ref{LP:standard_general_non}.
It turns out that \ref{LP:standard_general} and \ref{LP:standard_general_non} take the same optimum
value, no matter the stochastic graph $G$:

\begin{lemma}\label{lem:LP_equivalence}

For any stochastic graph $G=(V,E)$,
\[
	\LPOPT_{std}(G) = \LPOPT_{std-non}(G).
\]
\end{lemma}

\begin{proof}

Suppose we are presented a solution $(x_{e})_{e \in E}$ to \ref{LP:standard_general}.
In this case, if $z_{e}:= p_{e} \cdot x_{e}$ for $e \in E$, then
$(x_{e},z_{e})_{e \in E}$ is clearly a feasible solution to \ref{LP:standard_general_non}.
As such,
\[
	\LPOPT_{std}(G) \le \LPOPT_{std-non}(G).
\]
On the other hand, suppose that $(x_{e},z_{e})_{e \in E}$ is now an arbitrary
solution to \ref{LP:standard_general_non}. In this case, define
$\bar{x}_{e} := z_{e}/p_{e}$ for each $e \in E$. We claim that
$(\bar{x}_{e})_{e \in E}$ is a feasible solution to \ref{LP:standard_general}.

To see this, first observe that since $z_{e} \le p_{e} \cdot x_{e}$,
we know that $\bar{x}_{e} \le x_{e} \le 1$ for all $e \in E$.

Moreover, for each $v \in V$,
\[
	\sum_{\substack{e \in E: \\ v \in e}} p_{e} \cdot \bar{x}_{e} = \sum_{\substack{e \in E: \\ v \in e}} z_{e} \le 1,
\]
and
\[
	\sum_{\substack{e \in E: \\ v \in e}} \bar{x}_{e} \le \sum_{\substack{e \in E: \\ v \in e}} x_{e} \le \ell_v.
\]
Thus, $(\bar{x}_e)_{e \in E}$ is a feasible solution to \ref{LP:standard_general}.

Finally, observe that
\[
	\sum_{e \in E} w_{e} \cdot z_{e} = \sum_{e \in E} p_{e} \cdot w_{e} \cdot \bar{x}_e,
\]
so $\LPOPT_{std-non}(G) \le \LPOPT_{std}(G)$.

\end{proof}

\begin{proof}[Proof of Theorem \ref{thm:standard_LP_non_committal_relaxation}]

Let us suppose that $\scr{M}$
is the matching returned when the non-committal benchmark executes on $G=(V,E)$.
If we fix $e \in E$,
then we can define $x_{e}$ as the probability $\scr{A}$ probes the edge $e$, and $z_{e}$ as the probability that it includes $e$ in $\scr{M}$. Observe then that
\[
	\OPT_{non}(G) = \mb{E}[ \val(\scr{M}) ] = \sum_{e \in E} w_{e} \cdot z_{e}.
\]
Now, if we can show that $(x_{e},z_{e})_{e \in E}$ is a feasible solution to
\ref{LP:standard_general_non}, then this will imply that
\[
	\OPT_{non}(G) \le \LPOPT_{std-non}(G) = \LPOPT_{std}(G),
\]
where the final line follows via an application of Lemma \ref{lem:LP_equivalence}.
Thus, in order to complete the proof it suffices to show that $(x_{e},z_{e})_{e \in E}$
is a feasible solution to \ref{LP:standard_general_non}.

Suppose now that we fix vertex $v \in V$.
Observe that since it is matched to at most one edge of $G$, 
we have that
\[
	\sum_{\substack{e \in E \\ v \in e}} z_{e} \le 1.
\]
Similarly, 
\[
	\sum_{\substack{e \in E \\ v \in e}} x_{e} \le \ell_v,
\]
as at most $\ell_v$ edges including $v$ are probed by the non-committal benchmark.

If we now fix an edge $e \in E$, then observe that
in order for $e$ to be included in $\scr{M}$ , $e$ must be probed \textit{and} $e$ must be active. On the other
hand, these two events occur independently of each other. As such,
\[
z_{e} \le p_{e} \cdot z_{e}.
\]
This shows that all the constraints of \ref{LP:standard_general_non} hold
for $(x_{e},z_{e})_{e \in E}$, and so the proof is complete.

\end{proof}

We conclude the section by observing that Theorem \ref{thm:standard_LP_non_committal_relaxation_iid}
can be seen to hold by applying Theorem \ref{thm:standard_LP_non_committal_relaxation} 
to the instantiated graph $\hat{G} \sim (G, \bm{r},n)$,
and using a standard conditioning argument (see the proof of Lemma \ref{lem:LP_validity_known_iid} from
Section \ref{sec:stochastic_known_iid} for details).

\subsection{\ref{LP:new} and the Non-committal Benchmark}

We begin by observing that since \ref{LP:new} and \ref{LP:standard_definition} are the same LP when $G=(U,V,E)$ is
bipartite and has unit patience values on $V$, Theorem \ref{thm:standard_LP_non_committal_relaxation}
implies that the competitive ratios in Theorems \ref{thm:known_vertex_weights}, \ref{thm:known_ROM_edge_weights}
and \ref{thm:ROM_edge_weights} all hold against the \textit{non-committal} benchmark, for the special
case of unit patience. Moreover, Theorem \ref{thm:standard_LP_non_committal_relaxation_iid} implies
that the competitive ratio of Theorem \ref{thm:iid} holds against the non-committal benchmark 
when the type graph has unit patience.

In the case when $G$ has full patience, observe that \ref{LP:new} takes the same value
as the LP in \cite{Gamlath2019} considered by Gamlath et al. (see \ref{LP:relaxed_benchmark}
in Appendix \ref{appendix:committal_benchmark} for details). On the other hand,
Gamlath et al. argue that an optimum solution to their LP upper bounds the expected
weight of an optimum matching of $G$. As $G$ has full patience,
this corresponds to the performance of the non-committal benchmark on $G$, and so
\ref{LP:new} is a relaxation of the non-committal benchmark for this special
case. The competitive ratios
of Theorems \ref{thm:known_vertex_weights}, \ref{thm:known_ROM_edge_weights} and \ref{thm:ROM_edge_weights}
thus all hold against the non-committal benchmark for full patience.
When $G$ is a type graph with full patience, the same techniques used in Lemma \ref{lem:LP_validity_known_iid}
from Section \ref{sec:stochastic_known_iid} can be used to show that
\ref{LP:new_iid} upper bounds the non-committal benchmark. 
This implies that Theorem \ref{thm:iid} also holds against the non-committal benchmark
in this setting.

While our competitive ratios carry over in these specific scenarios, consider now the
setting when $G$ has a single online node $v$. In this case, it is clear that \ref{LP:new}
exactly encodes the committal benchmark; that is, $\LPOPT_{new}(G) = \OPT(G)$. Moreover,
$\OPT(G)$ can be attained via a non-adaptive strategy. 
On the other hand, the edges probabilities and edge weights of $G$ can be chosen
in such a way that the non-committal benchmarks gains strictly more
power than the committal benchmark; that is, $\OPT(G) < \OPT_{non}(G)$. We remark that these
problems are a special case of \textsc{ProblemMax}, a stochastic probing problem
which is studied in \cite{Asadpour2016, Fu2018, Segev2020}.

Let us suppose that $v$ has patience $\ell_{v} =2$, and that there are $3$ offline nodes $U=\{u_{1}, u_{2}, u_{3}\}$.
For each $i \in \{1,2,3\}$, we denote the weight of $(u_{i},v)$ by $w_{i}$ and assume that the edge $(u_i,v)$
is active with probability $p_{i}$. We make the following assumptions on these weights
and probabilities:

\begin{enumerate}
\item $w_{1} < w_{2} < w_{3}$.
\item $p_{1} > p_{2} > p_{3}$.
\item $w_{1} \cdot p_{1} \ge w_{2} \cdot p_{2} > w_{3} \cdot p_{3}$. \label{eqn:expected_reward}
\item $p_{2} \cdot w_{2} - p_{3} \cdot w_{3} \ge p_{1} \cdot w_{1} \cdot (p_{2} - p_{3})$ \label{eqn:OPT_strategy}
\end{enumerate}
Clearly, there exists a choice of weights and probabilities which satisfy these constraints. For instance,
take $w_{1} =3$, $w_{2} =4$, $w_{3}= 98$, $p_{1}=0.8$, $p_{2}=0.6$, and $p_{3} =0.01$.

Based on these assumptions let us now consider the value of $\OPT(G)$. Observe that
\[
	p_{2} \cdot w_{2} + (1 - p_{2}) \cdot p_{1} \cdot w_{1} \ge p_{3} \cdot w_{3} + (1- p_{3}) \cdot p_{1} \cdot w_{1}
															\ge p_{3} \cdot w_{3} + (1 - p_{3}) \cdot p_{2} \cdot w_{2},
\]
where the first inequality follows from \eqref{eqn:OPT_strategy}, and the second follows from \eqref{eqn:expected_reward}.
As a result, it is clear to see that the committal benchmark corresponds to probing $u_{2}$ and then $u_{1}$ (if necessary);
thus, $\OPT(G) = p_{2} \cdot w_{2} + (1 - p_{2}) \cdot p_{1} \cdot w_{1}$.

On the other hand, let us consider $\OPT_{non}(G)$, the value of the non-committal benchmark on $G$.
Consider the following \textit{non-committal} probing algorithm:
\begin{itemize}
\item Probe $(u_{2},v)$, and if $\st(u_{2},v)=1$, probe $(u_{3},v)$. \label{eqn:adaptive_step}
\item Else if $\st(u_{2},v)=0$, probe $(u_{1},v)$.
\item Return the edge of highest weight which is active (if any).
\end{itemize}
Clearly, this probing algorithm uses adaptivity to decide whether to reveal $(u_{3},v)$ or
$(u_{1},v)$ in its second probe. Specifically, if it discovers that $(u_{2},v)$ is active,
then it knows that it will return an edge with weight at least $w_{2}$. As such, it only
makes sense for $(u_{3},v)$ to be its next probe, as $w_{3} > w_{2} > w_{1}$.
On the other hand, if $(u_{2},v)$ is discovered to be inactive, it makes sense to
prioritize probing the edge $(u_{1},v)$ over $(u_{3},v)$, as the \textit{expected}
reward is higher; namely, $w_{1} \cdot p_{1} > w_{3} \cdot p_{3}$.

The expected value of the edge returned is 
\[
p_{2} \cdot p_{3} \cdot w_{3} + p_{2} \cdot (1 - p_{3}) \cdot w_{2} + (1 - p_{2}) \cdot p_{1} \cdot w_{1}.
\] 
Observe however that
\begin{align*}
	p_{2} \cdot p_{3} \cdot w_{3} + p_{2} \cdot (1 - p_{3}) \cdot w_{2} + (1 - p_{2}) \cdot p_{1} \cdot w_{1} &= p_{2} \cdot w_{2} \cdot ( (1- p_{3}) + p_{3} \cdot w_{3}/ w_{2}) + (1- p_{2}) \cdot p_{1} \cdot w_{1}	\\
				  &> p_{2} \cdot w_{2} + (1 - p_{2}) \cdot p_{1} \cdot w_{1}	\\
				  &= \OPT(G),
\end{align*}

where the final inequality follows since $w_{3} > w_{2}$.
As a result, it is clear that this strategy corresponds to the non-committal benchmark, and so 
$\OPT_{non}(G) > \OPT(G)$. In fact,
for the specific choice when $w_{1} =3$, $w_{2} =4$, $w_{3}= 98$, $p_{1}=0.8$, $p_{2}=0.6$, and $p_{3} =0.01$,
it holds that
\[
 	\frac{\LPOPT_{new}(G)}{\OPT_{non}(G)} =  \frac{\OPT(G)}{\OPT_{non}(G)} = 0.856269. 
\]
This example illustrates that since our competitive guarantees were proven against
$\LPOPT_{new}(G)$, they do not immediately extend to the non-committal benchmark. 
Moreover, this ratio improves upon the negative
result of \cite{costello2012matching}, in which the authors present
an example where the ratio between $\OPT(G)$ and $\OPT_{non}(G)$ is
at most $0.898$.

\section{LP Relations} \label{appendix:LP_relations}
In this section, we show how a number of the LPs present in the literature are related
to each other\footnote{We do not consider the full patience LP of Gamlath et al. in this section,
as we discuss a generalization of their LP in Appendix \ref{appendix:non_committal_benchmark}.}. In particular, we consider an LP introduced in by Brubach et al. \cite{Brubach2019},
which assumes a number of extra constraints in addition to those of \ref{LP:standard_definition}. We review the
motivation behind this LP, as well as how it is derived.

For each subset $R \subseteq U$ and $v \in V$, consider the induced stochastic subgraph,
denoted $G[\{v\} \cup R]$, formed by restricting the vertices of $G$ to $\{v\} \cup R$,
and the edges of $G$ to those between $v$ and $R$. 
We hereby denote $\OPT(v,R)$ as the value of the committal benchmark on the induced stochastic graph $G[\{v\} \cup R]$.

We can now formulate the LP of \cite{Brubach2019}, whose constraints ensure 
that for each $v \in V$, the expected stochastic reward
of $v$, suggested by an LP solution, is actually attainable by the committal benchmark. 

\begin{align}\label{LP:brubach_benchmark}
\tag{LP-DP}
& \text{maximize}  & \sum_{u \in U} \sum_{v \in V} w_{u,v} \, p_{u, v}  \, x_{u, v}\\
&\text{subject to} &\, \sum_{v \in V}  p_{u, v} \, x_{u, v} &\leq 1 &&\forall u \in U\\
&  &\sum_{u \in U} x_{u, v} &\leq \ell_{v} && \forall v \in V\\
&  &\sum_{u \in U}  p_{u, v} \, x_{u, v}  &\leq 1 && \forall v \in V\\
&  &\sum_{u \in R} w_{u,v} \, p_{u,v} \, x_{u,v} & \le  \OPT(v,R)  && \forall v \in V, \, R \subseteq U	\label{eqn:brubach_constraint}\\ 
&  &0 \leq x_{u, v} &\leq 1 && \forall u \in U, v \in V
\end{align}
Observe the following relations between the LPs considered in the paper.

\begin{theorem} \label{thm:relaxations_relations}

For any stochastic graph $G$, we have that
\begin{equation}
\LPOPT_{new}(G) \le LPOPT_{DP}(G) \le LPOPT_{std}(G),
\end{equation}

and the LPs are all equivalent, provided $G$ has unit patience.

\end{theorem}

The second inequality is immediate since \ref{LP:brubach_benchmark} is a tightening of \ref{LP:standard_definition}. To prove the first inequality, we will first proceed to state and prove Lemma \ref{lem:LP_relations_app}.

Assume that for each $u \in U$ and $v \in V$,
we are presented a fractional value, $0 \le x_{u,v} \le 1$.
Moreover, let us assume that the values $(x_{u,v})_{u \in U,v \in V}$ satisfy the following
properties:

\begin{enumerate}

\item For each $u \in U$,
\begin{equation} \label{eqn:offline_matching_constraint}
	\sum_{v \in V}p_{u,v} \, x_{u,v} \le 1.
\end{equation}
\item For each $v \in V$, there exists a probing algorithm $\scr{A}_{v}$
for the instance $G[ \{v\} \cup U]$ which respects commitment
and for which
\begin{equation} \label{eqn:existence_of_probing_algorithm}
\mb{P}[\text{$\scr{A}_{v}$ probes $(u,v)$}] = x_{u,v},	
\end{equation}
for each $u \in U$ \footnote{In the terminology of Section \ref{sec:prelim},
we say that the values $(x_{u,v})_{u \in U, v \in V}$ can be implemented \textbf{losslessly}.}.

\end{enumerate}

In this case, we get the following lemma:

\begin{lemma}\label{lem:LP_relations_app}

If the values $(x_{u,v})_{u \in U, v \in V}$ satisfy properties \eqref{eqn:offline_matching_constraint}
and \eqref{eqn:existence_of_probing_algorithm}, then $(x_{u,v})_{u \in U, v \in V}$ is a feasible
solution to \ref{LP:brubach_benchmark}.

\end{lemma}

\begin{proof}

Let us fix $v \in V$. We first observe that
\[
	\sum_{u \in U} \mb{P}[\text{$\scr{A}_{v}$ probes $(u,v)$}]
\]
corresponds to the expected number of probes that $\scr{A}_v$ makes when
executing on $G[ \{v\} \cup U]$. Thus, since $\scr{A}_v$ makes at most
$\ell_v$ probes, we know that
\[
	\sum_{u \in U} x_{u,v} = \sum_{u \in U} \mb{P}[\text{$\scr{A}_{v}$ probes $(u,v)$}] \le \ell_v.
\]
Let us now denote $\scr{M}$ as the matching returned once $\scr{A}_v$
finishes executing on $G[ \{v\} \cup U]$. This is either a single edge
including $v$, or the empty-set. As such, we denote $\scr{M}(v)$ to indicate
which vertex $v$ is matched to (where $\scr{M}(v) := \emptyset$ if $v$ remains
unmatched).

Observe then that for each $u \in U$,
we have that
\begin{align} \label{eqn:match_probability_fixed_vertex}
	\mb{P}[ \scr{M}(v) = u] &= \mb{P}[\text{$\scr{A}_v$ probes $(u,v)$ and $\st(u,v)=1$}] \\
							  &=  p_{u,v} \, x_{u,v},
\end{align}
as $\scr{A}_v$ respects commitment by assumption.

As a result, since $v$ is matched to at most one vertex of $U$,
\[
	\sum_{u \in U} p_{u,v} \, x_{u,v} \le 1.
\]

It remains to verify that the additional LP constraints present in \ref{LP:brubach_benchmark}
hold for $v$. Let us define $\val( \scr{M})$ as the weight of the edge matched to $v$ (which
is $0$ if $v$ remains unmatched by $\scr{A}_v$). Observe that
\[
	\mb{E}[ \val( \scr{M})] = \sum_{u \in U} w_{u,v} \, p_{u,v} \, x_{u,v},
\]
after applying \eqref{eqn:match_probability_fixed_vertex} and linearity of expectation. Thus,
since $\scr{A}_v$ is a valid probing algorithm which executes on $G[ \{v\} \cup U]$, this
value can be no larger than what is attained by the committal benchmark on $G[ \{v\} \cup U]$. As such,
\[
	\sum_{u \in U} w_{u,v} \, p_{u,v} \, x_{u,v} = \mb{E}[ \val( \scr{M})] \le \OPT(v,U),
\] 
where $\OPT(v,U)$ corresponds to the value of the committal benchmark on $G[ \{v\} \cup U]$. 
More generally, if we now fix $R \subseteq U$, then observe that
\[
	\sum_{u \in R} w_{u,v} \, p_{u,v} \, x_{u,v} = \mb{E}[ \val( \scr{M}) \, \bm{1}_{[\scr{M}(v)  \in R}]],
\]
where $\scr{M}(v) \in R$ corresponds to the event in which the vertex matched to $v$
lies in $R$.

Of course, we can also modify the probing algorithm $\scr{A}_v$ in such a way that it returns $\emptyset$
instead an edge within $\overline{R} \times \{v\}$. This alternative probing algorithm (which depends on $R$) will then
return an edge of expected value
\[
	\mb{E}[ \val( \scr{M}) \, \bm{1}_{[\scr{M}(v)  \in R]}].
\]
As such, for each $R \subseteq U$, 
\[
	\sum_{u \in R} w_{u,v} \, p_{u,v} \, x_{u,v} = \mb{E}[ \val( \scr{M}) \, \bm{1}_{[ \scr{M}(v)  \in R]}] \le \OPT(v,R),
\]
where $\OPT(v,R)$ corresponds to the committal benchmark on $G[\{v\} \cup R]$.

Now, the vertex $v$ was arbitrary, so we know that all the constraints on the vertices of \ref{LP:brubach_benchmark}
hold. By assumption, we also know that for each $u \in U$,
\[
	\sum_{v \in V} p_{u,v} \, x_{u,v} \le 1.
\]
Thus, $(x_{u,v})_{u \in U, v \in V}$ is a feasible solution to \ref{LP:brubach_benchmark},
thereby completing the proof.

\end{proof}

With this lemma, we now prove Theorem \ref{thm:relaxations_relations}.

\begin{proof}[Proof of Theorem \ref{thm:relaxations_relations}]

Suppose that we are presented an optimum solution
to \ref{LP:new}, denoted $(x_{v}(\bm{u}))_{v \in V, \bm{u} \in U^{( \le \ell_v )}}$.
Recall that for each $u \in U, v \in V$ and we defined the edge variable $x_{u,v}$, where
\[
	\til{x}_{u,v}= \sum_{i=1}^{\ell_v} \sum_{\substack{\bm{u}^* \in U^{( \le \ell_{v})}: \\ u_{i}^{*} = u}} \frac{g_{v}^{i}(\bm{u}^*) \, x_{v}(\bm{u}^*)}{p_{u,v}}.
\]
We first observe that the values 
$(\til{x}_{u,v})_{u \in U, v \in V}$ satisfy property \eqref{eqn:offline_matching_constraint}
by assumption (see \eqref{eqn:relaxation_efficiency_matching} of \ref{LP:new}). Moreover, for each fixed $v \in V$, if we consider
the values $(\til{x}_{u,v})_{u \in U}$,
then the \textsc{VertexProbe} algorithm
applied to the input $(G,(x_{v}( \bm{u}))_{\bm{u} \in U^{( \le \ell_v)}}, v)$, satisfies property \eqref{eqn:existence_of_probing_algorithm}, by Lemma \ref{lem:fixed_vertex_probe}.

We may therefore conclude that $(\til{x}_{u,v})_{u \in U, v \in V}$ is a feasible solution to \ref{LP:brubach_benchmark}.
On the other hand,  $(x_{v}(\bm{u}))_{v \in V, \bm{u} \in U^{( \le \ell_v )}}$ is an optimum solution
to \ref{LP:new}, so
\[
	\LPOPT_{new}(G) = \sum_{u \in U, v \in V} w_{u,v} \, p_{u,v} \, \til{x}_{u,v} \le \LPOPT_{DP}(G),
\]
thus proving Theorem \ref{thm:relaxations_relations}.

\end{proof}

We conclude the section by observing the following relation between the LPs in the known i.i.d.
stochastic matching setting, namely \ref{LP:new_iid} and \ref{LP:standard_definition_iid}:

\begin{proposition}
If $(G, \bm{r},n)$ is a known i.i.d. input,
then
\[
	\LPOPT_{new-iid}(G, \bm{r},n) \le \LPOPT_{std-iid}(G, \bm{r},n).
\]
In fact, \ref{LP:new_iid} and \ref{LP:standard_definition_iid} are identical when $G$ has unit patience.
\end{proposition}

This follows via a standard conditioning argument involving the instantiated graph $\hat{G} \sim (G, \bm{r},n)$,
combined with an application of Theorem \ref{thm:relaxations_relations}, so we omit the argument.

\section{Non-adaptive Probing Algorithms and Adaptivity Gaps} \label{appendix:adaptivity_results}
\label{appendix:adaptivity-results}

Suppose that $G=(U,V,E)$ is an arbitrary stochastic graph, 
and we are presented an online probing algorithm $\scr{A}$ which
is non-adaptive and respects commitment (as defined in Section \ref{sec:prelim}). We
may assume that $\scr{A}$ operates in the ROM setting, that is the ordering $\pi$ on
$V$ is chosen uniformly at random, though the definitions we now describe follow
identically when $\pi$ is chosen by an adversary, as well as in the known i.i.d. setting.

We hereby denote $\scr{A}(G)$ as the matching returned by executing $\scr{A}$
on $G$, and $\val(\scr{A}(G))$ as the (random) value of this matching.

With this notation, we define the \textit{adaptivity gap} of a stochastic graph $G$
in the ROM setting as the ratio,
\[
	\frac{\sup_{\scr{B}} \mb{E}[ \val( \scr{B}(G))]}{ \OPT(G)},
\]

where the supremum is over all non-adaptive online probing algorithms.

While all of the algorithms we consider throughout the paper
are implemented non-adaptivity, of particular interest to us are Algorithms \ref{alg:known_stochastic_graph}
and \ref{alg:modified_known_stochastic_graph} in which the stochastic graph $G$ is presented ahead of time.
Observe that Theorems \ref{thm:known_vertex_weights} and \ref{thm:modified_known_stochastic_graph},
imply the following bounds on the relevant
adaptivity gaps:

\begin{corollary}
The known stochastic matching problem with offline vertex weights, arbitrary patience 
and adversarial arrivals has an adaptivity gap no worse than $1-1/e$.
\end{corollary}

\begin{corollary}
The known stochastic matching problem with arbitrary patience, edge weights
and ROM arrivals has an adaptivity gap which is no worse than $1-1/e$.
\end{corollary}

\section{Deferred Proofs} \label{appendix:deferred_proofs}

\begin{proof}[Proof of Theorem \ref{thm:known_ROM_edge_weights}]
In this setting, the order of online vertices $\pi$ is generated uniformly at random. As such, we denote the vertices of $V$ as $v_{1}, \ldots ,v_{n}$,  where $v_{t}$ corresponds to the vertex in position $1 \le t \le n$ of $\pi$ (and $n:=|V|$).

For each $u \in U$ and $v \in V$, we once again make use of the edge variables $(\til{x}_{u,v})_{u \in U, v \in V}$
associated to the solution $(x_{v}(\bm{u}))_{v \in V, \bm{u} \in U^{(\ell_{v})}}$.

Let us now fix a particular vertex $v \in V$, and a vertex $u \in U$. We say that $u$ is \textit{free} for $v$, provided $u$ is unmatched when $v$ is processed by Algorithm \ref{alg:known_stochastic_graph}.

Observe then that if $C(u,v)$ corresponds to the event in which $v$ commits to $u$ during one of its $\ell_v$ probes, then
\begin{align*}
    \mb{P}[ \scr{M}(v)=u ] &= \mb{P}[ C(u,v) \; \text{and $u$ is free for $v$}] \\
    					   &=\mb{P}[ C(u,v)] \cdot \mb{P}[\text{$u$ is free for $v$}] \\
    					   &= p_{u,v} \, \til{x}_{u,v} \, \mb{P}[\text{$u$ is free for $v$}],
\end{align*}

where the final line follows from Lemma \ref{lem:fixed_vertex_probe}.

We know however that,
\begin{align} \label{eqn:conditional_probability_of_freeness}
    \mb{P}[ \text{$u$ is free for $v$}] &= \sum_{t=1}^{n} \mb{P}[ \text{$u$ is free for $v$} \, | \, v_{t}=v] \cdot \mb{P}[ v_{t} = v] \\
    &= \sum_{t=1}^{n} \frac{\mb{P}[ \text{$u$ is free for $v$} \, | \,  v_{t}=v]}{n},
\end{align}

where the last equality follows since $\pi$ is generated uniformly at random.

As such, we may lower bound $\mb{P}[ \text{$u$ is free for $v$} \, |  \, v_{t}=v]$ for each $t =1, \ldots ,n$ in order to derive a lower bound on the competitive ratio of the algorithm.

Let us now fix $1 \le t \le n$ and condition on the event in which $v_{t}=v$. Observe then that
\[
    \mb{P}[\text{$u$ is not free for $v_{t}$} \, | \, v_{t} = v] = \mb{P}[ \cup_{k=1}^{t-1} \scr{M}(v_{k}) = u \, | \, v_{t} = v] \le \sum_{k=1}^{t-1} \mb{P}[ \scr{M}(v_{k}) = u \, | \, v_{t} = v],
\]
as $u$ is not free for $v_t$, if and only if one of $v_{1}, \ldots ,v_{t-1}$ matches to $u$.

On the other hand, using Lemma \ref{lem:fixed_vertex_probe}, we know that for each $k=1, \ldots , t-1$ 

\begin{align*}
    \mb{P}[ \scr{M}(v_{k}) = u \, | \, v_{t} =v] & = \sum_{\substack{s \in V: \\ s \neq v}} \mb{P}[\scr{M}(s) =u\, | \, \{v_{t}=v\} \cap \{v_{k} = s\} ] \cdot \mb{P}[ v_{k}=s \, | \, v_{t}=v]\\
    &\le \sum_{\substack{s \in V: \\ s \neq v}} \mb{P}[C(s,u) \, | \, \{v_{t}=v\} \cap \{v_{k} = s\} ] \cdot \mb{P}[ v_{k}=s \, | \, v_{t}=v] \\
    & =  \sum_{\substack{s \in V: \\ s \neq v}} \frac{ p_{u,s} \, \til{x}_{u,s}}{n-1},
\end{align*}
as once we condition on $\{v_{t}=v\}$, $v_{k}$ is uniformly distributed amongst $V \setminus \{v\}$.

As a result, 
\[
    \mb{P}[\text{$u$ is not free for $v_{t}$} \, | \, v_{t} = v] \le (t-1)\sum_{\substack{s \in V: \\ s \neq v}} \frac{ p_{u,s} \, \til{x}_{u,s}}{n-1} \le \frac{t-1}{n-1},
\]
by the constraints of \ref{LP:new}.

Thus, combined with \eqref{eqn:conditional_probability_of_freeness},
\begin{align*}
    \mb{P}[\text{$u$ is free for $v$}] &\ge \sum_{t=1}^{n} \frac{1}{n} \left(1 - \frac{t-1}{n-1} \right)  \\	
    								   &= 1 - \frac{\sum_{t=1}^{n}(t-1)}{n \, (n-1)}	\\
    								   &=1/2.
\end{align*}
To conclude, for each edge $(u,v) \in E$, we have that
\[
	\mb{P}[ \scr{M}(v)=u] = p_{u,v} \, \til{x}_{u,v} \, \mb{P}[\text{$u$ is free for $v$}] \ge \frac{ p_{u,v} \, \til{x}_{u,v}}{2}.
\]
On the other hand, if we denote $\val(\scr{M})$ as the value of the matching $\scr{M}$, then $\val(\scr{M})= \sum_{u \in U, v \in V} w_{u,v} \, \bm{1}_{[\scr{M}(v)=u]}$. Thus,
\[
	\mb{E}[ \val(\scr{M})] = \sum_{u \in U, v\in V} w_{u,v} \, \mb{P}[\scr{M}(v)=u] \ge \sum_{u \in U, v\in V} \frac{w_{u,v} \, \til{x}_{u,v} \, p_{u,v}}{2}.
\]
As $(x_{v}(\bm{u}))_{v \in V, \bm{u} \in U^{( \le \ell_v)}}$ is an optimum solution to \ref{LP:new}, this completes the proof.
\end{proof}

\begin{proof}[Proof of Lemma \ref{lem:LP_validity_known_iid}]

Suppose that $(G, \bm{r},n)$ is a known i.i.d. instance,
where $G=(U,V,E)$ is a type graph with maximum
patience $\ell:= \max_{v \in V} \ell_v$. We can then
define the following collection of random variables, 
denoted $(X_{t}(\bm{u}))_{t \in [n], \bm{u} \in U^{( \le \ell)}}$,
based on the following randomized procedure:

\begin{itemize}
\item Draw the instantiated graph $\hat{G} \sim (G, \bm{r},n)$,
whose vertex arrivals we denote by $v_{1}, \ldots , v_{n}$.
\item Compute an optimum solution of \ref{LP:new} for $\hat{G}$,
which we denote by $(x_{v_t}(\bm{u}))_{t \in [n], \bm{u} \in U^{(\le v_{t})}}$.
\item For each $t=1, \ldots ,n$ and $\bm{u} \in U^{( \le \ell)}$, 
set $X_{t}(\bm{u}) = x_{v_t}(\bm{u})$ if $\bm{u} \in U^{( \le \ell_{v_t})}$,
otherwise set $X_{t}(\bm{u}) = 0$.
\end{itemize}

Observe then that by definition, $(X_{t}( \bm{u}))_{t \in [n], \bm{u} \in U^{(\le \ell_{v_t})}}$
is a feasible solution to \ref{LP:new} for $\hat{G}$. As such,
for each $t=1, \ldots ,n$
\begin{equation}\label{eqn:online_distribution_iid}
	\sum_{\bm{u} \in U^{ ( \le \ell)}} X_{t}(\bm{u}) \le 1,
\end{equation}
and for each $u \in U$,
\begin{equation} \label{eqn:offline_matching_iid}
	\sum_{t \in [n], i \in [\ell]} \sum_{\substack{\bm{u}^{*} \in U^{(\le \ell )} :\\ u_{i}^{*} = u}} g_{v}^{i}(\bm{u}^{*}) \cdot X_{t}(\bm{u}^{*}) \le 1.
\end{equation}
Moreover, $(X_{t}( \bm{u}))_{t \in [n], \bm{u} \in U^{(\le \ell_{v_t})}}$ is
a optimum solution to \ref{LP:new} for $\hat{G}$, and so Theorem \ref{thm:new_LP_relaxation}
implies that
\begin{equation}\label{eqn:known_iid_benchmark_relaxtion}
	\OPT(\hat{G}) \le \LPOPT_{new}(\hat{G}) = \sum_{t=1}^{n} \sum_{\bm{u} \in U^{(\le \ell)}} \left( \sum_{i=1}^{|\bm{u}|} w_{u_i,v} \cdot g_{v}^{i}(\bm{u}) \right) \cdot X_{t}(\bm{u}). 
\end{equation}

In order to make use of these inequalities in the context of the type graph $G$,
let us first fix a type node $v \in V$ and a tuple $\bm{u} \in U^{( \le \ell)}$. We can then define 
\begin{equation}
	y_{v}(\bm{u}):= \sum_{t=1}^{n} \mb{E}[ X_{t}(\bm{u}) \cdot \bm{1}_{[v_{t} = v]}],
\end{equation}
where the randomness is over the generation of $\hat{G}$. Observe
that by definition of the $(X_{t}(\bm{u}))_{t \in [n], \bm{u} \in U^{(\le \ell)}}$ values, 
\[
	y_{v}(\bm{u}) =0,
\]
provided $|\bm{u}| > \ell_{v}$.

We claim that $(y_{v}(\bm{u}))_{v \in V, \bm{u} \in U^{( \le \ell_v)}}$ is a feasible
solution to \ref{LP:new_iid}. To see this, first observe that if we multiply \eqref{eqn:online_distribution_iid} by the indicator random
variable $\bm{1}_{[v_{t} = v]}$ while summing over $t \in [n]$, then we get that 
\[
	\sum_{\bm{u}^{*} \in U^{ ( \le \ell)}} \sum_{t=1}^{n} X_{t}(\bm{u}) \cdot \bm{1}_{[v_{t} = v]} \le \sum_{t=1}^{n} \bm{1}_{[v_{t} = v]}.
\]
As a result, if we take expectations over this inequality,
\begin{align*}
	\sum_{\bm{u} \in U^{ ( \le \ell)}} y_{v}(\bm{u}) &= \sum_{\bm{u} \in U^{ ( \le \ell)}} \mb{E}\left[ \sum_{t=1}^{n} X_{t}(\bm{u}) \cdot \bm{1}_{[v_{t} = v]}\right] \\ 
	&\le \sum_{t=1}^{n} \mb{P}[v_t =v] \\
	&= r_{v},
\end{align*}
for each $v \in V$.

Let us now fix $u \in U$. Observe that by rearranging the left-hand side of \eqref{eqn:offline_matching_iid},
\begin{equation} \label{eqn:rearrangment}
	\sum_{t \in [n]} \sum_{ i \in [\ell]} \sum_{\substack{\bm{u}^{*} \in U^{(\le \ell) } :\\ u_{i}^{*} = u}} g_{v}^{i}(\bm{u}^{*}) \cdot X_{t}(\bm{u}^{*}) = \sum_{i \in [\ell]} \sum_{\substack{\bm{u}^{*} \in U^{(\le \ell)}: \\ u_{i}^{*} = u}} g_{v}^{i}(\bm{u}) \cdot\sum_{v \in V} \sum_{t \in [n]} X_{t}(\bm{u}^*) \cdot \bm{1}_{[v_{t} = v]}.
\end{equation}
Thus, after taking expectation over \eqref{eqn:offline_matching_iid},
\[
	 \sum_{v \in V} \sum_{i \in [\ell]} \sum_{\substack{\bm{u}^{*} \in U^{(\le \ell_{v})}: \\ u_{i}^{*} = u}} y_{v}(\bm{u}^{*}) \le 1,
\]
for each $u \in U$.

Since $(y_{v}(\bm{u}))_{v \in V, \bm{u} \in U^{(\le \ell_v)}}$ satisfies these inequalities,
and the variables are clearly all non-negative, 
it follows that $(y_{v}(\bm{u}))_{v \in V, \bm{u} \in U^{(\le \ell_v)}}$ is a feasible solution to \ref{LP:new_iid}.

In order to complete the proof, let us rearrange the right-hand side of \eqref{eqn:known_iid_benchmark_relaxtion}
as in \eqref{eqn:rearrangment} and take expectations. We then get that
\begin{align*}
	\mb{E}[ \OPT(\hat{G})] &\le \sum_{v \in V} \sum_{\bm{u} \in U^{(\le \ell)}} \left( \sum_{i=1}^{|\bm{u}|} w_{u_i,v} \cdot g_{v}^{i}(\bm{u}) \right) \cdot y_{v}(\bm{u}).
\end{align*}
Now, $\OPT(G, \bm{r},n) = \mb{E}[ \OPT(\hat{G})]$ by definition, so since $(y_{v}(\bm{u}))_{v \in V, \bm{u} \in U^{(\le \ell_v)}}$
is feasible, it holds that
\[
	\OPT(G,\bm{r},n) \le \LPOPT_{new-iid}(G,\bm{r},n),
\]
thus completing the proof.

\end{proof}

\newpage

\section{Extended Related Works} \label{appendix:extended_related_works}
\label{app:related-work}
Our results pertain to the online stochastic matching problem which (loosely speaking) is online bipartite matching where edges are associated with their probabilities of existence. There is a substantial body of research pertaining to the ``classical'' (i.e. non stochastic) online bipartite model in the fully adversarial online model, the random order model, and the i.i.d. input model. The ever growing interest in various online bipartite matching problems is a reflection of the importance of online advertising but there are many other natural applications. The literature concerning competitive analysis\footnote{Initially, competitive analysis refered to the relative performance (i.e., the competitive ratio) of an online algorithm as compared to an optimal solution (in the worst case over all input sequences determined adversarially). We extend the meaning of the competitive ratio to also refer to input sequences generated in the ROM model as well as sequences generated i.i.d. from a known or unknown distribution; that is, whenever the algorithm has no control over the order of input arrivals.}  of online bipartite matching is too extensive to do justice to many important papers. We refer the reader to the excellent 2013 survey by Mehta \cite{Mehta13} with emphasis on online variants relating to ad-allocation. Given the continuing interest in ad-allocation, the survey is not current but does describe the basic results.

The seminal result for unweighted  online bipartite matching is due to   Karp, Vazirani, and Vazirani \cite{KarpVV90}. They gave the randomized Ranking algorithm that achieves competitive ratio $1 - 1/e$ in the adversarial online setting which they show is the best possible ratio for any randomized algorithm. There have been many proofs of this seminal result, such as the primal-dual approach due to Devanur et al. \cite{DJK2013}. Any greedy algorithm (i.e., one that  always makes a match when possible) has a $0.5$ ratio, and this is the best possible a deterministic algorithm can attain. The Ranking algorithm can also be viewed as a deterministic algorithm in the ROM input model. In the ROM model, Madhian and Yan \cite{Mahdian2011} show that the randomized Ranking algorithm achieves competitive ratio $.696$. For the case of weighted offline vertices and adversarial input sequences, Aggarwal et al. \cite{AggarwalGKM11} were able to achieve a randomized $1-1/e$ competitive ratio by their Perturbed Ranking algorithm.  Huang et al. \cite{huang2018online} show that the Perturbed Ranking algorithm obtains a $.6534$ competitive ratio in the ROM input model. 

Feldman et al. \cite{FeldmanMMM09} introduced online bipartite matching in the i.i.d. model in which each online vertex is independently and identically generated from some known distribution. In this model, they were able to beat the  $1-1/e$ inapproximation for bipartite matching that applies to the fully adversarial online  model.  The i.i.d. online bipartite model has been studied for the unweighted and edge weighted models. The most recent competitive ratios  for integral arrival rates are due to Brubach et al. \cite{BrubachSSX16} in which they derive a $.7299$ ratio for 
the (offline) vertex weighted case and a $.705$ ratio for edge weighted graphs. 
Karande et al. \cite{KarandeMT11} show that any competitive ratio for the ROM model applies to the unknown (and therefore known) i.i.d. models. It follows that any inapproximation for the known i.i.d. model applies  to the ROM model. Kesselheim et al. \cite{KRTV2013} extend the classical secretary result and established 
the optimal $1/e$ ROM ratio for bipartite matching with edge weights.

An early example of stochastic probing without commitment is the Pandora's box problem attributed to Weitzman \cite{Weitzman1979}. 
In Weitzman's  Pandora's box problem, a set of boxes is given, where each box contains a stochastic value from a known distribution and a cost for opening (i.e., probing) the box. The algorithm has the option at any time of accepting the value of any opened box and pays the total cost of all opened boxes. This is an offline probing problem in that
boxes can be opened in any order. An online version of the Pandora's box problem has recently been studied in Esfandiari et al. \cite{EsfandiariHLM19}.       
Stochastic probing with commitment has been studied for  various packing problems,  most notably for the knapsack problem, as studied in Dean et al.
\cite{DeanGV05,DeanGV08}. In the stochastic knapsack setting, the stochastic inputs are items whose values are known but whose sizes are stochastic and not known until the algorithm probes the item. As soon as the knapsack capacity is exceeded by a probed item, the algorithm terminates.  Dean et al. also introduced the offline issue of measuring the benefit of adaptively choosing probes versus
having a fixed order of probes.

Turning back to matching problems, 
Chen et al. \cite{Chen} introduced the stochastic matching problem assuming a known stochastic graph and algorithms that can probe any edge in any order. They obtained a $4$-approximation\footnote{Unfortunately, approximation and competitive bounds for maximization problems are sometimes represented both as ratios $> 1$ and as fractions $<1$. We shall report these ratios as stated in the relevant papers. Our results will be stated as fractions.} 
greedy algorithm in the unweighted case for arbitrary patience values. They conjectured that their greedy algorithm was a $2$-approximation. 
Subsequently, 
Adamczyk \cite{Adamczyk11} 
confirmed that the greedy algorithm is a $2$-approximation for the unweighted problem and that this approximation is tight. Bansal et al. \cite{BansalGLMNR12} established a $4$-approximation for the edge weighted case with arbitrary patience and a $3$-approximation for the special case of bipartite graphs. Adamczyk et al. 
\cite{Adamczyk15} improved the Bansal et al. bounds  providing 
  an approximation algorithm with a ratio of $2.845$ for bipartite graphs and an algorithm with a ratio of $3.709$ for general graphs. Baveja et al. \cite{BavejaBCNSX18} recently improved the analysis of the original algorithm of Bansal et al., yielding an approximation ratio of $3.224$ for general graphs.

Of particular importance to our paper is the known stochastic matching framework with ROM arrivals, as defined precisely
in Section \ref{sec:prelim}.  
Gamlath et al. \cite{Gamlath2019} presented 
a probing algorithm which is a $1 - \frac{1}{e}$-approximation for the bipartite case
in the \textit{full patience setting}; that is, when there are no patience restrictions for nodes on either side of the bipartition. The full patience setting is closely related to the bipartite matching algorithm studied by Ehsani et al. \cite{Ehsani2017}, which they prove is a $1-\frac{1}{e}$-approximation as a corollary of their work in the more general \textit{combinatorial auctions prophet secretary problem}.
While not explicitly stated in \cite{Ehsani2017}, their bipartite matching algorithm
can be interpreted as an adaptive probing algorithm in the known stochastic matching framework with ROM arrivals,
attaining the same $1-\frac{1}{e}$ non-adaptive approximation ratio as Gamlath et al..
Very recently, Tang et al. \cite{Tang2020} provided an alternative algorithm also attaining the same approximation ratio
of $1 - \frac{1}{e}$ in the more general \textit{oblivious bipartite matching} setting, however their algorithm
does not execute in an online fashion, and so is incomparable.
See also Tang et al. \cite{TangWZ2020} for an online  greedy algorithm achieving a .501 ratio for a known stochastic graph with edge weights.    
   
Mehta and Panigrahi \cite{MehtaP12} adapted the stochastic matching problem to the online setting problem with unit patience where the stochastic graph is not known to the algorithm. They specifically considered the unweighted case for unit patience (for the online nodes) and uniform edge probabilities (i.e,, for every edge $e$, $p_e  = p$  for some fixed probability $p$). They showed that every greedy algorithm has competitive ratio $\frac{1}{2}$. In the same online setting, they provided a greedy algorithm that achieves competitive ratio $\frac{1}{2}(1 + (1-p)^{2/p})$ which limits to $\frac{1}{2}(1+e^{-2}) \approx .567$ as $p \rightarrow 0$. They also show that against a ``standard linear programming (LP)'' benchmark, that the best possible ratio is $.621 < 1-\frac{1}{e}$.  However, this does not preclude a $1-\frac{1}{e}$ competitive ratio for a stricter LP bound on an optimal stochastic probing algorithm. Preceding the Mehta and Panigrahi work is a result in Bansal et al. \cite{BansalGLMNR12} where they consider a known stochastic (type) graph with a distribution on the online nodes. This can be called the stochastic matching problem with known i.i.d. inputs. Bansal et al. achieve a $7.92$ competitive ratio (or approximately, $.13$ as a fraction) in this stochastic i.i.d. model. This was improved to $.24$ by Adamczyk \cite{Adamczyk15} and most recently, by  Brubach et al. \cite{BrubachSSX20} where they obtain a .46 competitive ratio and a  $1-\frac{1}{e}$ inapproximation against a standard LP.

Returning to the unknown stochastic graph setting, there are recent independent papers by Goyal and Udwani \cite{Goyal2020} and Brubach et al. \cite{Brubach2019}. Goyal and Udwani consider the vertex weighted  unit patience problem and establish a (best possible) $1-\frac{1}{e}$ competitive ratio against an LP that acts as an upper bound on the committal benchmark
under the assumption that the edge probabilities are decomposable (i.e., $p_{u,v} = p_u \cdot p_v$) and a $.596$ competitive ratio for vanishingly small edge probabilities. Our paper is motivated by and most closely follows the Brubach et al. \cite{Brubach2019} paper. Brubach et al. use and motivate the ``ideal stochastic  benchmark'' (for arbitrary patience) and an LP relaxation for that ideal benchmark. They establish a best possible  deterministic $\frac{1}{2}$ competitive ratio against their LP for the vertex weighted online stochastic matching problem. In a recent paper, Huang and Zhang \cite{huang2020online} provide a randomized algorithm for unit patience and offline vertex weights in the online stochastic matching framework. In the limit as edge probabilities decrease, their algorithm achieves a $.572$ competitive ratio.

%

\end{document}